\documentclass[acmtog,screen,nonacm]{acmart}

\usepackage{booktabs} 
\usepackage{soul}
\usepackage{subcaption}

\usepackage{tabularx}
\usepackage{caption}
\usepackage{tikz}
\usetikzlibrary{calc,arrows.meta}
\usepackage{placeins}
\usepackage{float}
\usepackage{wrapfig}
\usepackage{comment}

\AtBeginDocument{%
  \setlength{\abovedisplayskip}{6pt plus 2pt minus 2pt}
  \setlength{\belowdisplayskip}{6pt plus 2pt minus 2pt}
  \setlength{\abovedisplayshortskip}{2pt plus 1pt}
  \setlength{\belowdisplayshortskip}{4pt plus 2pt minus 2pt}
}

\allowdisplaybreaks

\usepackage[ruled,vlined]{algorithm2e} 

\SetAlFnt{\small}
\SetAlCapFnt{\small}
\SetAlCapNameFnt{\small}
\SetAlCapHSkip{0pt}

\usepackage{hyperref}

\usepackage[labelfont=bf,format=plain]{caption}
\newcommand{\commentText}[1]{#1}
\newcommand{\todo}[1]{\commentText{{\color{red}[\textbf{\textsc{TODO}}: \textit{#1}]}}}

\newcommand{\add}[1]{#1}
\newcommand{\del}[1]{}
\newcommand{\madd}[1]{#1}
\newcommand{\mdel}[1]{}

\usepackage{xcolor}
\definecolor{addcolor}{RGB}{0, 128, 0}
\definecolor{pink}{RGB}{227, 119, 194}
\definecolor{orange}{RGB}{255, 127, 14}
\definecolor{blue}{RGB}{31, 119, 180}
\definecolor{olive}{RGB}{188, 189, 34}

\renewcommand{\vec}[1]{\mathbf{#1}}

\newcommand{\MM}{\mathbf{M}}

\acmJournal{TOG}

\setcopyright{none}

\begin{document}
\title{Divide and Truncate: A Penetration and Inversion Free Framework for Coupled Multi-physics Systems}

\author{Anka H. Chen}
\orcid{0000-0002-5819-3453}
\affiliation{%
  \institution{NVIDIA}
  \city{Kirkland}
  \country{United States of America}}
\email{ankachan92@gmail.com}

\author{Jerry Hsu}
\orcid{0000-0003-2333-0224}
\affiliation{%
  \institution{ZeroMatter}
  \city{Mountain View}
  \country{United States of America}}
\email{jerry.hsu.research@gmail.com}

\author{Youssef Ayman}
\orcid{0009-0005-5149-8261}
\affiliation{%
  \institution{The American University in Cairo}
  \city{Cairo}
  \country{Egypt}}
\email{youssefayman55200@gmail.com}

\author{Miles Macklin}
\orcid{0000-0003-3954-8009}
\affiliation{%
  \institution{NVIDIA}
  \city{Auckland}
  \country{New Zealand}}
\email{mmacklin@nvidia.com}

\renewcommand\shortauthors{Chen, A.H. et al}

\begin{abstract}
We present \emph{Divide and Truncate} (DAT), a unified framework for coupling multi-physics systems through penetration-free collision handling, including rigid bodies, volumetric soft bodies, thin shells, rods, and animated objects. By partitioning the ambient space into exclusive regions and truncating displacements to remain within them, DAT guarantees penetration-free contact resolution. Our \emph{Planar-DAT} variant further refines this by restricting only motion toward nearby surfaces, leaving tangential movement unconstrained, which addresses the artificial damping and deadlock problems of previous works. The framework is material-agnostic: each object responds to contact without knowledge of the opposing body's physics. Our method is also solver-agnostic; it can be integrated seamlessly with any iterative optimizer as a post-processing step, enabling robust simulation of complex multi-body interactions.
\end{abstract}

%
%

\ccsdesc[500]{Computing methodologies~Physical simulation}
\ccsdesc[300]{Computing methodologies~Collision detection}

%
%

\begin{teaserfigure}
\captionsetup{skip=5pt}
\vspace{-6em}
\includegraphics[width=0.49\linewidth,trim=450 190 300 300,clip]{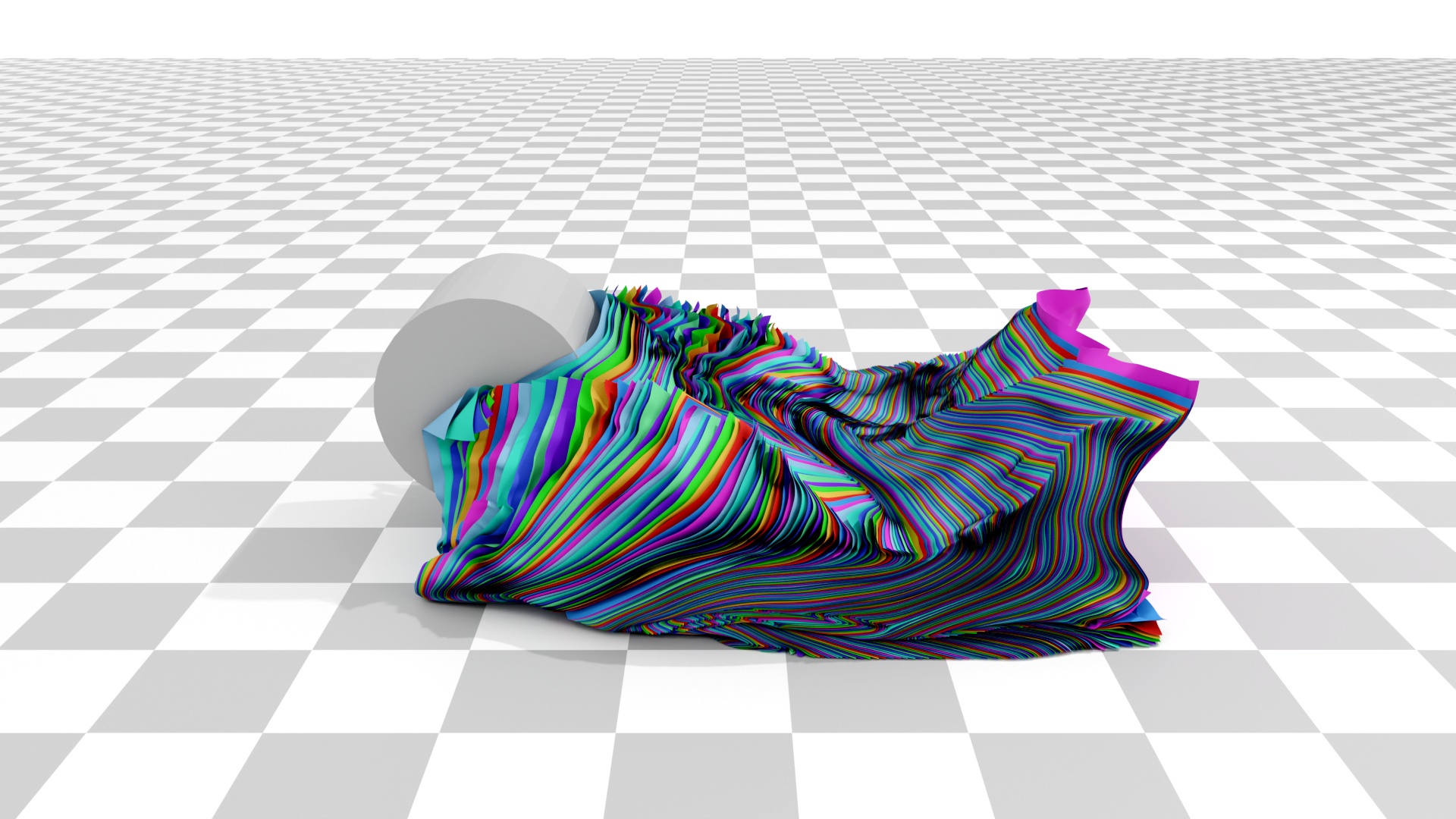}\hspace{-10px}
\raisebox{-5em}{%
  \includegraphics[
    width=0.49\linewidth,
    trim=490 0 350 0,
    clip,
    angle=10
  ]{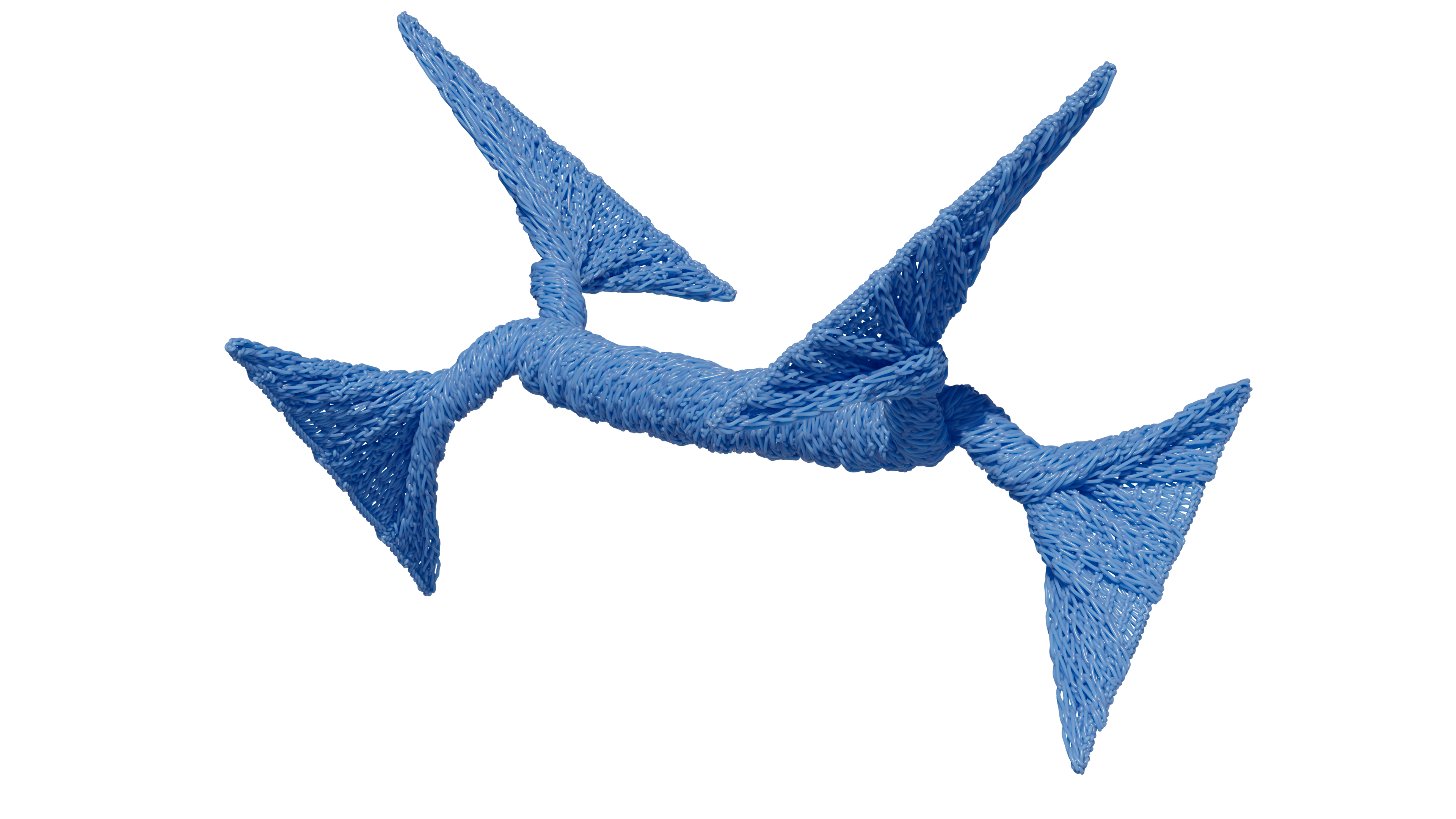}
}\\[-5em]
\includegraphics[width=0.49\linewidth,trim=100 420 200 450,clip]{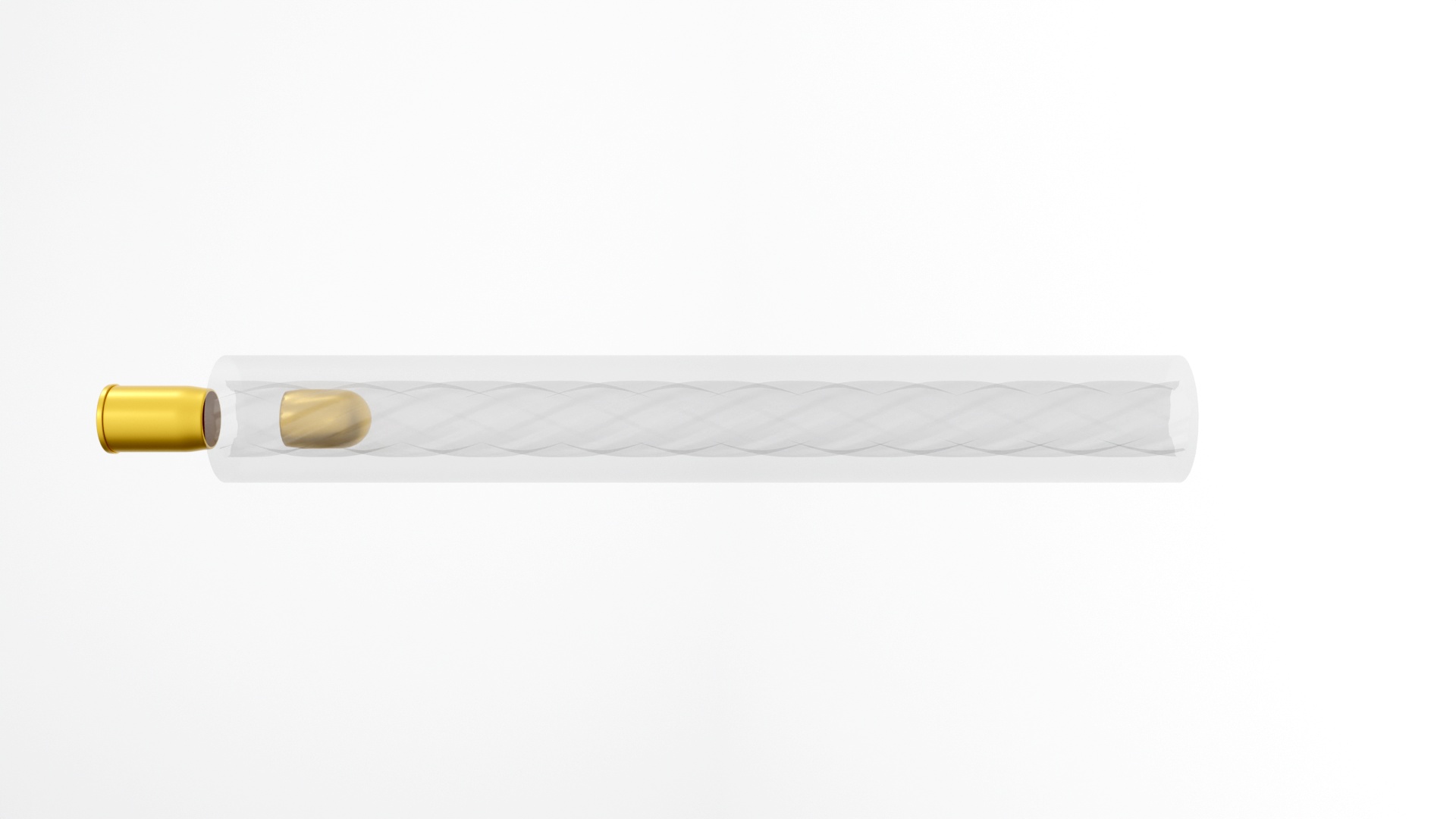}
\includegraphics[width=0.49\linewidth,trim=100 420 200 450,clip]{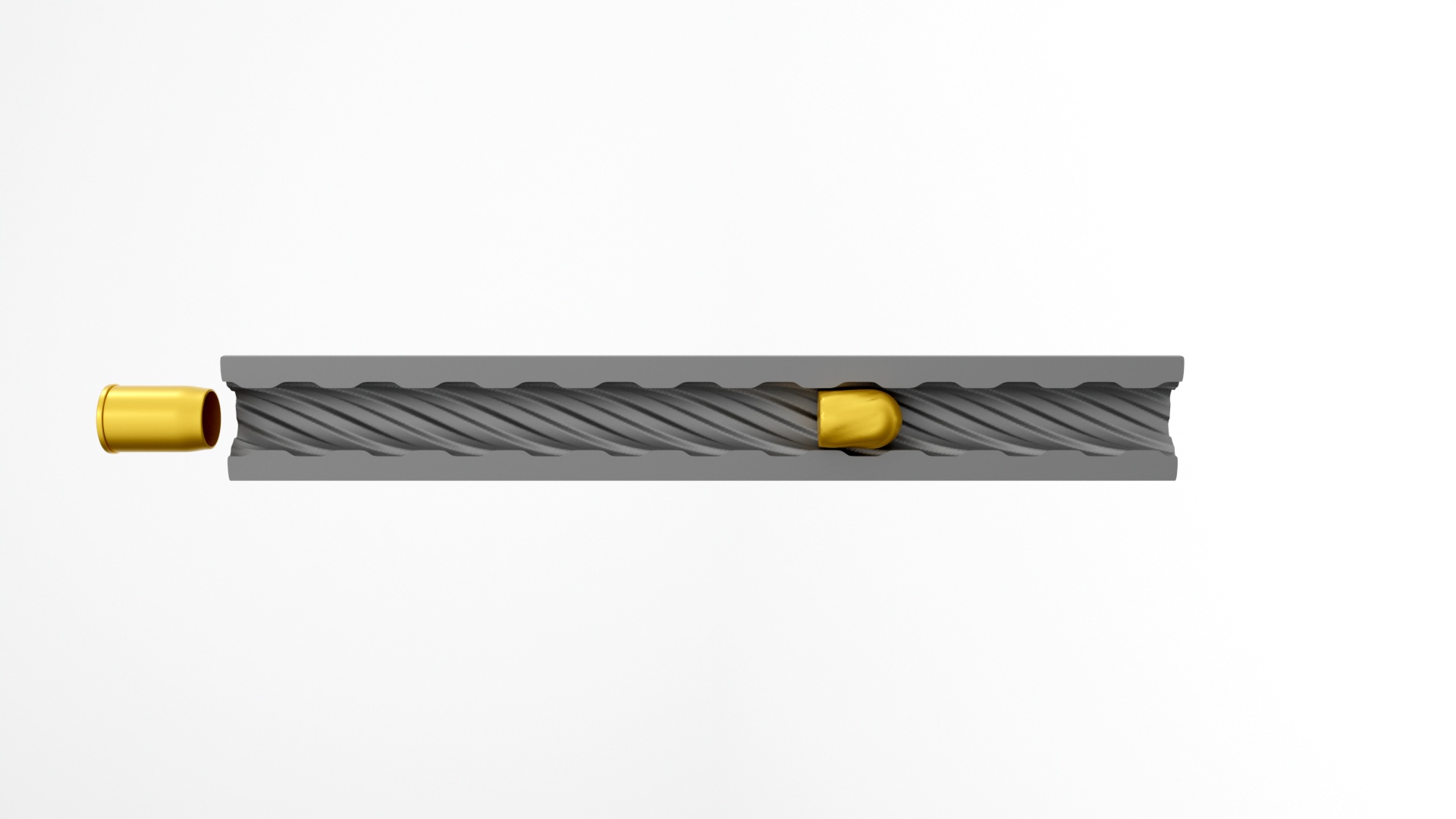}\hfill\\[-1.5em]
\caption{Our \emph{Planar Divide and Truncate} method enables efficient, penetration-free simulation of diverse multi-physics scenarios.
\textbf{Top Left:} 200 layers of cloth (4M vertices, 7.84M triangles) dropping onto a cylinder, creating more than 20 million collisions.
\textbf{Top Right:} 130k segments of yarn intertwined to form two pieces of cloth, which are then further twisted and tied around one another. 
\textbf{Bottom:} A lead bullet deforming through a rifled barrel at more than $370\,\mathrm{m/s}$, shown with transparent barrel (left) and cross-section view (right).
}
\Description{Teaser figure showing simulation results: 200 layers of cloth dropping (top left), a double-twist DAT example (top right), and a bullet through a rifled barrel with transparent and cross-section views (bottom row).}
\label{fig:teaser}
\end{teaserfigure}

\maketitle

\section{Introduction}

Penetration-free techniques have become a cornerstone of robust physics-based animation. In applications ranging from robotics \cite{du2024embedded, zhu2024efficient} to virtual garment design \cite{Liu:2024:ADG} and geometry processing \cite{fang2021guaranteed}, maintaining strict non-intersection guarantees is essential for both physical plausibility and downstream processing. However, achieving penetration-free guarantees efficiently across diverse physical systems remains a significant challenge.

The dominant paradigm employs continuous collision detection (CCD) embedded directly within a line-search framework, as in Incremental Potential Contact (IPC)~\cite{li2020incremental}. This strategy tightly couples collision handling with the solver: CCD must be executed at every solver iteration, and the entire mesh's motion can be globally restricted by the earliest detected collision. Moreover, when multiple materials or physical models are coupled—such as cloth interacting with rigid bodies or codimensional objects—each requires its own specialized CCD formulation and tuning~\cite{li_codimensional_2021, ferguson2021intersection, lan2022affine}, complicating the integration of heterogeneous multi-physics systems.

An alternative family of trust-region methods bounds each vertex's motion using discrete collision detection. Methods such as conservative advancement~\cite{zhang2006interactive, tang2009c, wu2020safe, chen2025offset} limit displacement to half the distance to nearby primitives, preventing penetration without per-iteration CCD. However, these isotropic bounds restrict motion in \emph{all} directions—not just toward collisions—introducing artificial damping and causing deadlock in dense contact.

Both paradigms share a common limitation: they guard against worst-case motions. Line-search methods stop the global solver step at the earliest detected collision, regardless of whether most regions could safely advance further. Trust-region methods assume motion in any direction is potentially dangerous and thus impose isotropic bounds. Yet in practice, the displacement direction is likely known before constraints are enforced, and most displacements move tangentially or away from nearby surfaces rather than directly toward them due to the repulsive force. Those displacements should not be truncated at all.

We present \emph{Divide and Truncate} (DAT), a universal framework for penetration-free simulation that exploits displacement direction to achieve both efficiency and accuracy. Our key contributions are:
\begin{itemize}
    \item \textbf{Unified Formulation.} We reformulate previous trust-region methods as \emph{Isotropic Divide and Truncate}, a specific instance of a general DAT framework, revealing the limitations of isotropic approaches.
    
    \item \textbf{Direction-aware truncation.} We propose \emph{Planar-DAT}, which uses direction-aware half-space divisions instead of isotropic spherical bounds. By restricting only motion \emph{toward} surfaces, Planar-DAT eliminates artificial damping and deadlock while enabling natural dynamics in dense contact.
    
    \item \textbf{Material and Solver Agnostic.} The same truncation scheme works universally across cloth, soft bodies, rods, rigid bodies, and animated objects. Each object can respond to contact without knowledge of the other's material or physical model, enabling straightforward coupling of heterogeneous systems. As a post-processing step decoupled from the optimizer, our method integrates seamlessly with any iterative solver.
\end{itemize}

We demonstrate Planar-DAT on challenging scenarios where it prevails over previous methods, including 200 layers of cloth generating over 20 million simultaneous contact pairs, simulated at interactive rates with strict penetration-free guarantees (\autoref{fig:teaser}). Where CCD-based line search is over an order of magnitude slower per iteration, and Isotropic-DAT locks up or introduces crippling damping, Planar-DAT delivers both speed and natural dynamics---preserving over twice the intended motion without sacrificing safety.
\section{Related Works}

As optimizers struggle to capture the discontinuous or asymptotic nature of hard collisions, additional enforcement mechanisms are often used to strictly guarantee penetration-free. In physics-based simulations, these enforcement mechanisms typically come in two flavors: line-search and trust-region methods. 

\subsection{Line-Search-Based Methods}

As optimizers commonly construct their systems of equations based on the current configuration, they often fail to capture distant contacts. Even when contacts are properly registered, the local evaluation can often produce an insufficient signal to guarantee penetration-free. To address this issue, the line-search family of methods like Incremental Potential Contacts (IPC) \cite{li2020incremental} searches for the time of first impact (TOI) along a descent direction. By limiting the step size to always be smaller than the TOI, IPC guarantees that contacts can draw closer without penetration. This gives the local approximation a chance to be reevaluated more accurately closer to the asymptotic contact barrier and to be ultimately resolved. 

Various extensions built on the same line-search scheme have been proposed. \citet{li_codimensional_2021} extended IPC to handle co-dimensional objects and strain limiting. \citet{lan2021Medial} derived closed-form distance computations for barriers defined on medial geometry for reduced subspace simulations. \citet{ferguson2021intersection} use approximately linear subdivisions to allow for rigidbodies under rotational motion. \citet{lan2022affine} proposed the use of a generalized affine velocity for almost-rigidbody motion. \citet{lan2022penetration} combined the line-search strategy with projective dynamics to speed up elasticity. \citet{lan2023Stencil} employed a coordinate descent scheme with CCD to accelerate convergence and parallelism\add{, and introduced a Local CCD technique that only run collision filtering for existing collision pairs from the most recent regular CCD, mitigating TOI locking. However,  it does not provide strict penetration-free guarantees on its own---new collision pairs arising from locally modified displacements (from the local collision filtering) are not checked, so a global CCD pass remains necessary as a safeguard}. \citet{huang2025} formulated a discretization independent barrier to remove spurious forces while retaining penetration-free. 

\subsection{Trust-Region-Based Methods}
Trust-region methods establish a region within which the optimization domain is well-behaved---singularity-free, well-approximated quadratically, or of positive curvature. Since trust regions are easy to construct conservatively and simple to enforce, they offer a cheap and efficient way to ensure convergence. Line-search methods can be seen as a special case where the trust region is a line segment along the descent direction; general trust regions are more flexible as they need not be restricted to a line.

Trust-region methods are well-suited for constrained optimization where constraint satisfaction is critical~\cite{DBLP:journals/mp/Yuan15}. Constraints can be incorporated directly into the trust region formulation~\cite{burke1992robust, more1983recent}, typically linear and convex~\cite{conn1988global, burke1990convergence}.

In graphics, this was previously leveraged to safely approximate contact singularities \cite{sassen2024}. Even without a penetration-free guarantee on the whole trajectory, trust regions can still act as an effective mechanism for convergence enforcement. However, trust-region methods for fully penetration-free trajectories have not been extensively explored in simulation. Unlike IPC, which combines CCD with line search, trust-region methods use \emph{discrete collision detection} (DCD) to define per-vertex (or per-rigid-body) trust regions that constrain movements to prevent penetration.

This idea was initially explored for rigid body dynamics, termed Conservative Advancement. \citet{zhang2006interactive} uses extremal vertex queries to find motion bounds for objects with constant translational and rotational velocities; \citet{tang2009c} extends this to triangulated models without assumptions on geometry or topology.
For cloth simulation, \citet{Wu:2020:SFR} identified vertex displacement constraints to prevent self-intersection. \citet{Wang:2023:FGB} and \citet{chen2025offset} utilized these constraints within step-and-project frameworks for fast and realistic simulation.

Our work follows this trust-region paradigm but addresses the fundamental limitation shared by all prior methods: their reliance on isotropic bounds that ignore displacement direction. By reformulating the trust region as direction-aware half-spaces, Planar-DAT restricts only motion toward nearby surfaces, avoiding the artificial damping and deadlock that can arise with existing approaches while retaining their efficiency and solver independence.

\section{Method}

While our method works with any convex primitive (ex. edge, triangle, or convex meshes), without loss of generality, we introduce our method using a triangular mesh. 

Let us denote a triangular mesh as $M=\{\mathcal{V}, \mathcal{E}, \mathcal{T}\}$, where $\mathcal{V}, \mathcal{E}$, and $\mathcal{T}$ are the sets of its vertices, edges, and triangles, respectively. We collectively refer to these elements as faces.
Given a penetration-free state, $X\in\mathbb{R}^{3\times N }$, which is the stacked coordinates of all vertices, a deformation can be represented by $\Delta X\in\mathbb{R}^{3\times N }$, such that the deformed configuration is $\Delta X + X$. $\Delta X$ may come from an optimizer that minimizes an objective function defined on $M$, e.g., a simulator. While $\Delta X$ is expected to decrease the objective function, it may introduce intersections in $\Delta X + X$, which is generally undesirable. Therefore, we seek to modify $\Delta X$ while preserving its direction and magnitude as much as possible, so that it remains effective for reducing the objective. In fact, we recover the CCD-aware line search strategy used in IPC \cite{li2020incremental} if we use a single scalar $t\in [0, 1]$ to scale $\Delta X$.

\subsection{Cage the Beast}
We propose an alternative scheme that offers greater flexibility in modifying $\Delta X$. The key idea is to allocate an \textit{exclusive region} to each primitive $a\in M$ such that it can deform only within its assigned region. Given this assumption, the resulting deformation would then be guaranteed to remain penetration-free. Conceptually, this is akin to caging each primitive so that it cannot “escape” and intersect others. Our scheme consists of two steps: (i) divide the ambient space into mutually exclusive regions and assign each region to a primitive; and (ii) truncate the deformation so that each primitive’s deformation remains within its dedicated region. We therefore call the method Divide and Truncate (DAT).

This principle can be applied to various primitive types such as splines, triangles, tetrahedra, or even rigidbodies. With triangle meshes in particular, we only need to ensure that there is no vertex-triangle penetration or edge-edge penetration throughout the deformation. Accordingly, the divide-and-truncate procedure is applied separately to all vertex–triangle pairs and all edge–edge pairs. 

We begin in \autoref{sec:iso_div} by introducing \emph{Isotropic Divide and Truncate}, which is an existing form of DAT analogous to that used in \cite{wu2020safe} and \add{Offset Geometric Contact (}OGC\add{)}~\cite{chen2025offset}. We then introduce the full version of our \emph{planar DAT} in \autoref{sec:pdat} before addressing its extension into inversion prevention (\autoref{sec:inversion}), rotational degrees of freedom (\autoref{sec:rotational}), and animated degrees of freedom (\autoref{sec:animated}). 

\subsection{Isotropic Division}
\label{sec:iso_div}

One option for dividing the space is to make sure that a face never moves more than half of the distance to the closest face it might penetrate; then, even in the worst case where two faces move directly toward each other, penetration cannot happen. Let us start by considering the example of a vertex-triangle pair and denote a vertex $v$'s distance to its closest non-neighbor triangle by $d_{\text{min},v}$. We can define a vertex $v$'s exclusive region as the ball centered at $v$ with a radius $r_v < 0.5\;d_{\text{min},v}$ :
\begin{equation}
    \label{Eq:vertex_isometric_exclusive_space}
    \mathcal{S}_v = \{\vec{x}\in\mathbb{R}^3 \mid |\vec{x} - \vec{x}(v)| < r_v\}   
\end{equation} 
Similarly, denote a triangle $t$'s distance to its closest non-member vertex by $d_{\text{min},t}$. We can obtain its exclusive region by offsetting it with a radius $r_t < 0.5\; d_{\text{min},t}$: 
\begin{equation}
    \label{Eq:triangle_isometric_exclusive_space}
    \mathcal{S}_t = \{\vec{x}\in\mathbb{R}^3\mid  dis(\vec{x}, t) < r_t\}   
\end{equation} 

We name this way of dividing space as Isotropic Division, because it expands each face uniformly in each direction---treating all directions equally. Similarly, the exclusive region for each edge is:
\begin{equation}
    \label{Eq:edge_isometric_exclusive_space}
    \mathcal{S}_e = \{\vec{x}\in\mathbb{R}^3\mid  dis(\vec{x}, e) < r_e\}   
\end{equation} 

By construction, $S_v$ does not overlap with any non-neighbor triangle’s region, and $S_e$ does not overlap with any non-neighbor edge’s region, completing the division step. Truncation can then be performed by assigning each vertex an individual scaling parameter $t_v\in[0,1]$ such that:
\begin{equation}
    |t_v \Delta \vec{x}_v\mid  \leq \min(r_v, \min_{t\in \mathcal{T}_v}(r_t),  \min_{\mdel{t}\madd{e}\in \mathcal{E}_v}(r_e) ),
    \label{Eq:conservativeBound}
\end{equation}
\add{where $\mathcal{T}_v$ and $\mathcal{E}_v$ denote the incident triangles and edges of $v$ (see \autoref{thm:vertex_level_eq}).}

This formulation gives rise to the penetration-free scheme employed by \cite{wu2020safe} and OGC~\cite{chen2025offset}. We refer to this method as \textit{Isotropic Divide and Truncate}, or Isotropic-DAT.
\begin{figure}[t]
\centering
\newcommand{\subfigvtdatt}[3]{%
\begin{subfigure}{0.42\linewidth}
\centering
\fbox{\includegraphics[width=0.9\linewidth,trim=#2,clip]{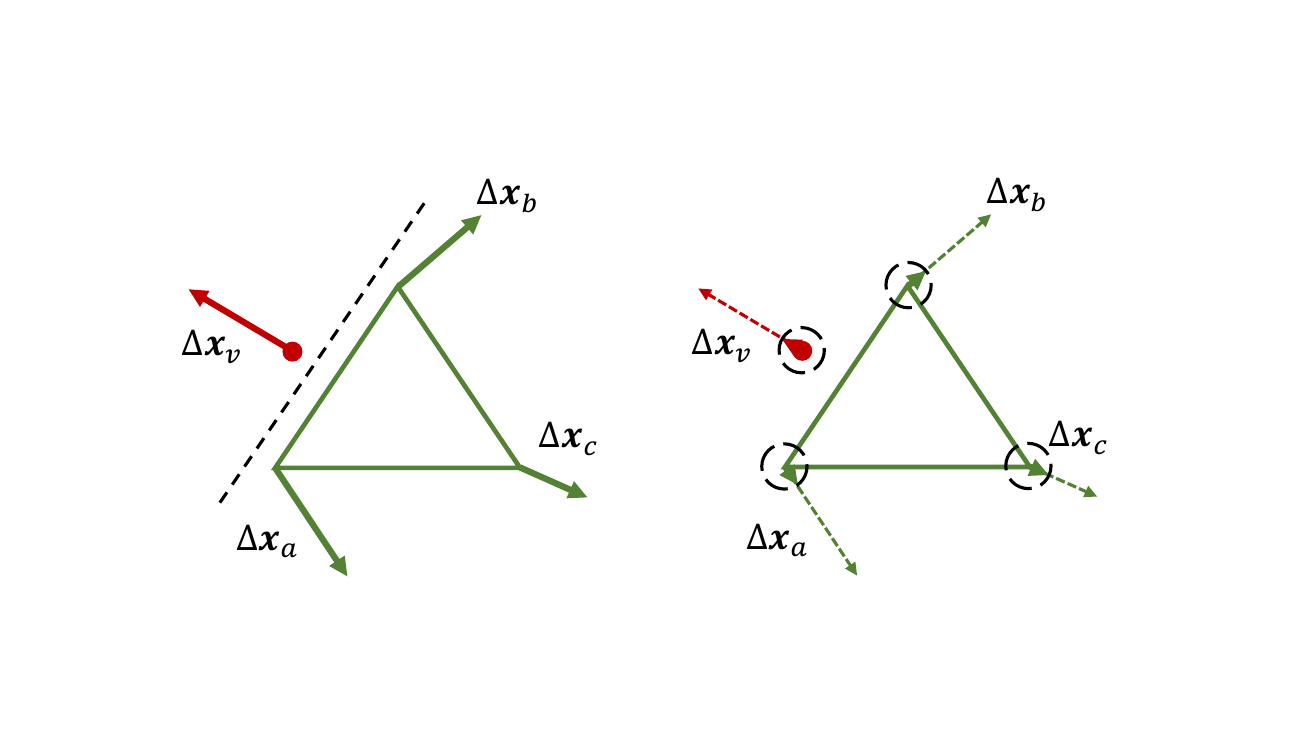}}
\caption{#3}
\label{fig:dat_comparison:#1}
\end{subfigure}}

\subfigvtdatt{a}{50 75 300 75}{Planar Truncation}%
\hspace{1.5em}%
\subfigvtdatt{b}{300 75 50 75}{Isotropic Truncation}%
\\[5pt]
\caption{Displacement truncation for a vertex-triangle pair. Vertex $v$ and triangle $t=(a, b, c)$ are close but moving away from each other. (a) Planar truncation preserves motion completely for this scenario. (b) Isotropic truncation causes severe damping.}
\label{fig:dat_comparison}
\end{figure}

\subsection{Planar Division}
\label{sec:pdat}

Isotropic Division is effective in preventing penetration, but it is inherently conservative: by ignoring the displacement direction, it must guard against the worst case where all nearby primitives move directly toward each other.
This conservatism causes two problems in simulation. First, it introduces artificial damping (\autoref{fig:clothTwistComparison}), since displacements are truncated even when primitives are separating or sliding past each other. Second, it can cause deadlocks (\autoref{fig:treadmill}) when primitives are very close---their available motion shrinks to near zero regardless of direction, as illustrated by \autoref{fig:dat_comparison:b}.

However, we typically know $\Delta X$ before we apply division and truncation. By incorporating $\Delta X$, we design a less conservative \textit{Planar Division} that guarantees penetration-free deformation while better preserving the direction and magnitude of $\Delta X$.

Consider a vertex--triangle pair $(v,t)$ in an intersection-free configuration. Let $\vec{c}_{v,t}$ denote the closest point on $t$ to $v$, and define the normal vector $\vec{n}_{v,t} = \text{normalize}(\vec{x}_v - \vec{c}_{v,t})$. It can be shown that for any $\lambda_{v,t} \in (0,1)$, the plane passing through the point $\vec{p}_{v,t} = \vec{x}_v \mdel{+}\madd{-} \lambda_{v,t}(\vec{x}_v - \vec{c}_{v,t})$ with normal $\vec{n}_{v,t}$ will not intersect either $v$ or $t$.

For this particular pair, we thus define their exclusive regions as follows:
\begin{equation}
    \label{Eq:vertex_planar_exclusive_space_single}
    \mathcal{S}_v^{v,t} = \{\vec{x}\in\mathbb{R}^3 \mid (\vec{x} - \vec{p}_{v,t}) \cdot \vec{n}_{v,t} > 0\}
\end{equation}
\begin{equation}
    \label{Eq:triangle_planar_exclusive_space_single}
    \mathcal{S}_t^{v,t} = \{\vec{x}\in\mathbb{R}^3 \mid (\vec{x} - \vec{p}_{v,t}) \cdot \vec{n}_{v,t} < 0\}
\end{equation}

The \del{two faces}\add{vertex and the triangle} each occupies an (open) half space. One can immediately see the difference from the isotropic division method: the union of the exclusive regions spans the entire real space. In comparison, the isotropic division method only occupies a small, finite part of the real space. Intuitively, this is a great improvement in space utilization.
Over all vertex–triangle pairs, the exclusive regions become intersections of half-spaces:
\begin{equation}
    \label{Eq:vertex_planar_exclusive_space}
    \mathcal{S}_v = \bigcap_{t\in \mathcal{T}}  \mathcal{S}_v^{v,t}=\{\vec{x}\in\mathbb{R}^3\mid (\vec{x} - \vec{p}_{v,t}) \cdot \vec{n}_{v,t} < 0, \forall t \in \mathcal{T}\}
\end{equation} 
\begin{equation}
    \label{Eq:triangle_planar_exclusive_space}
    \mathcal{S}_t = \bigcap_{v\in \mathcal{V}}  \mathcal{S}_t^{v,t} =\{\vec{x}\in\mathbb{R}^3\mid (\vec{x} - \vec{p}_{v,t}) \cdot \vec{n}_{v,t} > 0, \forall t \in \mathcal{T}\}
\end{equation} 
which are mutually exclusive by construction.

This idea extends to every convex shape pairs. As long as they do not intersect, the division plane can be found by finding the closest point to each other. For example, the division plane for an edge-edge pair will be the plane passing through the point $\vec{p}_{e,e'} = \vec{c}_{e,e'} + \mdel{\lambda_{v,t}}\madd{\lambda_{e,e'}}(\vec{c}_{e',e} - \vec{c}_{e,e'})$ with normal $\mdel{\vec{n}_{e,e}}\madd{\vec{n}_{e,e'}}= \text{normalize}(\vec{c}_{e',e} - \vec{c}_{e,e'})$, where $\vec{c}_{e,e'}$ is the closest point on $e$ to $e'$.

Now that we have designated the exclusive regions, the next step is to choose an appropriate $\lambda_{v,t}$ for each pair. While a naive option is $\lambda_{v,t}=0.5$, we instead adapt $\lambda_{v,t}$ based on $\Delta X$, so that the available room on each side is proportional to the normal components of the displacements. Let:
\begin{align}
\label{Eq:vertex_triangle_deltas}
    \delta_v &= \max(-\Delta\vec{x}_v \cdot \vec{n}_{v,t}, 0) \\
    \delta_t &= \max(\Delta\vec{x}_a \cdot \vec{n}_{v,t}, \Delta\vec{x}_b \cdot \vec{n}_{v,t}, \Delta\vec{x}_c \cdot \vec{n}_{v,t}, 0) 
\end{align}
where $\Delta\vec{x}_v, \Delta\vec{x}_a, \Delta\vec{x}_b, \Delta\vec{x}_c$ are the displacements of $v$ and the three vertices of $t$ respectively.
We then compute $\lambda_{v,t}$ as:
\begin{equation}
    \label{Eq:triangle_planar_lambda}
    \lambda_{v,t} = 
    \begin{cases}
      0.5 & \text{if $\delta_v=\delta_t=0$} \\
      \frac{\delta_t}{\delta_t + \delta_v} & \text{otherwise}
    \end{cases}
\end{equation} 
If both sides are moving towards the division plane, the plane's placement is proportional to the maximum normal component of the displacements on each side, see \autoref{fig:vtdat:b}. 
When either side is moving away from (or parallel to) the plane, no truncation is imposed for that side, see \autoref{fig:vtdat}cd. 
This idea extends to other convex shape pairs as well, for example, for an edge-edge pair we have:
\begin{equation}
    \label{Eq:edge_planar_lambda}
    \lambda_{e,e'} = 
    \begin{cases}
      0.5 & \text{if $\delta_e=\delta_{e'}=0$} \\
      \frac{\delta_{e'}}{\delta_e + \delta_{e'}} & \text{otherwise}
    \end{cases}
\end{equation} 

\begin{figure}[ht]
\centering
\newcommand{\subfigvtdat}[3]{%
\begin{subfigure}{0.42\linewidth}
\centering
\fbox{\includegraphics[width=0.9\linewidth,trim=#2,clip]{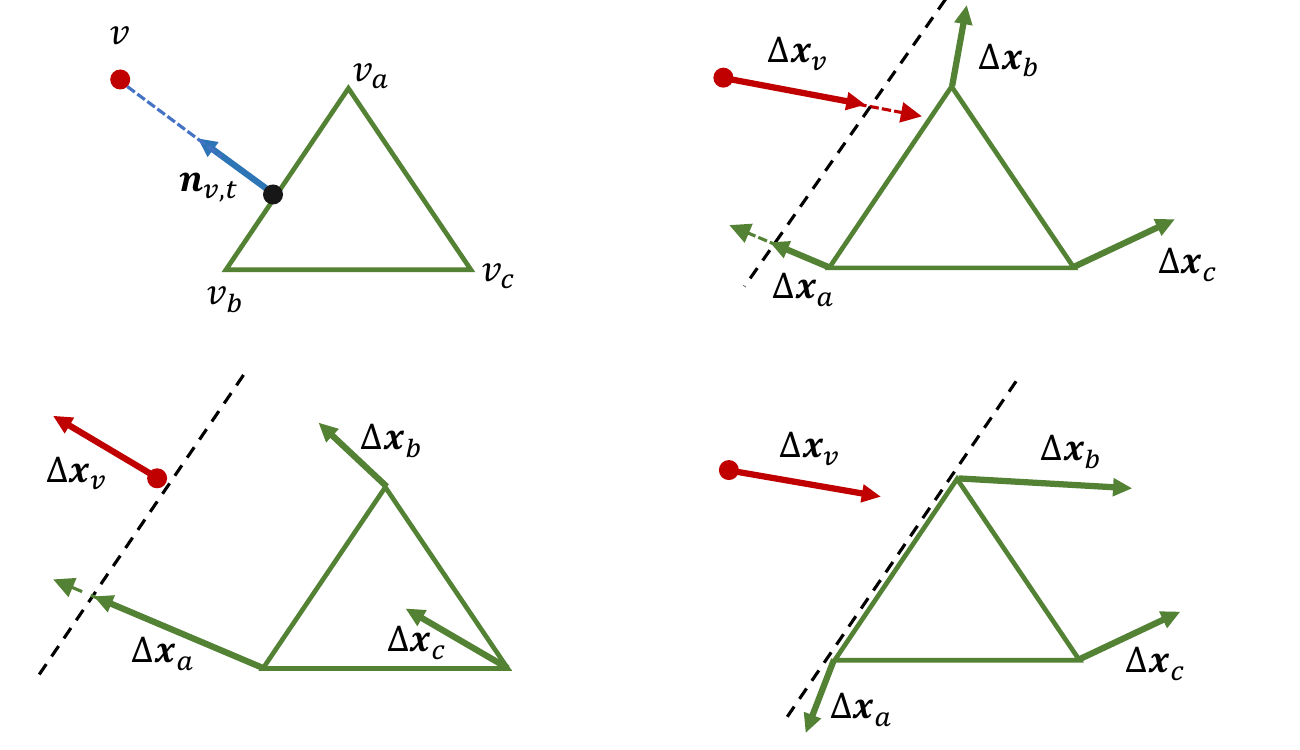}}
\caption{#3}
\label{fig:vtdat:#1}
\end{subfigure}}

\subfigvtdat{a}{7 207 347 0}{Geometry setup}%
\hspace{1em}%
\subfigvtdat{b}{340 207 20 0}{Approaching}\\[5pt]
\subfigvtdat{c}{7 0 347 160}{Vertex Moving Away}%
\hspace{1em}%
\subfigvtdat{d}{340 0 20 160}{Triangle Moving Away}
\caption{Planar division for a vertex--triangle pair. (a)~Vertex $v$ and triangle $t=(v_a, v_b, v_c)$ with normal $\vec{n}_{v,t}$ pointing from the closest point on $t$ to $v$. (b)--(d)~Different displacement scenarios: the dashed line shows the adaptive division plane. When $v$ moves toward $t$ (b), the plane shifts toward $v$; when separating (c) or moving tangentially (d), the plane shifts away, allowing freer motion.}
\label{fig:vtdat}
\end{figure}

Now we have finished the division step. The next step will be translating face-level constraints into vertex-level constraints. In particular, a deformation remains within the exclusive regions if and only if each displaced vertex satisfies the corresponding half-space constraints induced by its incident primitives:
\begin{theorem}
    \label{thm:vertex_level_eq}
    Each face's deformation stays in its exclusive region if and only if for every vertex \(v\in \mathcal{V}\):
    \begin{align}
    \label{Eq:vertex_displacement_constraint}
    \Delta\vec{x}_v + \vec{x}_v \in S_v, \\
    \Delta\vec{x}_v + \vec{x}_v \in S_t, \forall t \in \mathcal{T}_v, \\
    \Delta\vec{x}_v + \vec{x}_v \in S_e, \forall e \in \mathcal{E}_v 
\end{align}
\end{theorem}
where $\mathcal{T}_v$ and $\mathcal{E}_v$ respectively denote the one-ring triangles and edges incident to $v$.

Finally, we enforce these constraints by truncating $\Delta X$ so that it complies with those constraints. A straightforward method would be to use a separate parameter $t_v \in [0,1]$ for scaling the displacement of each vertex. 
Specifically, for each division plane defined by $\vec{n}$ and $\vec{p}$
we solve its intersection point with the ray $(\vec{x}_v, \Delta \vec{x}_v)$: $\vec{p}_i=\vec{x}_v + t_i \Delta \vec{x}_v$.
We use a conservative relation  $\gamma_r\in(0,1)$ to make sure the displacement stops before the division plane:
\begin{equation}
    \label{Eq:Conservative_truncation}
        t_v = 
    \begin{cases}
      1 & \text{if $\mdel{t}\madd{t_i}<0$ or $\mdel{t}\madd{t_i}\geq 1/\gamma_r$ } \\
      \gamma_r t_i & \text{otherwise}
    \end{cases}
\end{equation}
We call this algorithm \textit{Planar Divide and Truncate}, or Planar-DAT.

\subsection{Divide and Project}
While truncation is fast, it may not be energy-optimal. Since the constraints of Planar-DAT are a collection of half-spaces, the feasible region of $X$ is a polytope.
We can project $\Delta X$ onto the feasible region using an energy-aware metric\add{, where $\mathbf{H}$ is the energy Hessian and $\mathcal{S}$ is the feasible region defined by all half-space constraints}:
\begin{align}
\underset{ \Delta X'}{\text{arg min}} \; \| \Delta X' - \Delta X \|^2_{\mathbf{H}} \;\; \text{s.t.} \;\; \Delta X' \in \mathcal{S}\mdel{,}\madd{.}
\end{align}
We solve this using Dykstra's projection algorithm~\cite{boyle1986method}, which iteratively projects onto each half-space with closed-form updates. We call this variant \emph{Divide and Project} (DAP). While DAP can potentially converge faster  per iteration, the iterative projection incurs ${\sim}50\times$ overhead compared to truncation (\autoref{fig:convergence}), making Planar-DAT more practical for real-time applications. See \autoref{sec:dykstra} for the full DAP derivation. 

\subsection{Handling Inversion Constraints}
\label{sec:inversion}

\begin{figure}[t]
\centering
\includegraphics[width=0.8\linewidth,trim=0 250 0 80,clip]{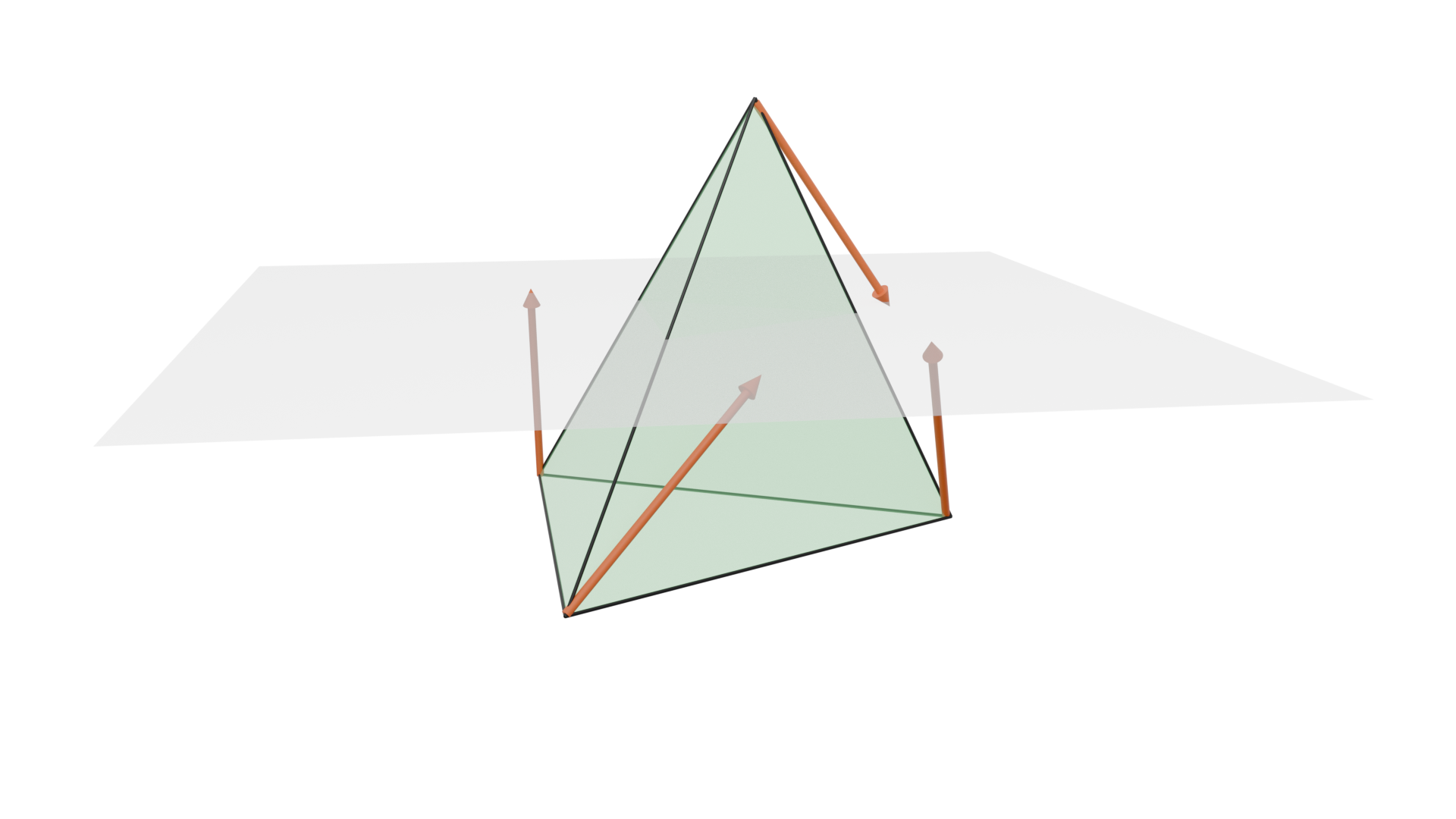}
\caption{Inversion constraint for a tetrahedron. The gray plane represents the division plane defined by one face of the tetrahedron. Each vertex must remain on its original side of the opposite face's plane to prevent inversion. Orange arrows indicate truncated vertex displacements.}
\label{fig:tetInversion}
\end{figure}

The inversion constraint, which prevents volumetric elements (such as tetrahedra) from flipping inside out, naturally integrates into the Planar-DAT framework (see \autoref{fig:tetInversion}). Specifically, for a tetrahedron, inversion occurs when one vertex crosses the plane defined by the opposite face (the other three vertices). Namely, this can be cast as a half-space exclusion condition: the vertex must remain on the same side of the face plane as it was in the rest configuration.
Importantly, the division plane for inversion is computed in exactly the same way as for a vertex–triangle contact. Therefore, the very same procedure for constructing planar exclusive regions and computing truncation can be directly reused for inversion constraints as for penetration prevention, with no additional handling required. This unifies penetration-free and inversion-free guarantees in a single Planar-DAT formulation.

\subsection{Handling Rotational Degrees of Freedom}
\label{sec:rotational}

So far, we have assumed that each vertex follows a linear trajectory during deformation. However, when simulating rigid bodies, vertices trace \textit{curved trajectories} due to rotational motion. A naive application of DAT to rigid bodies—treating vertex displacements as linear—can miss potential collisions that occur along the curved path, compromising the penetration-free guarantee.

\paragraph{Curved Trajectories.}
Following~\cite{ferguson2021intersection}, we model the trajectory of a vertex $i$ in a rigid body as it transforms from configuration $(\vec{\theta}_i, \vec{x}_i)$ with a transformation $(\Delta\vec{\theta}_i, \Delta\vec{x}_i)$:
\begin{equation}
    \label{Eq:curved_trajectory}
    \phi_i(t) = \mathcal{R}(\vec{\theta}_i + t\Delta\vec{\theta}) \hat{\vec{x}}_i + \vec{x}_i + t\Delta\vec{x}_i,
\end{equation}
where $\mathcal{R}(\vec{\theta})$ denotes the rotation matrix parameterized by the rotation vector $\vec{\theta}$ via Rodrigues' formula.
The rotation vectors and translations are linearly interpolated, but the resulting vertex trajectory $\phi_i(t)$ is nonlinear due to the rotation matrix $\mathcal{R}(\vec{\theta}_i(t))$.

\paragraph{Adapting DAT for Curved Trajectories}
We show that \autoref{thm:vertex_level_eq} extends naturally to rotational DoFs:
\begin{theorem}
    \label{thm:rigid_face_plane}
    Suppose a face (triangle or edge) undergoes rigid motion, starting entirely within a half-space bounded by a plane. The face's trajectory intersects the plane if and only if at least one of its vertices crosses the plane.
\end{theorem}
This theorem implies that we only need to ensure each vertex of a rigid body does not cross any of its associated division planes. Note that, unlike deformable meshes, these division planes arise solely from other rigid bodies, not from faces within the same body.


We find the latest safe time $t^*$ before a vertex crosses a plane via a two-stage algorithm (\autoref{alg:RigidTruncation}).
Given a plane defined by normal $\vec{n}$ and point $\vec{p}$, we uniformly sample $K$ points along $[0,1]$ and evaluate the signed distance:
\begin{equation}
    s(t_k) = \vec{n}^\top (\phi_i(t_k) - \vec{p}).
\end{equation}
We find the largest $t_k$ such that $s(t_j)$ has the same sign as $s(0)$ for all $j \leq k$. We then further apply bisection search in the interval $[t_k, t_{k+1}]$ for a certain number of iterations to further find the latest point $t^*$ before $s$ changes. The purpose of dividing the interval into $K+1$ intervals is to reduce the chance of not finding the first root in the presence of multiple roots.

However, this only guarantees that the endpoints lie on the correct side---the curved arc between $\phi_i(0)$ and $\phi_i(t^*)$ may still cross the plane.
To ensure safety, we compute a conservative bounding box of the trajectory segment $\phi_i([0, t^*])$ using the interval arithmetic method of~\cite{ferguson2021intersection}.
If this bounding box intersects the plane, we perform a bisection search within $[0, t^*]$:
we repeatedly bisect the interval and at each step check whether the trajectory bounding box over the current subinterval lies entirely on one side of the plane.
We continue this process until we find the largest $t^*$ (up to a desired accuracy) such that the bounding box of $\phi_i([0, t^*])$ remains fully on the safe side.
Note that when taking small steps, the trajectory is nearly linear and the algorithm rarely enters the second subdivision stage. In most cases, the more expensive interval arithmetic only needs to run once at the end to verify that the trajectory remains penetration-free.
After finding $t^*$, we apply the same relaxation ratio $\gamma_r$.
The detailed algorithm is given in \autoref{alg:RigidTruncation}.

\begin{algorithm}
\small
\SetKwFor{DoParallel}{parallel for each}{do}{end}
\newcommand\mycommfontB[1]{\footnotesize\textcolor[RGB]{0 128 0}{#1}}
\SetCommentSty{mycommfontB}
\LinesNumbered
\DontPrintSemicolon
\SetAlgoNoLine

\KwIn{
$\phi_i(t)$: curved trajectory of vertex $i$ (see \autoref{Eq:curved_trajectory});\\
$\vec{n}$: division plane normal;\\
$\vec{p}$: point on division plane;\\
$K$: number of uniform samples;\\
$n_{\text{bisect}}$: maximum bisection iterations;\\
$\gamma_r$: relaxation ratio
}
\KwOut{$t^*$: truncation ratio for vertex $i$}

\vspace{3pt}
\hrule
\vspace{6pt}

\tcp{Stage 1: Find latest safe time via sampling + bisection}
$s_0 \gets \vec{n}^\top(\phi_i(0) - \vec{p})$ \tcp{Initial signed distance}
$t^* \gets 1$\;
$t_\text{prev} \gets 0$\;

\For{$k \gets 1$ \KwTo $K$}{
    $t_k \gets k / K$\;
    $s_k \gets \vec{n}^\top(\phi_i(t_k) - \vec{p})$\;
    \uIf{$s_0 \cdot s_k < 0$}{
        \tcp{Sign change detected; bisect in $[t_\text{prev}, t_k]$}
        $(t_\text{lo}, t_\text{hi}) \gets (t_\text{prev}, t_k)$\;
        \For{$j \gets 1$ \KwTo $n_{\text{bisect}}$}{
            $t_\text{mid} \gets (t_\text{lo} + t_\text{hi}) / 2$\;
            $s_\text{mid} \gets \vec{n}^\top(\phi_i(t_\text{mid}) - \vec{p})$\;
            \uIf{$s_0 \cdot s_\text{mid} > 0$}{
                $t_\text{lo} \gets t_\text{mid}$ \tcp{Same side as start}
            }
            \Else{
                $t_\text{hi} \gets t_\text{mid}$ \tcp{Crossed to other side}
            }
        }
        $t^* \gets t_\text{lo}$\;
        \textbf{break}\;
    }
    $t_\text{prev} \gets t_k$\;
}

\tcp{Stage 2: Verify trajectory arc using interval arithmetic}
\uIf{$\texttt{boundingBoxIntersectsPlane}(\phi_i, [0, t^*], \vec{n}, \vec{p})$}{
    \tcp{Trajectory bounding box crosses plane; bisect in $[0, t^*]$}
    $(t_\text{lo}, t_\text{hi}) \gets (0, t^*)$\;
    \For{$j \gets 1$ \KwTo $n_{\text{bisect}}$}{
        $t_\text{mid} \gets (t_\text{lo} + t_\text{hi}) / 2$\;
        \uIf{$\texttt{boundingBoxIntersectsPlane}(\phi_i, [0, t_\text{mid}], \vec{n}, \vec{p})$}{
            $t_\text{hi} \gets t_\text{mid}$ \tcp{Shrink upper bound}
        }
        \Else{
            $t_\text{lo} \gets t_\text{mid}$ \tcp{Safe; expand lower bound}
        }
    }
    $t^* \gets t_\text{lo}$\;
}

$t^* \gets \gamma_r \cdot t^*$ \tcp{Apply relaxation ratio}

\Return{$t^*$}\;
\caption{Truncation Ratio for Curved Trajectory}
\label{alg:RigidTruncation}
\end{algorithm}

For the entire rigid body, the truncation ratio is the minimum across all its vertices and all division planes:
\begin{align}
    t_b = \min_{v \in b} \min_{(\vec{n}, \vec{p}) \in \mathcal{P}_v} t^*_{v, \vec{n}, \vec{p}},
\end{align}
where $\mathcal{P}_v$ denotes the set of division planes associated with vertex $v$. The rigid body's transformation is then scaled uniformly: $\Delta\vec{\theta}_b \gets t_b \cdot \Delta\vec{\theta}_b$ and $\Delta\vec{x}_b \gets t_b \cdot \Delta\vec{x}_b$.
Additionally, we bound the maximum displacement to $0.5\gamma_r r_q$ using a similar bisection algorithm to handle faces beyond the query range, for which the method degenerates to isotropic-DAT.

\subsection{Handling Animated Degrees of Freedom}
\label{sec:animated}

Animated objects are common in simulation---their motion is prescribed externally (via animation data) rather than by the physical solver. Since collisions with animated objects must be handled as robustly as with simulated ones, we treat the animated deformation just like optimizer-provided deformation---it enters the DAT framework in the same way.

Rather than applying the full animated step at once, we incrementally apply the animated deformation over multiple iterations, subject to the same truncation and safety checks as optimizer updates. Conceptually, the animation behaves as if driven by an infinitely stiff force toward its target pose---it will eventually overcome any finitely stiff resistance and reach its destination. This ensures penetration-free interaction and seamless integration between animated and simulated objects.
Practically, we iterate until the full animation is applied or a maximum iteration count is reached. If the animation cannot fully complete within this limit, the object continues from its current shape in the next time step, moving toward the updated animation target.

\section{Algorithm}
We propose an efficient parallel implementation of Planar-DAT, see \autoref{Alg:PlanarDAT}.
The preceding theoretical analysis defines each face's exclusive region by iterating over all other faces, which would result in a prohibitive O($N^2$) complexity in practice. To make the algorithm more practical, we only consider faces that are within a certain radius $r_q$. This can be achieved by querying neighboring faces within $r_q$ using a spatial acceleration data structure, such as a BVH (Bounding Volume Hierarchy). When the objective function of the optimizer includes collision energy, this neighborhood information is typically already available for reuse. The downside of restricting the neighborhood search to a certain range $r_q$ is that faces farther than $r_q$ are ignored, and their geometry and displacements are not considered. For these faces, the method degenerates to isotropic-DAT to prevent penetration, i.e., the maximum displacement of each vertex is up-bounded by $0.5\gamma_r r_q$. For rigid bodies, we similarly use bisection and interval arithmetic to ensure that the curved trajectory of each vertex remains within a sphere of radius $0.5\gamma_r r_q$.

We parallelize the implementation over collision pairs: the vertex-triangle pairs and edge-edge pairs returned by the collision queries described above. We allocate a thread for each pair; the thread computes the corresponding division plane and the truncation scalar $t_v$ for the four vertices involved in the pair.  We first initialize all $t_v$ with 1, then let each thread update it using atomic min operations. We found this to be the most efficient implementation, about $20\times$ faster than a naive parallel implementation of parallelizing by each vertex and iterating over all its adjacent faces and their collision pairs. It is even $2\times$ faster than a block-based implementation on GPU, where we launch a thread block for each vertex which processes all the adjacent faces in parallel and computes the minimum $t_v$ using a local reduction min. We attribute the advantage to the fact that the atomic min quickly reduces $t_v$ and most of the threads will not write to the global memory at all.
We also present an algorithm for integrating DAT into a simulation framework to enable penetration-free simulation, see \autoref{alg:SimulationPipeline}.

\begin{algorithm}
\small
\SetKwFor{Foreach}{for each}{do}{end}
\SetKwFor{DoParallel}{parallel for each}{do}{end}
\newcommand\mycommfont[1]{\footnotesize\textcolor[RGB]{0 128 0}{#1}}
\SetCommentSty{mycommfont}
\LinesNumbered
\DontPrintSemicolon
\SetAlgoNoLine

\KwIn{
$X \in \mathbb{R}^{3\times N}, \theta \in \mathbb{R}^{L}$: base linear and rotational states;\\
$\Delta X \in \mathbb{R}^{3\times N}, \Delta \theta \in \mathbb{R}^{L}$: the linear and rotational deformation;\\
$\mathcal{B}$: set of rigid bodies;\\

$\mathcal{F}_c\in \mathcal{V}\times\mathcal{F}, \mathcal{E}_c\in \mathcal{E}\times\mathcal{E}$: vertex-triangle and edge-edge pairs that are in contact;\\
$\gamma_r$: relaxation ratio;\\ 
$r_q$: query radius
}
\KwOut{
    $\Delta X_\mathrm{trunc} \in \mathbb{R}^{3\times N}, \Delta\theta_\mathrm{trunc} \in \mathbb{R}^L$: the truncated deformation
    
    }

\vspace{3pt}
\hrule
\vspace{6pt}

$t_v=1, \forall v \in \mathcal{V}$\;

\DoParallel{$(f_1,f_2) \in \mathcal{F}_c \cup \mathcal{E}_c$}{
    $(\vec{n}, \vec{p}) = \texttt{calculateDivisionPlane}(X, \Delta X, \Delta \theta, f_1,f_2)$ \tcp{\autoref{Eq:triangle_planar_lambda}, \autoref{Eq:edge_planar_lambda}}
    \tcp{$a,b,c,d$ are the four vertices on $f_1,f_2$}
    $(t_a^*, t_b^*, t_c^*, t_d^*) = \texttt{calculateTruncationRatio}(\vec{n}, \vec{p}, \Delta X, \Delta \theta, f_1,f_2)$\;
    $t_a \gets\texttt{atomic\_min}(t_a, t_a^*)$,
    $t_b \gets\texttt{atomic\_min}(t_b, t_b^*)$\;
    $t_c \gets\texttt{atomic\_min}(t_c, t_c^*)$,
    $t_d \gets\texttt{atomic\_min}(t_d, t_d^*)$\;
}


\DoParallel{$v \in \mathcal{V}$}{
   \(\Delta \vec{x}_v\gets \gamma_r \Delta \vec{x}_v \)\;
   \uIf{\(||\vec{x}_v|| > 0.5 \gamma_r r_q\)}{
        \(\Delta \vec{x}_v\gets \frac{0.5 \gamma_r r_q}{||\vec{x}_v|| }\Delta \vec{x}_v \)\;
   }
}

\tcp{Handle rigid body rotational DoFs (\autoref{sec:rotational})}
\DoParallel{$b \in \mathcal{B}$}{
    $t_b \gets \min_{v \in b} t_v$\;
    $t_b \gets \min(t_b, \texttt{boundTrajectory}(\phi_b, 0.5\gamma_r r_q))$\;
    $\Delta\vec{\theta}_b \gets t_b \cdot \Delta\vec{\theta}_b$,
    $\Delta\vec{x}_b \gets t_b \cdot \Delta\vec{x}_b$\;
}

\Return{$\Delta X, \Delta \theta$}\;
\caption{Planar DAT}
\label{Alg:PlanarDAT}
\end{algorithm}

\subsection{Integrating DAT into Simulation Pipeline}
We introduce how to effectively incorporate our DAT algorithm into a simulation pipeline to enable penetration-free simulation, see \autoref{alg:SimulationPipeline}.
One notable design choice is that collision detection does not need to run at every solver iteration, but can instead be executed at a user-specified frequency to save computation. To make this happen, we do not update the state after every iteration, instead we accumulate the incremental update returned by each solver iteration to $\Delta X$, and apply DAT truncation at every iteration. We only commit the accumulated deformation to the state and reset $\Delta X$ to zero when performing a new collision detection (see line 9\textasciitilde 12 in \autoref{alg:SimulationPipeline}). Note that DAT must be computed using the positions recorded before the collision detection, rather than $X+\Delta X$, because collision detections are applied on state $X$. The reason we choose to use a fixed collision detection frequency is that it allows us to capture the computation of an entire frame into a single CUDA graph.

\begin{algorithm}
\small
\SetKwFor{DoParallel}{parallel for each}{do}{end}
\newcommand\mycommfontC[1]{\footnotesize\textcolor[RGB]{0 128 0}{#1}}
\SetCommentSty{mycommfontC}
\LinesNumbered
\DontPrintSemicolon
\SetAlgoNoLine

\KwIn{
$X_\text{in}\in \mathbb{R}^{K\times3}$: stacked positions of vertices from previous step;\\
$V_\text{in}\in \mathbb{R}^{K\times3}$: stacked velocities of vertices from previous step;\\ 
$\vec{a}_{\text{ext}}$: external acceleration;\\ 
$\gamma$: a relaxiation parameter;\\ 
$M=\{\mathcal{V}, \mathcal{E}, \mathcal{T}\}$;\\
$r$: contact radius, $r_q$: query radius
}
\KwOut{$X\in \mathbb{R}^{K\times3}$: stacked positions of vertices for current step}

\vspace{3pt}
\hrule
\vspace{6pt}

$X \gets X_\text{in}$\;
$Y \gets X_\text{in} + \Delta t\,V_\text{in} + \Delta t^2 \vec{a}_{\text{ext}}$\;
$\Delta X \gets \text{initialGuess}(X_\text{in}, V_\text{in}, Y)$\;

$\mathcal{F}_c \gets \text{vertexTriangleCollisionDetection}(X, \mathcal{V}, \mathcal{T}, r_q)$\;
$\mathcal{E}_c \gets \text{edgeEdgeCollisionDetection}(X, \mathcal{E}, r_q)$\;

$\Delta X \gets \text{DAT}(X, \Delta X, \mathcal{F}_c, \mathcal{E}_c, r_q)$\;

\ForEach{$i$ \KwTo $n_\text{iter}$}{

    \uIf{requireCollisionDetection()}{
        $X \gets X + \Delta X$\;
        $\Delta X \gets \vec{0}$\;
        $\mathcal{F}_c \gets \text{vertexTriangleCollisionDetection}(X, \mathcal{V}, \mathcal{T}, r_q)$\;
        $\mathcal{E}_c \gets \text{edgeEdgeCollisionDetection}(X, \mathcal{E}, r_q)$\;
    }

    $\Delta X \gets \Delta X + \text{simulationIteration}(X, \Delta X, Y, \mathcal{F}_c, \mathcal{E}_c, M, r_q)$\;
    $\Delta X \gets \text{DAT}(X, \Delta X, \mathcal{F}_c, \mathcal{E}_c, r_q)$\;

    \tcp{Optional Convergence Evaluation}
    \uIf{$\text{evaluateConvergence}(X, \Delta X, Y, \mathcal{F}_c, \mathcal{E}_c, M, r_q)$}{
        break\;
    }
}

$X \gets X + \Delta X$\;
$V \gets (X - X_\text{in})/\Delta t$\;

\Return{$X, V$}\;
\caption{Simulation Step with DAT}
\label{alg:SimulationPipeline}
\end{algorithm}


\section{Results}
\subsection{Implementation}


We implemented both Isotropic-DAT and Planar-DAT in CUDA and integrated them with a GPU-based VBD (Vertex Block Descent) solver.
It is important to note that since VBD is a Gauss-Seidel solver, each solver iteration in \autoref{alg:SimulationPipeline} corresponds to processing a particular color, i.e., we immediately apply the DAT process before proceeding to the next color. 
We use \del{Offset Geometric Contact (}OGC\del{)} as our contact model, and a GPU-based BVH (Bounding Volume Hierarchy) for spatial queries. The entire computation pipeline resides on the GPU without requiring CPU-GPU data transfer, allowing us to capture a full simulation step in a single CUDA graph. We collect performance data on an NVIDIA RTX PRO 6000 Blackwell GPU.

OGC includes a scheme for determining when new collision detection is needed by counting vertices that have moved beyond their conservative bounds. However, this check requires CPU involvement and would alter the structure of the CUDA graph. For simplicity, we run collision detection at a fixed frequency of once every 5 solver iterations. In most experiments, we use a contact radius $r_c$ of approximately 0.2mm and set the query radius to $r_q = 1.5\,r_c$. Unlike Isotropic-DAT, we found that Planar-DAT is less sensitive to the choice of contact and query radii. It exhibits significantly reduced damping artifacts for contacts with high relative velocity and never falls into deadlock in our tests. Performance statistics and simulation parameters can be found in \autoref{table:Performance}. The  numbers of sampling points for rotational DoFs depends on the size of the time step. We use $K=8$ when the time step is 1/600 and scale it accordingly when time step changes.

Planar-DAT incurs a moderate increase in per-iteration runtime compared to Isotropic-DAT, as its direction-aware constraints require additional computation for each contact. However, Planar-DAT's robust handling of complex, high-velocity contact events allows us to safely use larger simulation time steps and significantly fewer solver iterations and collision detection per frame. More importantly, because the majority of simulation time is spent in collision detection, the extra overhead from the Planar formulation represents less than 10\% of the total cost. Consequently, Planar-DAT reduces the number of expensive collision queries and solver steps required, and in practice, overall simulation runs at least 2$\times$ faster than Isotropic-DAT for dense contact-rich scenarios.

As to the implementation for DAP (Divide and Project), we use a per-vertex projection scheme, as oppose to projecting the entire $X$ over all the half space constraints. Since each projection require a linear solve, and a simulation scene can have millions of collision, a global DAP algorithm is prohibitively expensive. In fact in our test we are never able to finish even one iteration. Instead, we project each vertex separately, onto its own polytope constraints. We use the vertex block of energy hessian matrix we obtained from the VBD solve as the metric.

\begin{figure}[t]
    \centering
    \includegraphics[width=0.8\linewidth]{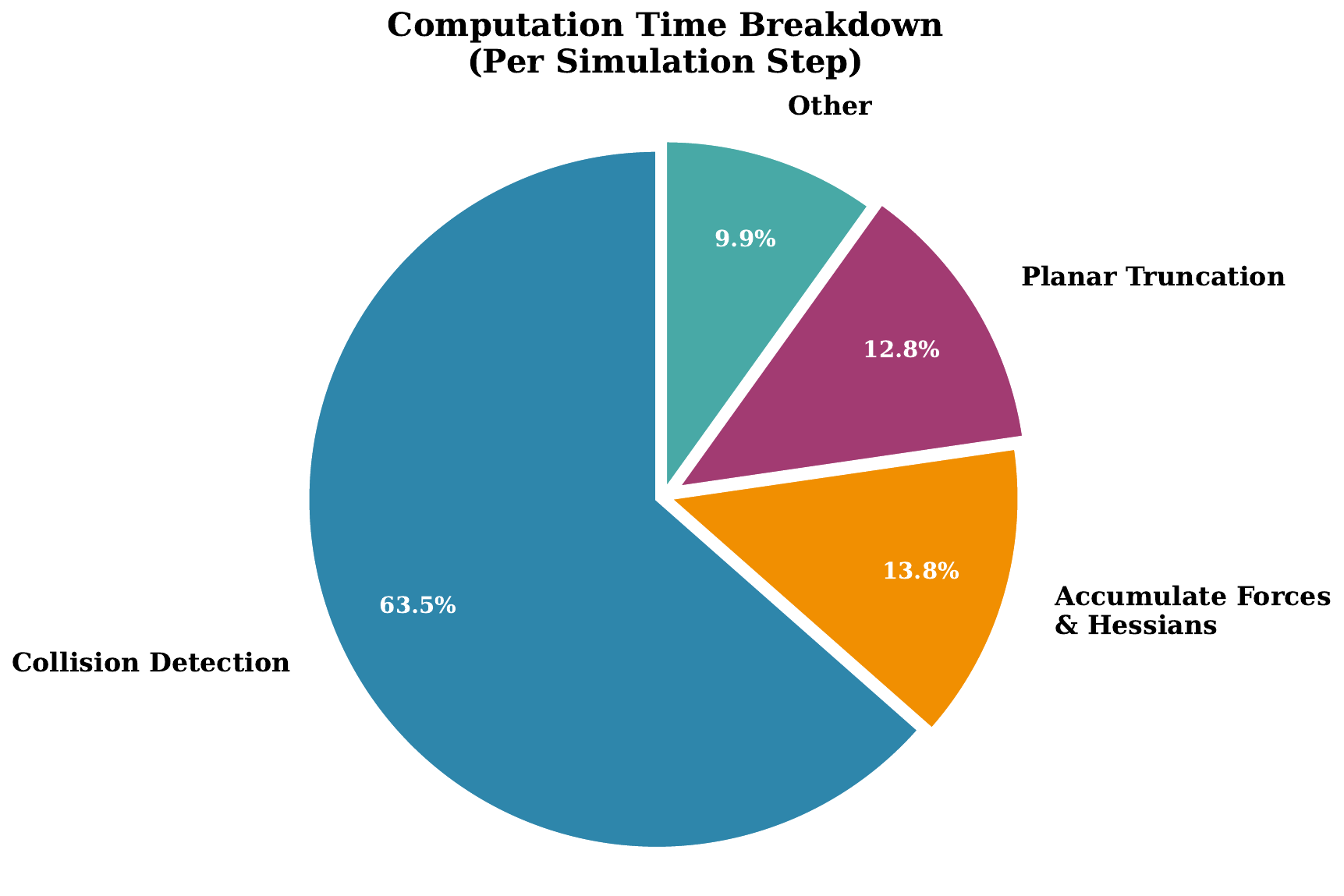}
    \caption{Computation time breakdown per simulation step on the cloth twist scenario (\autoref{fig:clothTwistComparison}). Collision detection dominates at 63.5\%, while Planar Truncation accounts for only 12.8\% of the total cost.}
    \label{fig:computation_time}
\end{figure}

\autoref{fig:computation_time} shows the computation time breakdown for Planar-DAT. Collision detection dominates at 63.5\% of total time, followed by accumulating forces and Hessians (13.8\%), Planar Truncation (12.8\%), and other operations (9.9\%). This breakdown highlights a key advantage of our approach: since collision detection is the bottleneck, the additional cost of direction-aware truncation is modest relative to the overall simulation cost. Moreover, Planar-DAT's ability to take larger stable steps often reduces the number of collision detection calls required, yielding net performance gains despite the per-step overhead.

\begin{table*}[t]
    \caption{Performance results and simulation parameters for all experiments. The query radius is $r_q = 1.5 r_c$. \add{Effective FPS is computed for Planar-DAT assuming a 60\,fps rendering target: $\text{FPS} = 1000\,/\,(\text{substeps} \times \text{avg.\ time/step})$, where $\text{substeps} = \lfloor(1/60)\,/\,\Delta t\rfloor$.}}
    \label{table:Performance}
    \resizebox{\linewidth}{!}{%
    \begin{tabular}{|l||ll|ll|lll|ll|ll|l|}
    \hline
    Experiment Name                                    & \multicolumn{2}{l|}{Number of}          & \multicolumn{2}{l|}{Stiffness (Tets / Triangles)}           & \multicolumn{3}{l|}{Contact}                                  & \multicolumn{2}{l|}{Simulation}          & \multicolumn{2}{l|}{Time per step (avg./max, ms)} & \add{Eff.} \\
                                                       & \multicolumn{1}{l|}{Vert.} & Tets / Triangles & \multicolumn{1}{l|}{$\lambda$} & $\mu$   & \multicolumn{1}{l|}{$k_c$} & \multicolumn{1}{l|}{$\mu_f$} & $r_c$ (mm) & \multicolumn{1}{l|}{Time Step} & Iter.    & \multicolumn{1}{l|}{Isotr-DAT} & Planar-DAT & \add{FPS} \\ \hline
    200 Layers of Cloth (\autoref{fig:multilayerClothes200}) & \multicolumn{1}{l|}{4M}  & 7.84M     & \multicolumn{1}{l|}{1e4}       & 1e4     & \multicolumn{1}{l|}{5e4}   & \multicolumn{1}{l|}{0.1}     & 2.5       & \multicolumn{1}{l|}{1/600}     & 10       & \multicolumn{1}{l|}{NA}    &     12.6/18.1  & \add{8}     \\
    Bullet (\autoref{fig:bullet})                      & \multicolumn{1}{l|}{9.7K}   & 46.6K       & \multicolumn{1}{l|}{2e10}      & 5e9     & \multicolumn{1}{l|}{2e10}  & \multicolumn{1}{l|}{0.01}       & 1        & \multicolumn{1}{l|}{1.7e-5}    & 20       & \multicolumn{1}{l|}{NA}        &     \del{1.8/2.5}\add{1.0/1.4}  & \add{1.0}    \\
    Gear Crusher (\autoref{fig:crusher})               & \multicolumn{1}{l|}{15K}    & 60K         & \multicolumn{1}{l|}{1e6}       & 1e5     & \multicolumn{1}{l|}{1e6}   & \multicolumn{1}{l|}{0.2}        & 5        & \multicolumn{1}{l|}{1/600}     & 10       & \multicolumn{1}{l|}{NA}        &     \del{1.0/1.4}\add{1.1/1.6}  & \add{91}    \\
    Multi-Physics (\autoref{fig:multiphysics})         & \multicolumn{1}{l|}{10.8K} & 14.3K/12.8K  & \multicolumn{1}{l|}{5e4/5e5}       & 5e4/5e5     & \multicolumn{1}{l|}{1e5}   & \multicolumn{1}{l|}{0.1}     & 3        & \multicolumn{1}{l|}{1/600}     & 10        & \multicolumn{1}{l|}{NA}              &     \del{1.1/1.6}\add{1.8/2.5}  & \add{56}    \\
    Yarn Twist Release (\autoref{fig:yarnTwist}) & \multicolumn{1}{l|}{130K}  & 130K      & \multicolumn{1}{l|}{0.1}       & NA     & \multicolumn{1}{l|}{2e-4}   & \multicolumn{1}{l|}{0.1}     & 1.5        & \multicolumn{1}{l|}{1/2520}     & 4       & \multicolumn{1}{l|}{0.39/0.97}  &     0.45/0.90  & \add{53}     \\
    Treadmill (\autoref{fig:treadmill})                & \multicolumn{1}{l|}{4.8K}  & 9.3K      & \multicolumn{1}{l|}{1e5}       & 1e5     & \multicolumn{1}{l|}{2e5}   & \multicolumn{1}{l|}{0.1}     & 2        & \multicolumn{1}{l|}{1/600}     & 10       & \multicolumn{1}{l|}{0.5/0.8}        &     0.55/0.9  & \add{182}     \\
    Unroll Cloth (\autoref{fig:unrollCloth})           & \multicolumn{1}{l|}{3K}    & 5.7K      & \multicolumn{1}{l|}{1e5}       & 1e5     & \multicolumn{1}{l|}{1e5}   & \multicolumn{1}{l|}{0.1}     & 2        & \multicolumn{1}{l|}{1/600}     & 10       & \multicolumn{1}{l|}{0.4/0.6}        &     0.45/0.65  & \add{222}     \\
    Gift Box Drop (\autoref{fig:giftBoxDrop})          & \multicolumn{1}{l|}{1.5K}  & 72/2.4K      & \multicolumn{1}{l|}{1e5/1e5}       & 1e5/1e5     & \multicolumn{1}{l|}{1e5}   & \multicolumn{1}{l|}{1.0}     & 2        & \multicolumn{1}{l|}{1/600}      & 10       & \multicolumn{1}{l|}{0.35/0.5}        &     0.4/0.55  & \add{250}     \\
    Cloth Twist Release (\autoref{fig:clothTwistComparison}) & \multicolumn{1}{l|}{10K}  & 19.6K      & \multicolumn{1}{l|}{1e5}       & 1e5     & \multicolumn{1}{l|}{1e5}   & \multicolumn{1}{l|}{0.2}     & 2        & \multicolumn{1}{l|}{1/600}     & 10       & \multicolumn{1}{l|}{1.0/1.6}  &     1.1/1.6  & \add{91}     \\
    Oscillating Cloth (\autoref{fig:oscillatingCloth}) & \multicolumn{1}{l|}{7.5K}  & 14.5K     & \multicolumn{1}{l|}{1e5}       & 1e5     & \multicolumn{1}{l|}{1e5}   & \multicolumn{1}{l|}{0.1}     & 2        & \multicolumn{1}{l|}{1/300}     & 10       & \multicolumn{1}{l|}{0.6/1.0}        &     0.65/1.1  & \add{308}     \\ \hline
    \end{tabular}
    }
\end{table*}

\subsection{Experiments}
\subsubsection{Stress Tests}

\begin{figure*}[t]
    \centering
    \newcommand{\figclothdrop}[1]{\includegraphics[width=0.245\linewidth,trim=250 150 200 20,clip]{Figures/ClothDrop200/#1.jpg}}
    \figclothdrop{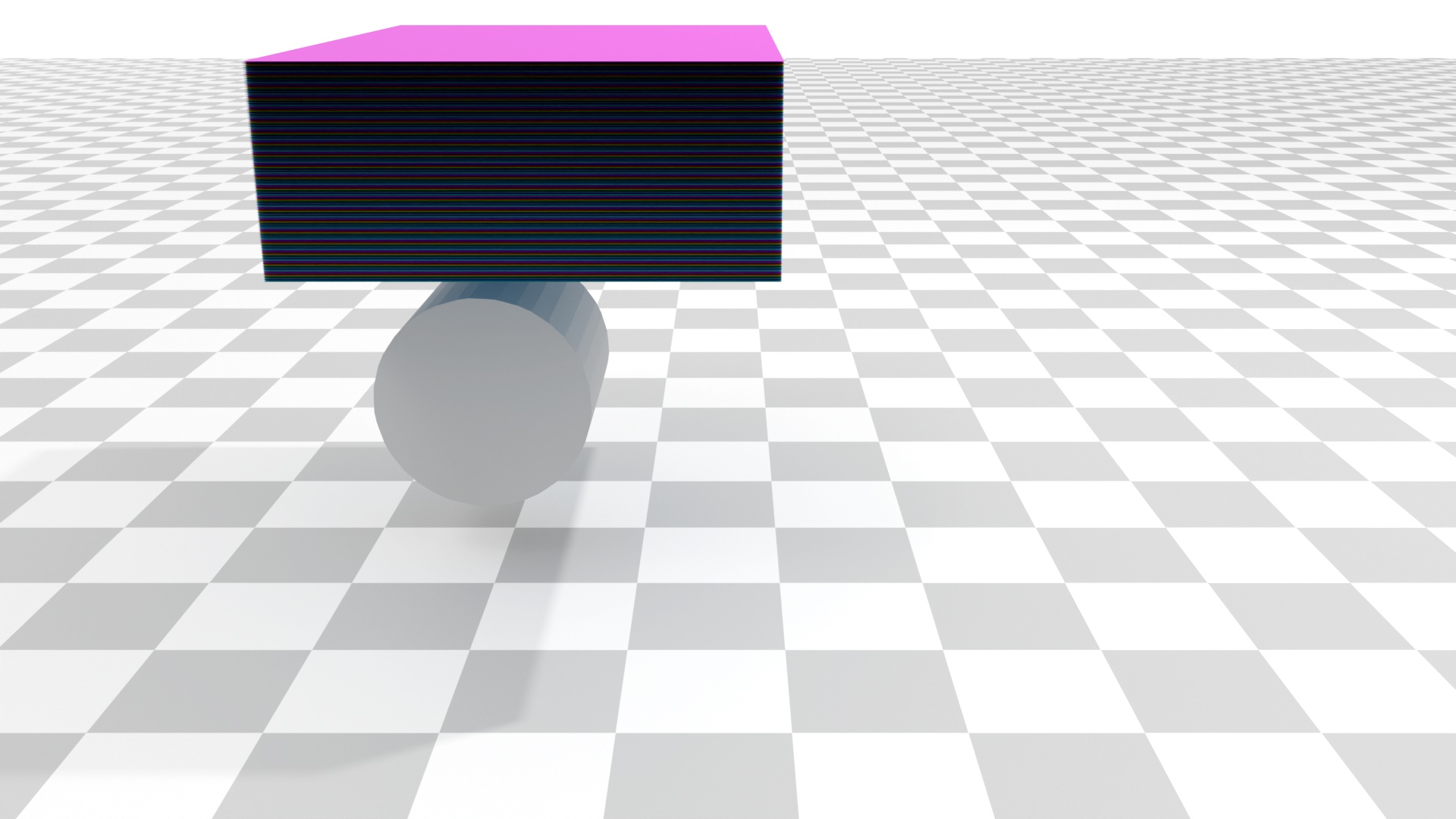}\hfill%
    \figclothdrop{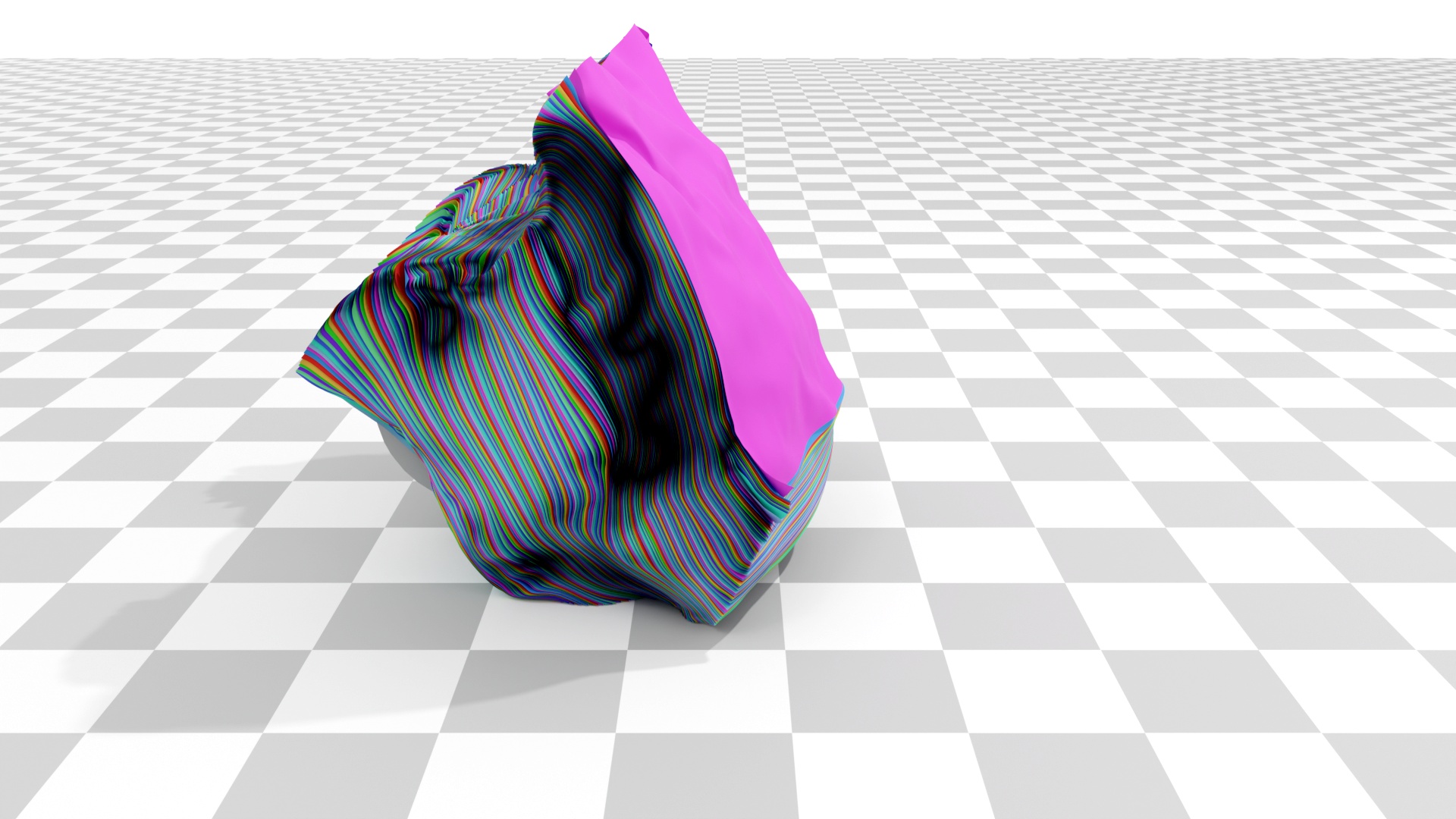}\hfill%
    \figclothdrop{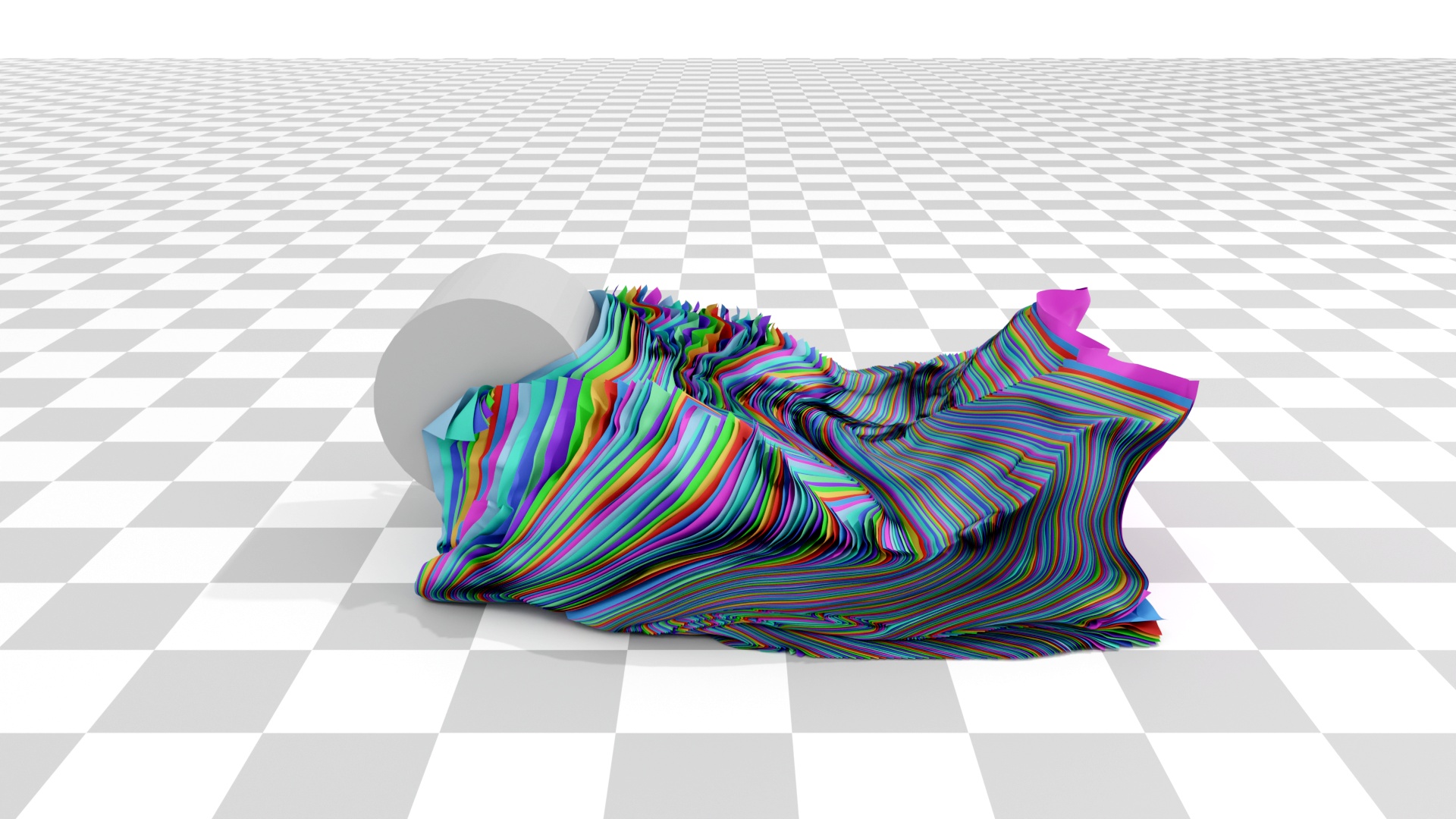}\hfill%
    \figclothdrop{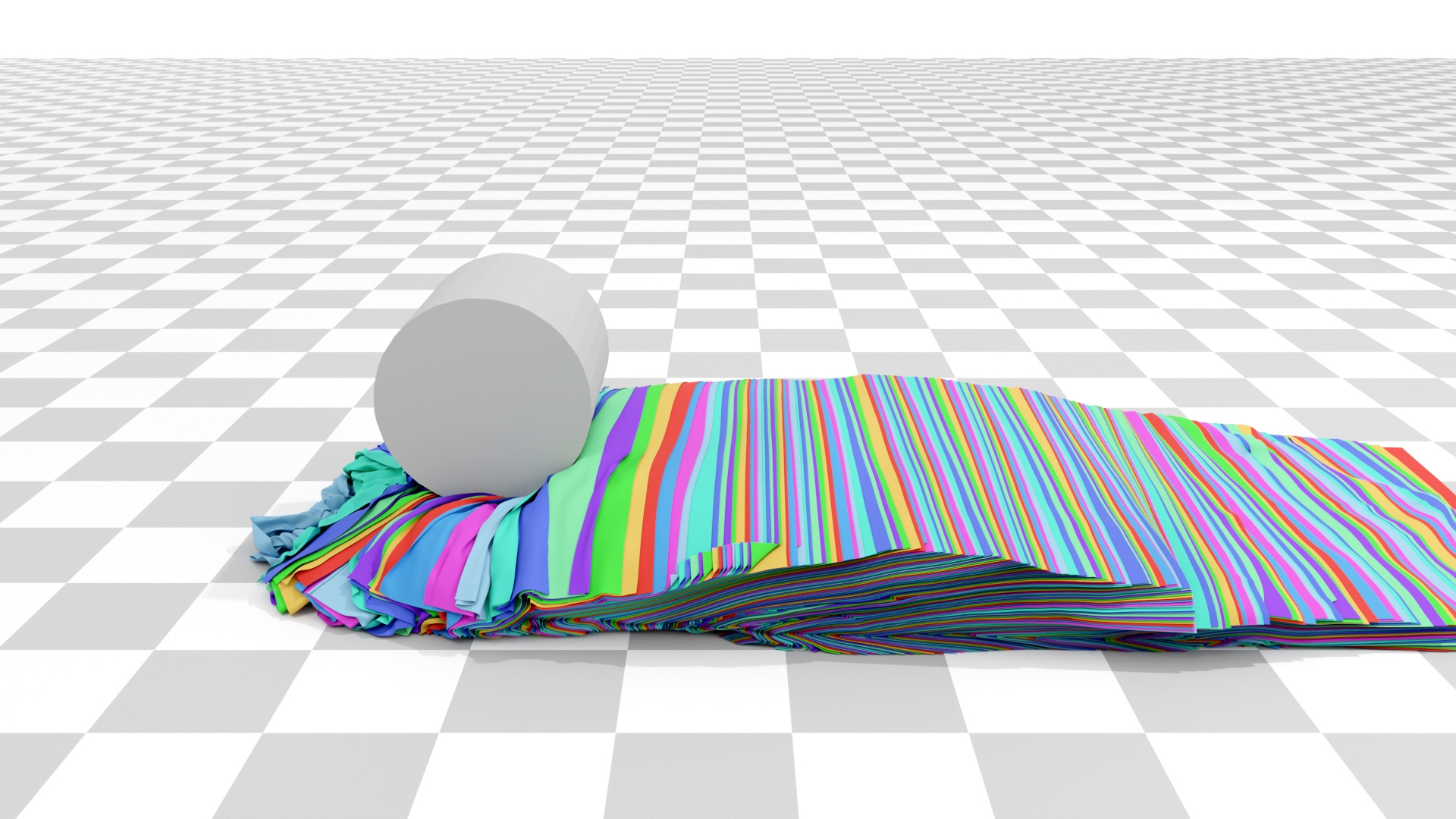}
    \caption{Two hundred layers of cloth are dropped onto a cylinder, then slide to the ground. This simulation has 4M vertices and 7.84M triangles, constantly producing more than 20 million contact pairs. ($\Delta t=1/600$s, avg.\ step time: 12.6ms)}
    \label{fig:multilayerClothes200}
\end{figure*}

We stress-test our method by dropping \emph{200 layers of cloth} onto a cylinder (\autoref{fig:multilayerClothes200}), constantly producing more than 20 million simultaneous contact pairs as the layers interact and settle. Despite extensive parameter tuning, OGC with Isotropic-DAT fails in this scenario due to locking and performance degradation from overly restrictive truncation. In contrast, Planar-DAT avoids these issues by not unnecessarily suppressing tangential and escaping motion. Our simulation uses twice the time step and half the iterations per step compared to those reported in OGC \cite{chen2025offset} on a scene four times larger. It remains stable, penetration-free, and highly dynamic, averaging under 12.6\,ms per step on the GPU.

\begin{figure*}[t]
    \centering
    \newcommand{\figbullet}[2]{%
    \begin{tikzpicture}
    \node[anchor=south west,inner sep=0] (image) at (0,0) {\includegraphics[width=0.19\linewidth,trim=100 200 120 100,clip]{Figures/Bullet/#1.jpg}};
    \ifx\relax#2\relax\else\node[anchor=south west,fill=white,fill opacity=0.7,text opacity=1,inner sep=2pt] at (image.south west) {\scriptsize #2};\fi
    \end{tikzpicture}}
    \figbullet{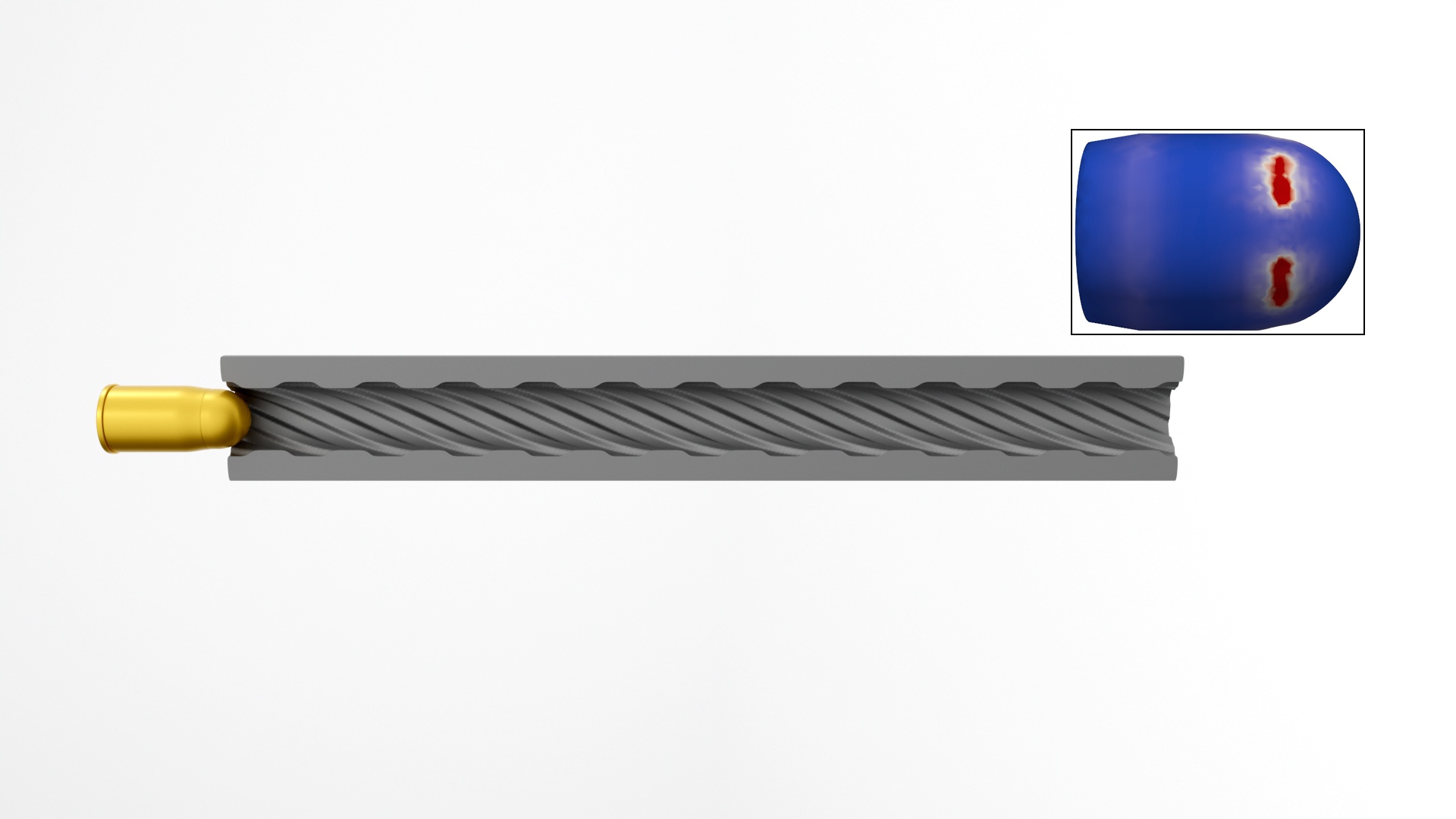}{$t\!=\!0$}\hfill%
    \figbullet{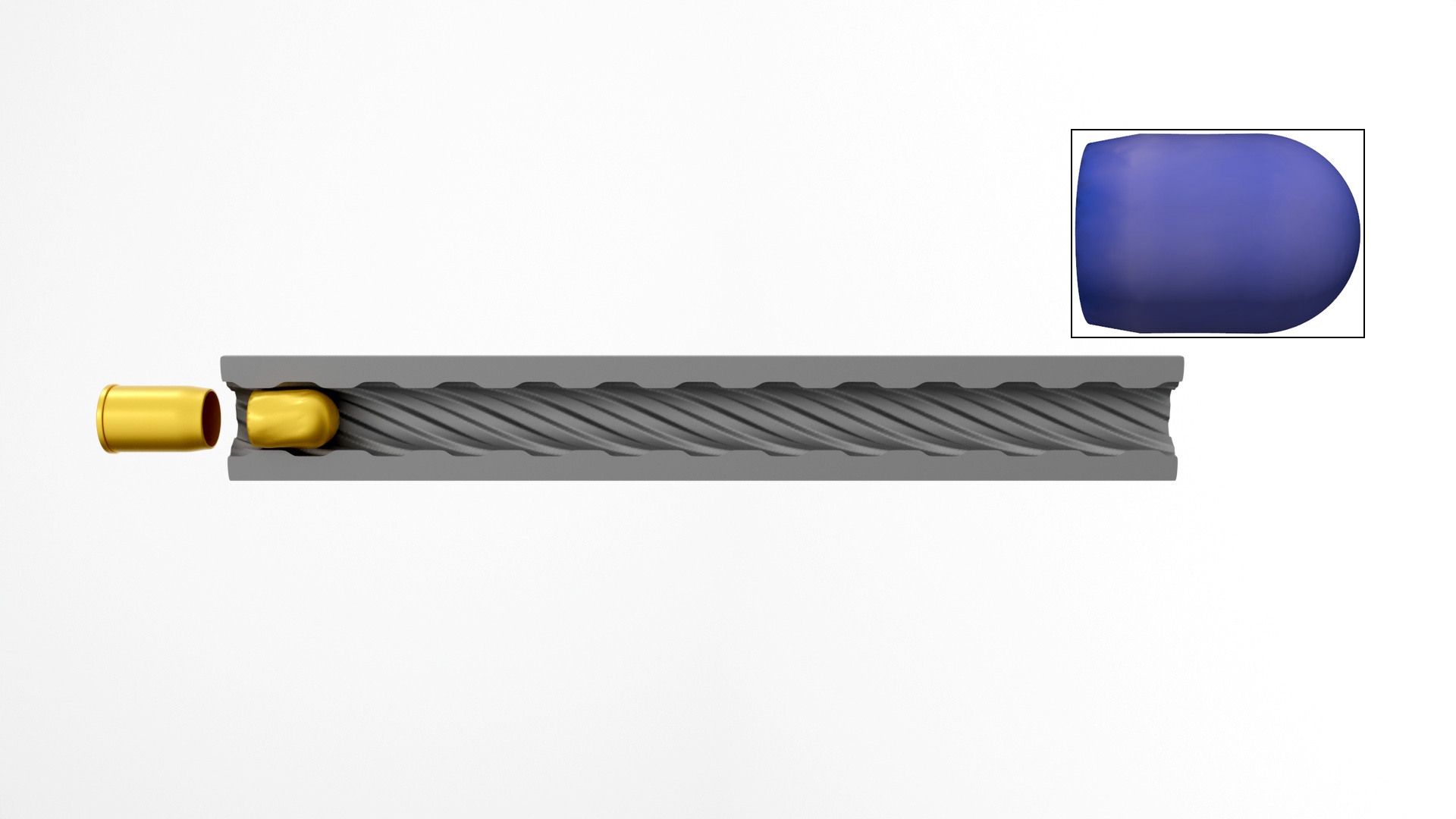}{$t\!=\!0.5$ms}\hfill%
    \figbullet{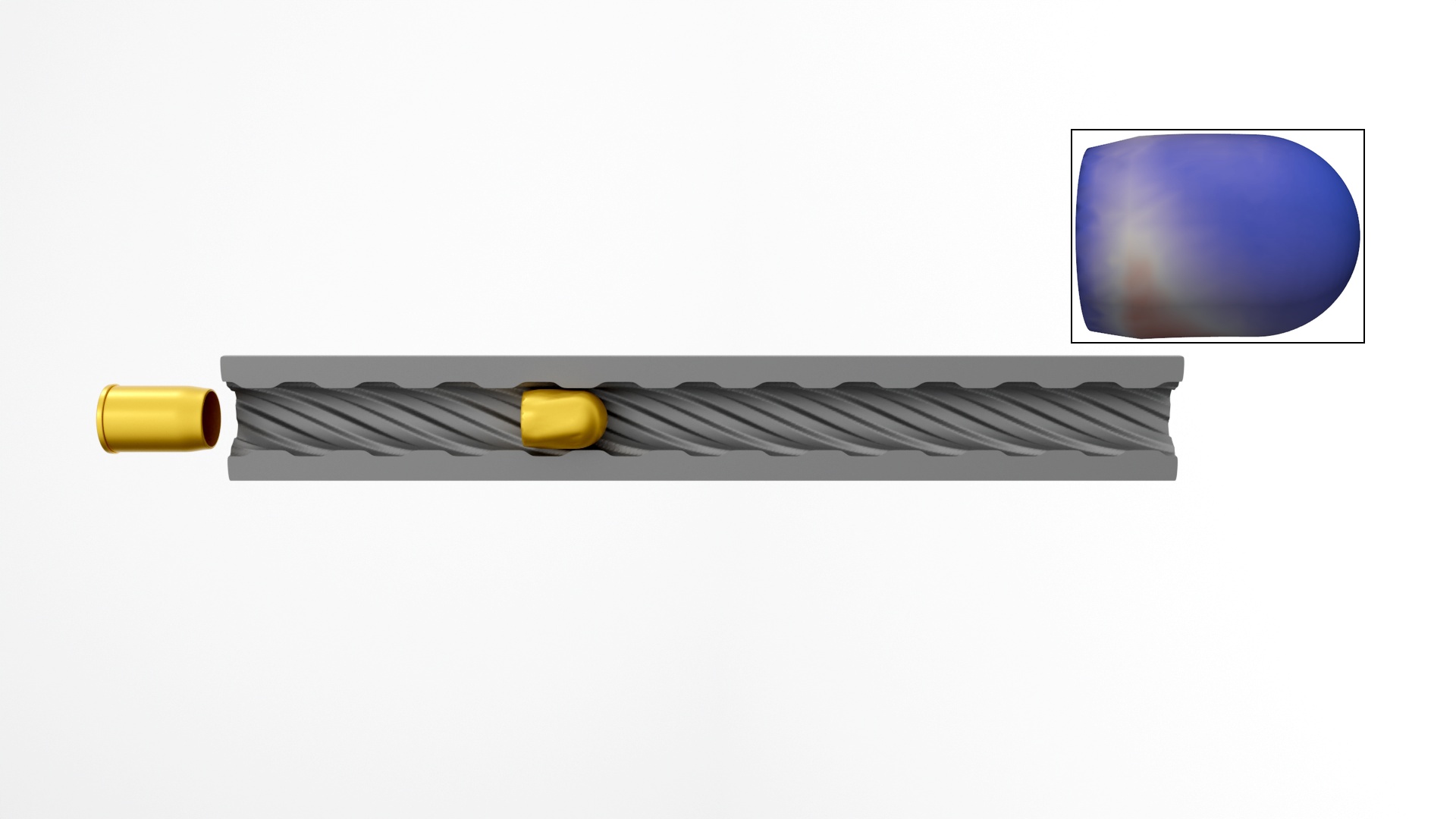}{$t\!=\!1$ms}\hfill%
    \figbullet{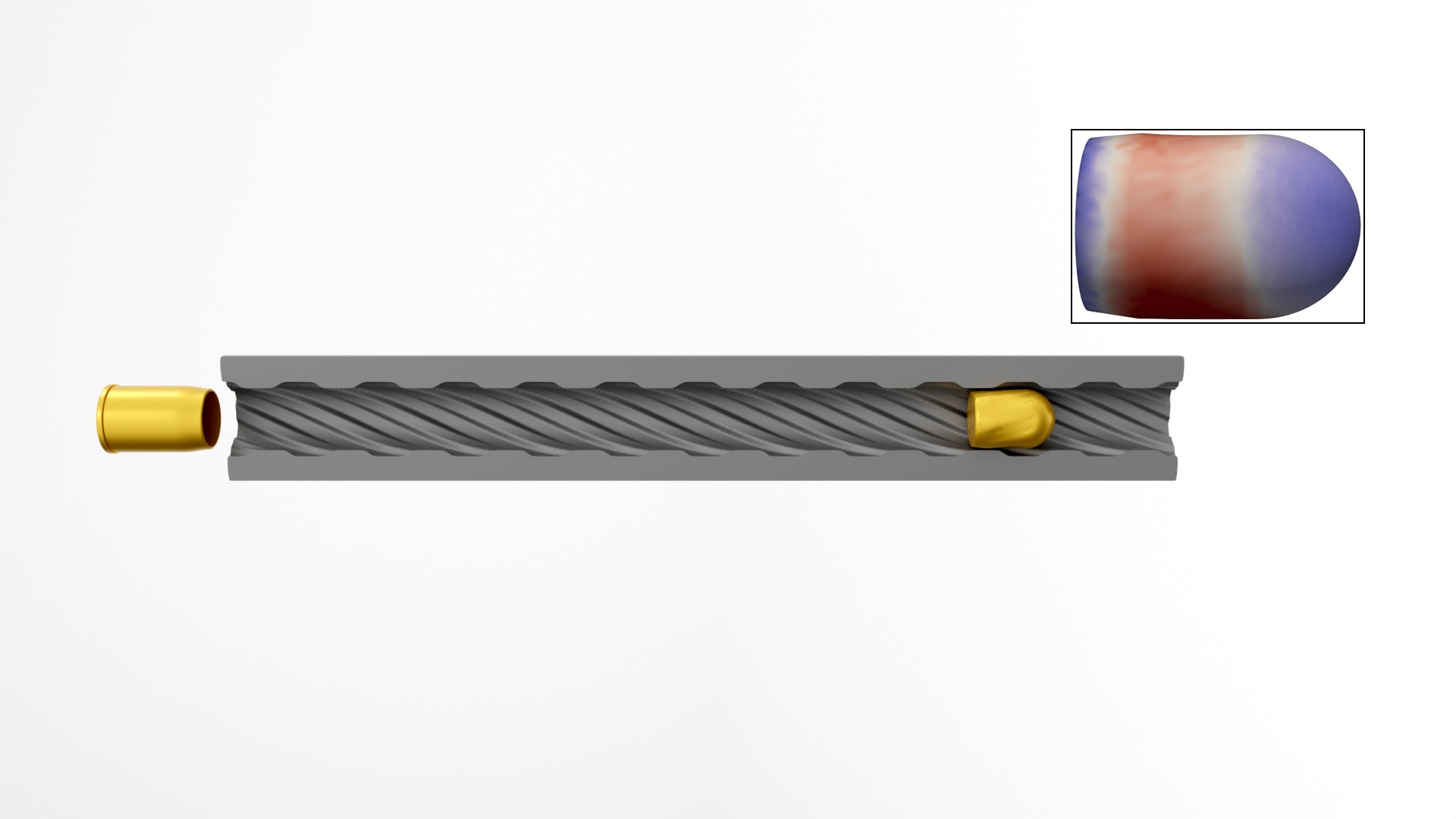}{$t\!=\!1.5$ms}\hfill%
    \figbullet{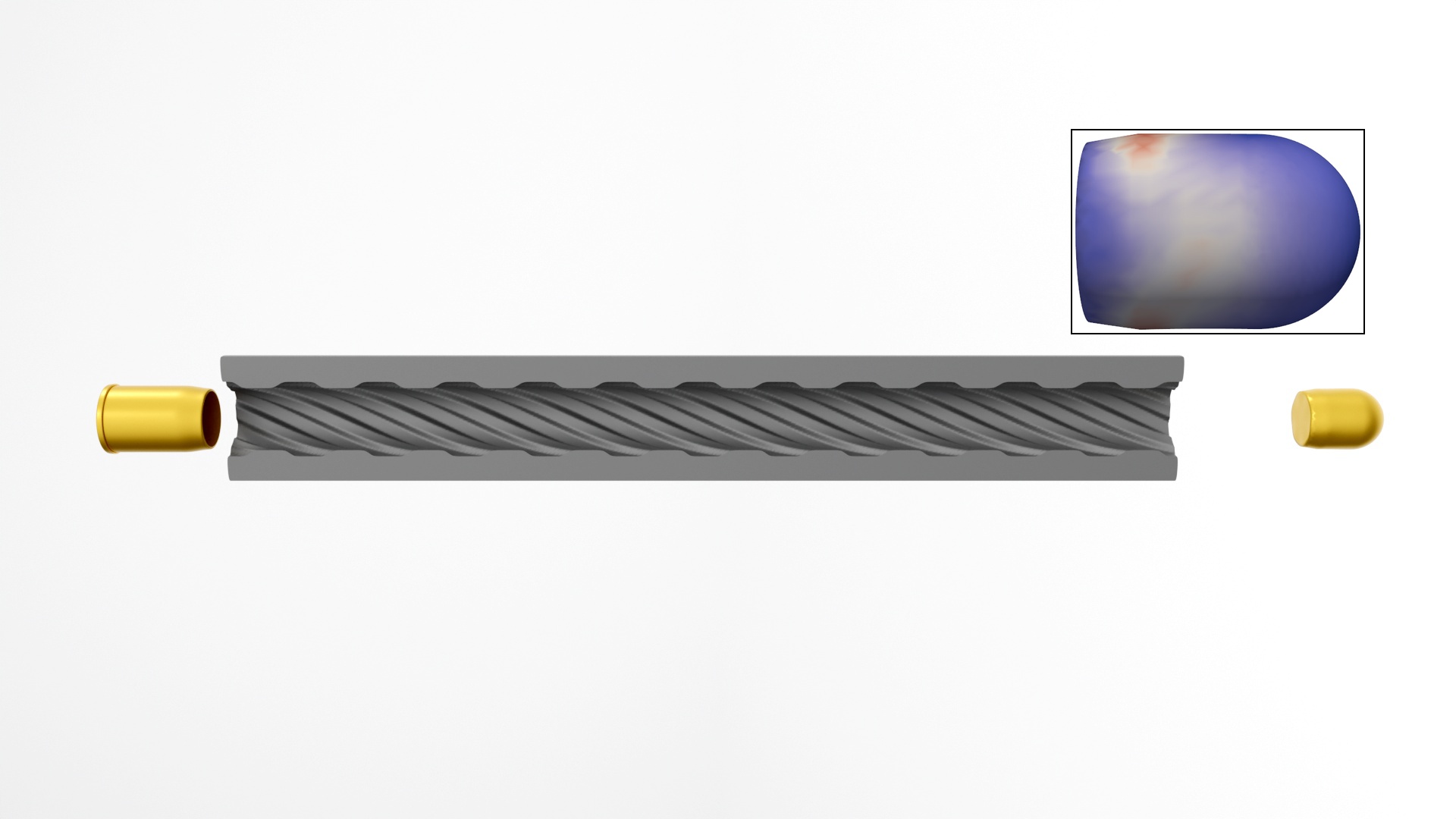}{$t\!=\!1.8$ms}
    \caption{Lead bullet firing through a rifled barrel (cross-section view). A lead bullet is pushed through a rigid barrel with helical rifling grooves. The bullet deforms to engage with the rifling as it accelerates through the barrel. \add{Inset: von Mises strain colormap (blue=low, red=high) showing deformation intensity.} ($\Delta t=1.7\times10^{-5}$s, avg.\ step time: \del{1.8}\add{1.0}ms)}
    \label{fig:bullet}
\end{figure*}
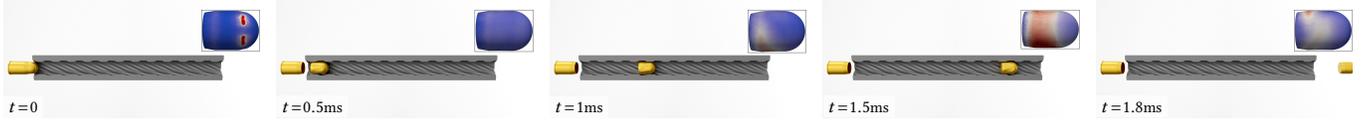

\autoref{fig:bullet} showcases our method's capability to robustly handle high relative speed contact with complex geometry. In this experiment, a 9mm lead bullet (represented as a tetrahedral mesh) is propelled through a rigid barrel with helical rifling by an applied force of $2 \times 10^9$ at the base. As the bullet moves along the barrel, it deforms to conform to the rifling grooves, while our method maintains strict penetration-free contact with the barrel walls, reaching an end speed of 369m/s. This scenario highlights our framework's ability to handle frictional contacts with extremely large relative velocities.

\begin{figure*}[t]
\centering
\newcommand{\figcrusher}[2]{%
\begin{tikzpicture}
\node[anchor=south west,inner sep=0] (image) at (0,0) {\includegraphics[width=0.19\linewidth,trim=300 100 100 300,clip]{Figures/Crusher/#1.jpg}};
\ifx\relax#2\relax\else\node[anchor=south west,fill=white,fill opacity=0.7,text opacity=1,inner sep=2pt] at (image.south west) {\scriptsize #2};\fi
\end{tikzpicture}}
\figcrusher{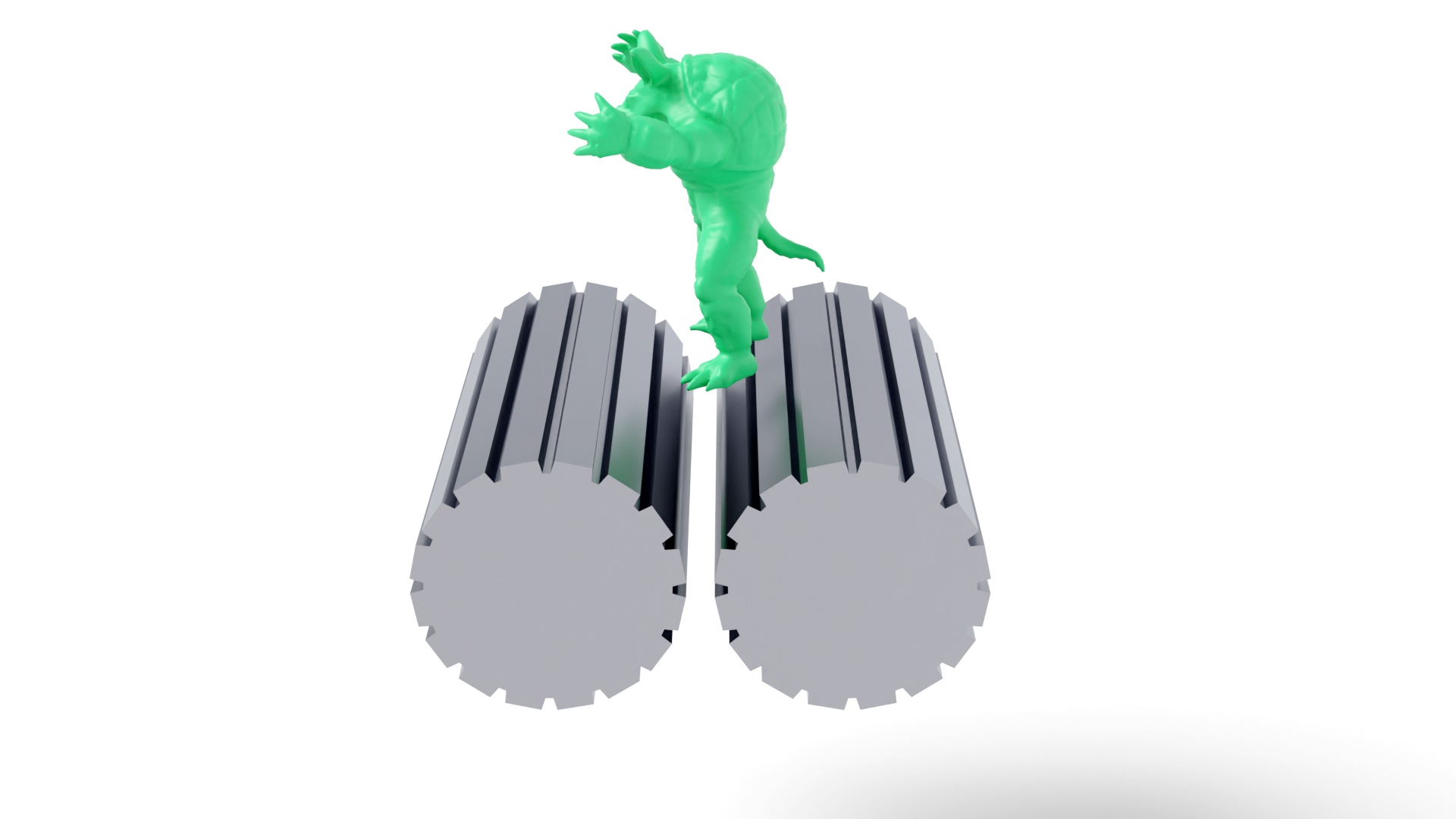}{$t\!=\!0$}\hfill%
\figcrusher{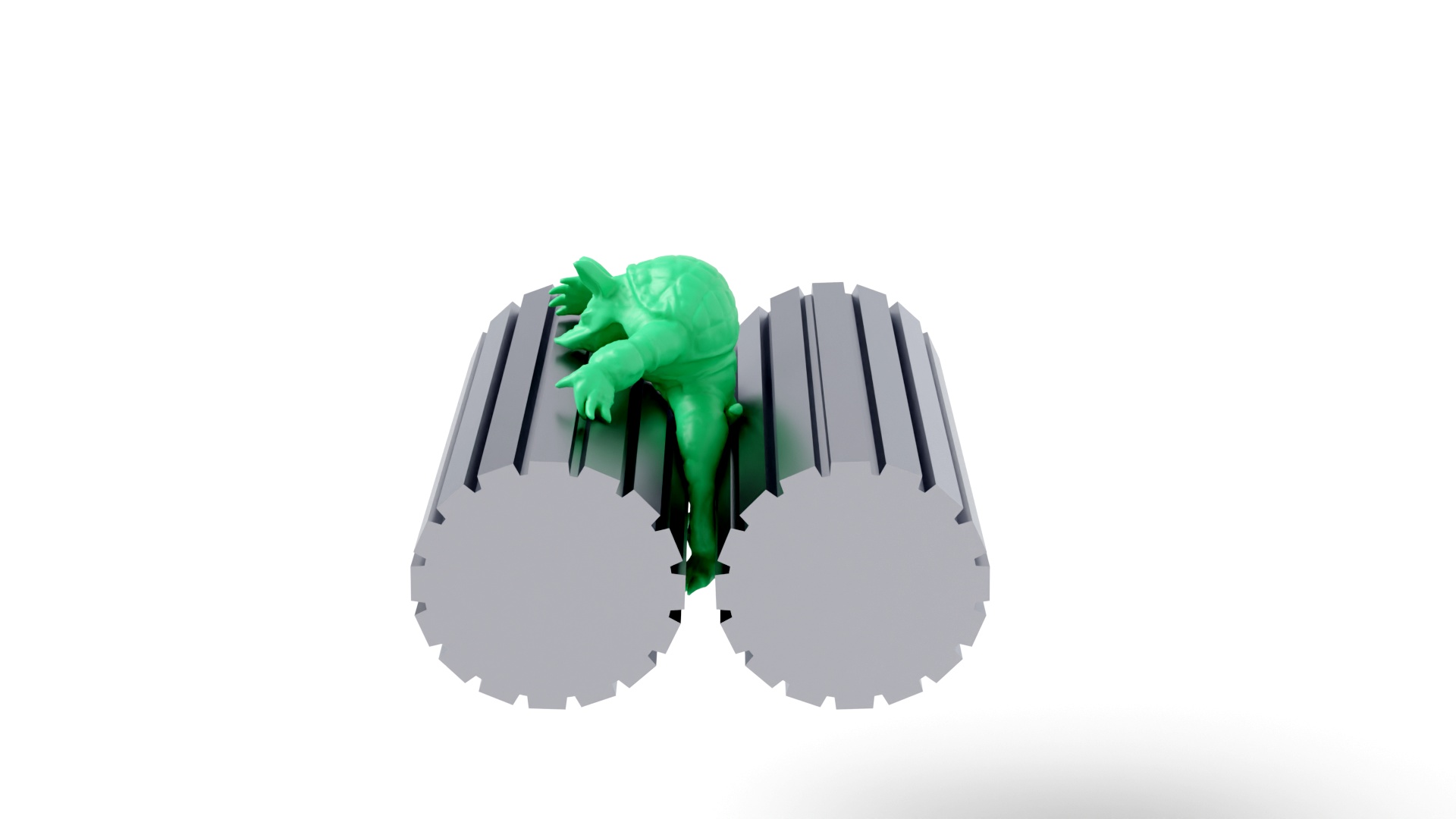}{$t\!=\!1$}\hfill%
\figcrusher{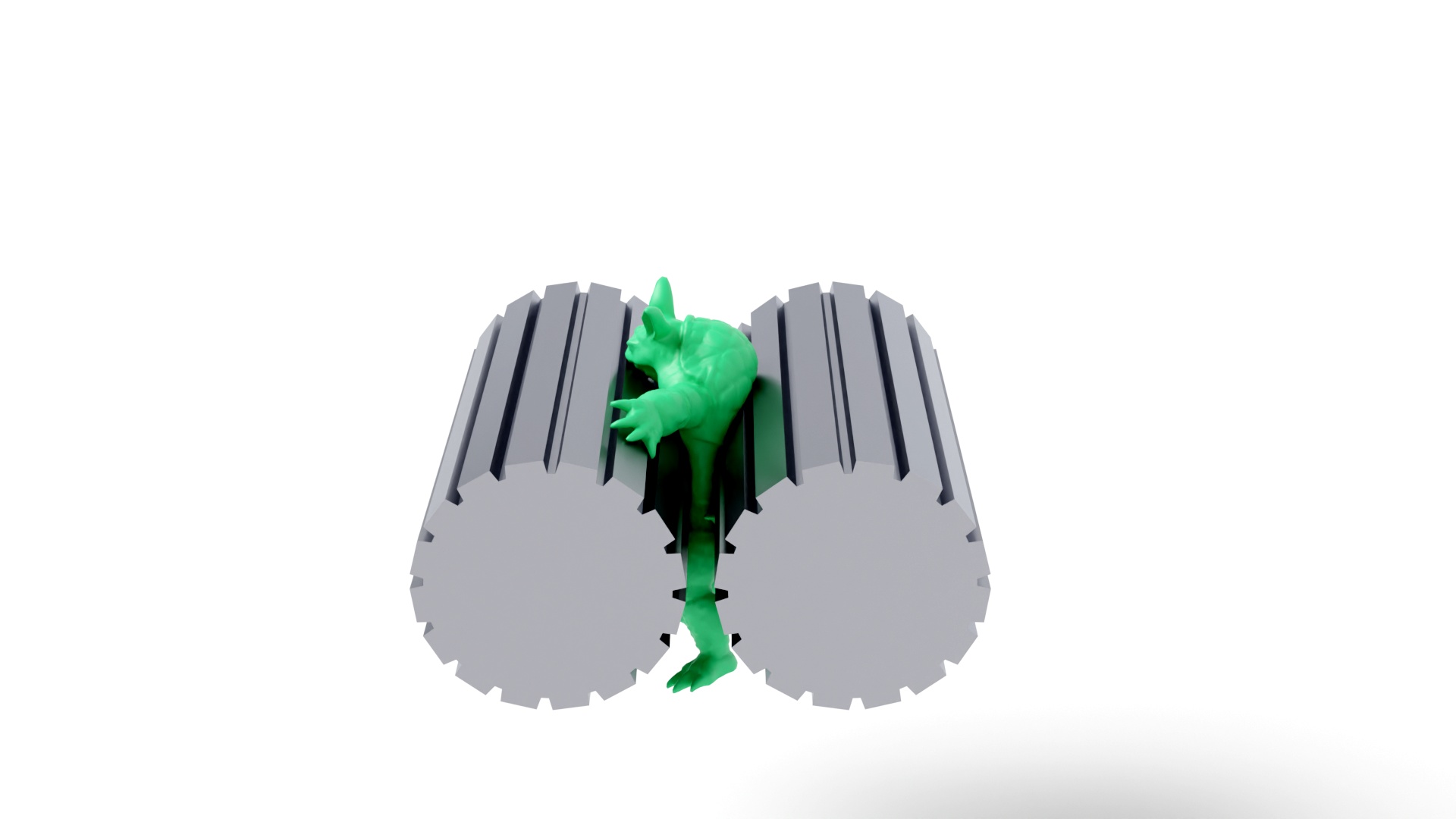}{$t\!=\!2$}\hfill%
\figcrusher{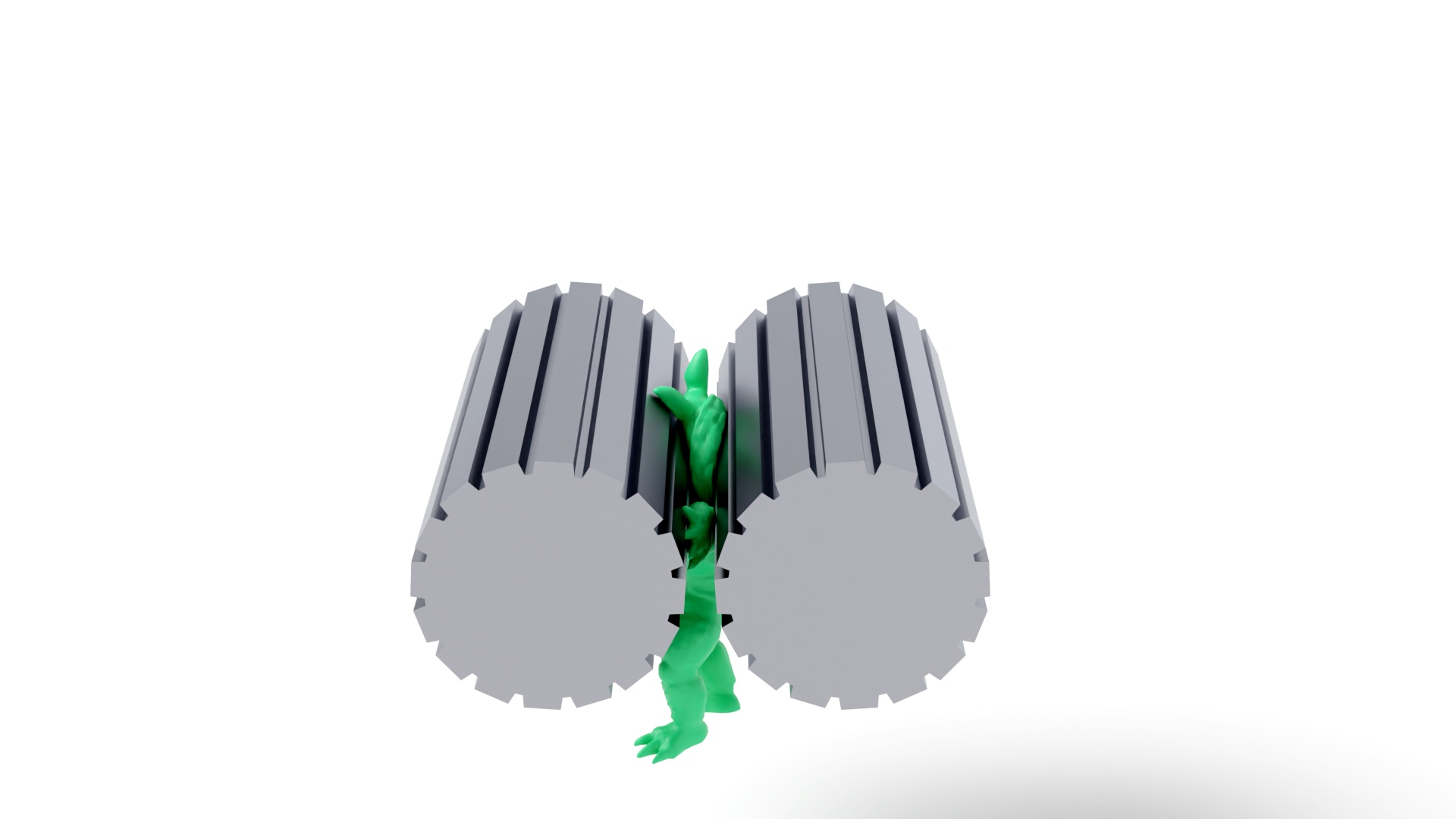}{$t\!=\!3$}\hfill%
\figcrusher{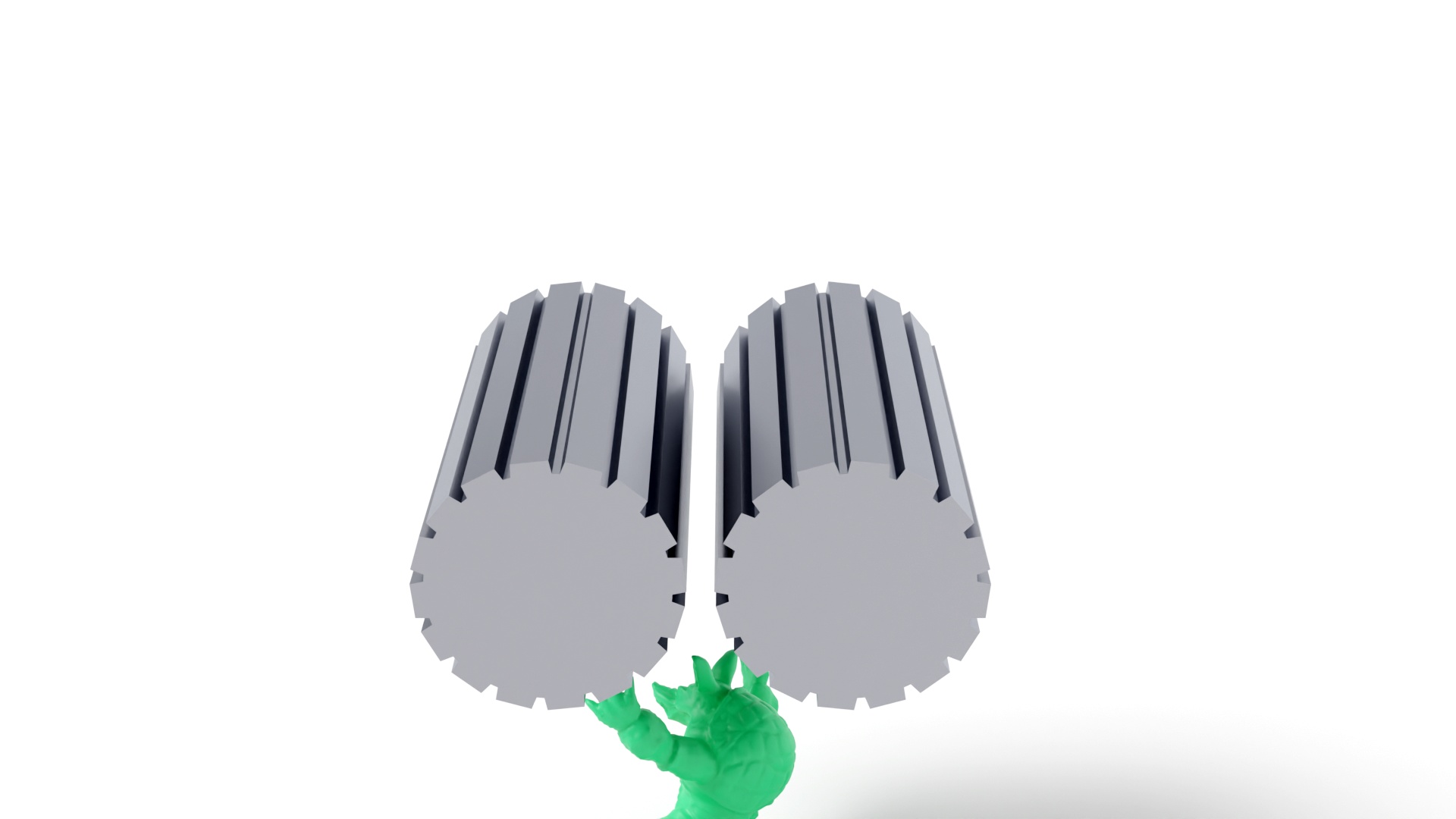}{$t\!=\!4$}
\caption{Gear crusher simulation. A deformable Armadillo is fed between two counter-rotating gear-shaped crushers. As the gears rotate, the soft body is squeezed and deformed through the narrow gap, demonstrating our method's ability to handle extreme compression with complex geometry contacts while maintaining penetration-free and inversion-free constraints throughout. ($\Delta t=1/600$s, avg.\ step time: \del{1.0}\add{1.1}ms)}
\label{fig:crusher}
\end{figure*}
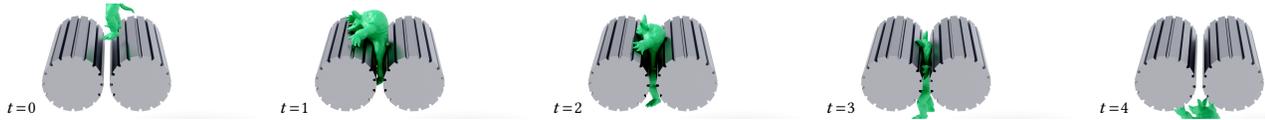

\autoref{fig:crusher} demonstrates our method's robustness under extreme compressive deformation with complex rigid-deformable contact. A soft body Armadillo (15K vertices, 60K tetrahedra) is fed between two counter-rotating gear-shaped crushers, whose motions are animated. As the gears rotate inward, the soft body is squeezed through the narrow gap between gear teeth, undergoing severe compression and shearing. Our Planar-DAT formulation maintains penetration-free and inversion-free constraints throughout this challenging scenario, preserving simulation stability even as the mesh undergoes extreme deformation.

\subsubsection{Multi-Physics}

\begin{figure*}[t]
\centering
\newcommand{\figmultiphys}[2]{%
\begin{tikzpicture}
\node[anchor=south west,inner sep=0] (image) at (0,0) {\includegraphics[width=0.19\linewidth,trim=100 200 100 150,clip]{Figures/Multiphysics/#1.jpg}};
\ifx\relax#2\relax\else\node[anchor=south west,fill=white,fill opacity=0.7,text opacity=1,inner sep=2pt] at (image.south west) {\scriptsize #2};\fi
\end{tikzpicture}}
\figmultiphys{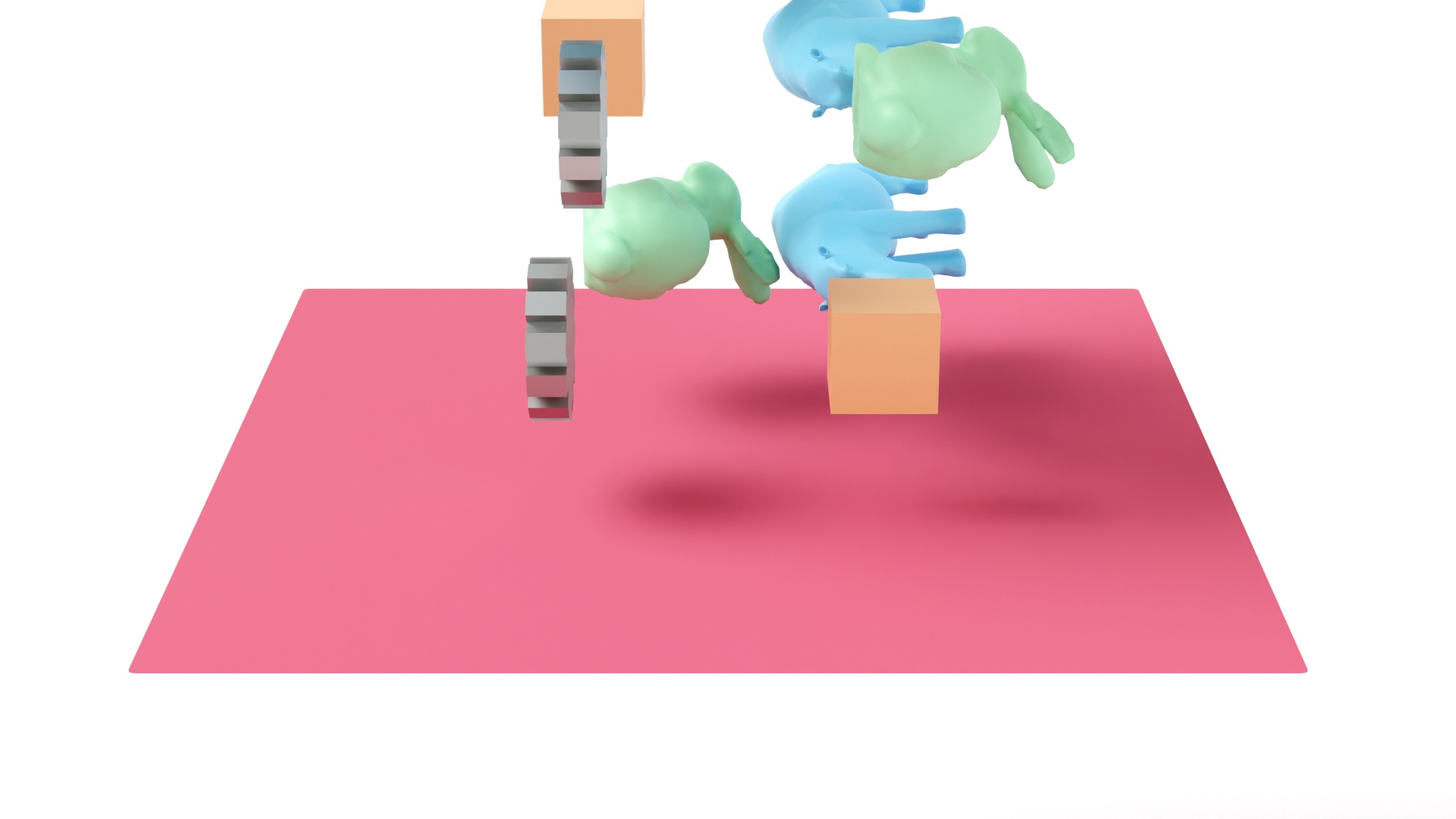}{$t\!=\!0$}\hfill%
\figmultiphys{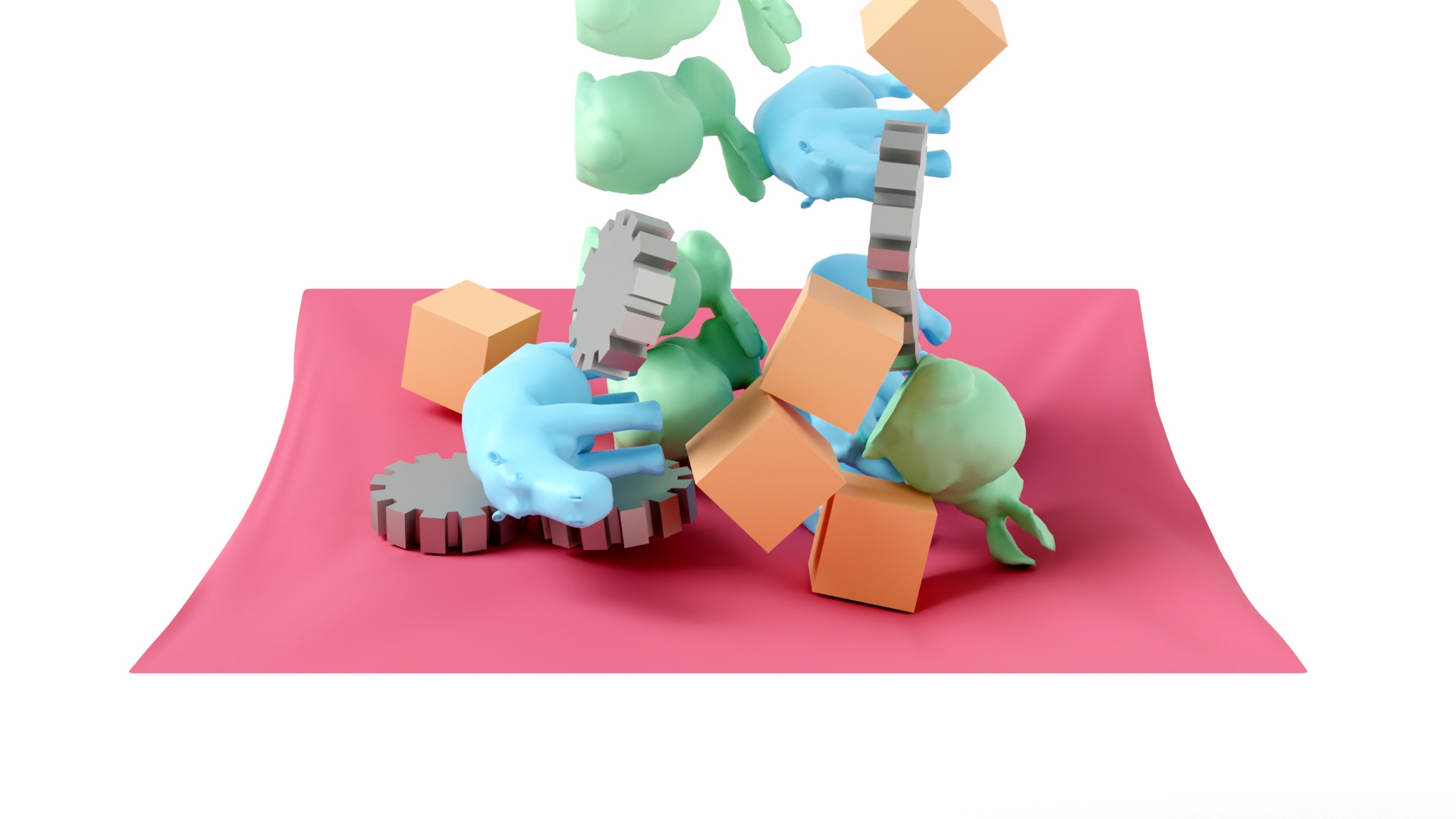}{$t\!=\!0.62$}\hfill%
\figmultiphys{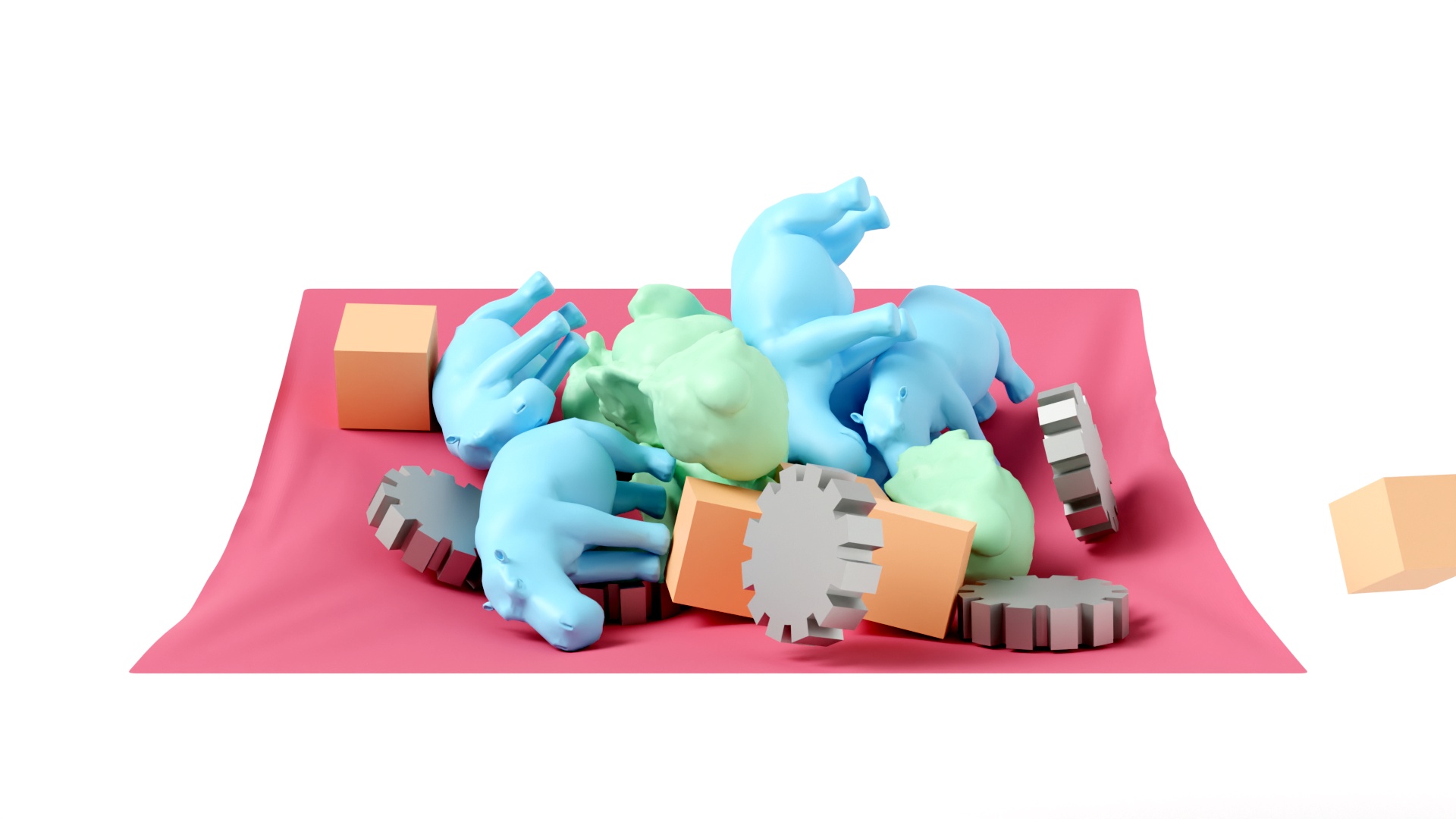}{$t\!=\!1.25$}\hfill%
\figmultiphys{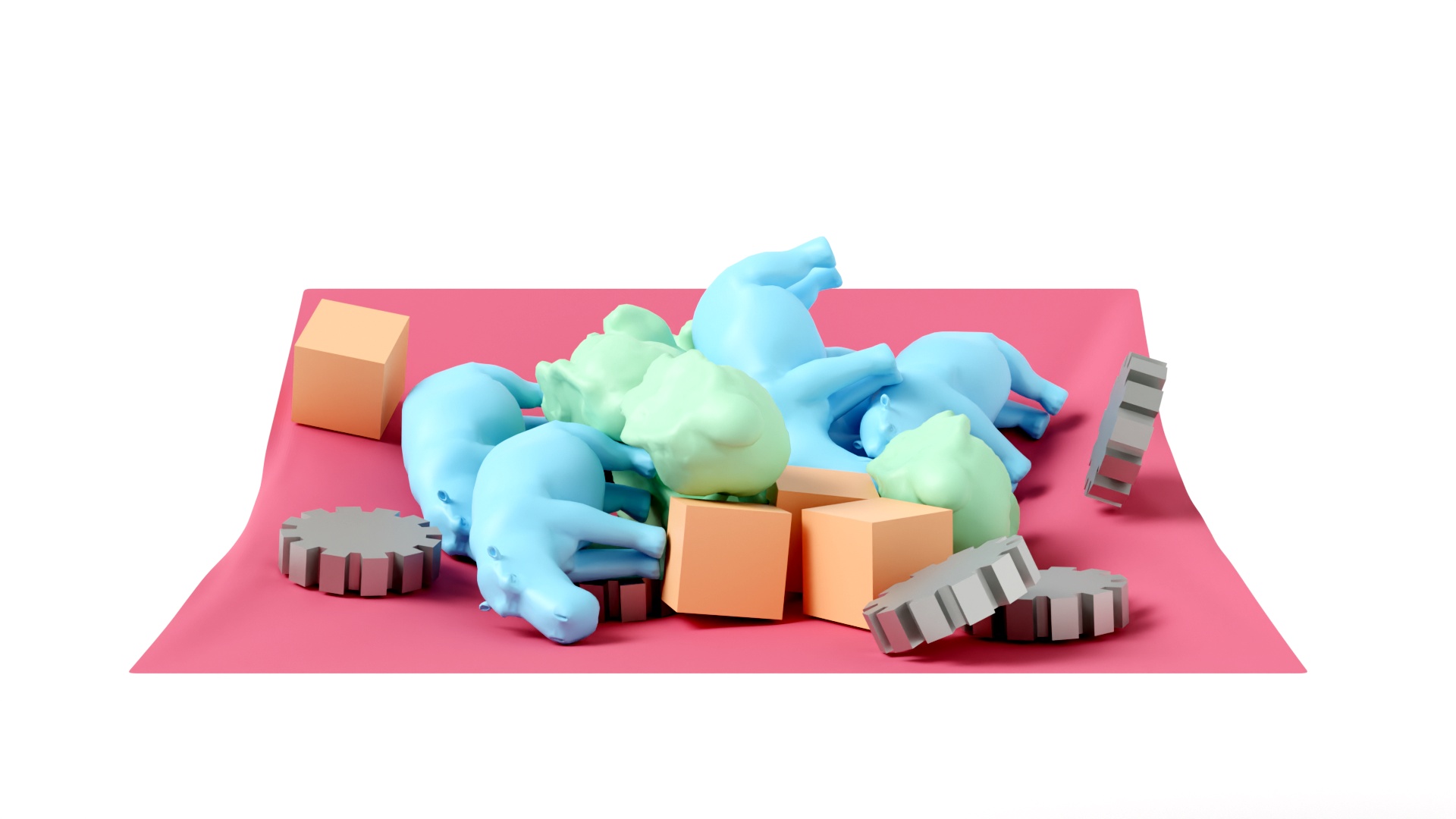}{$t\!=\!1.87$}\hfill%
\figmultiphys{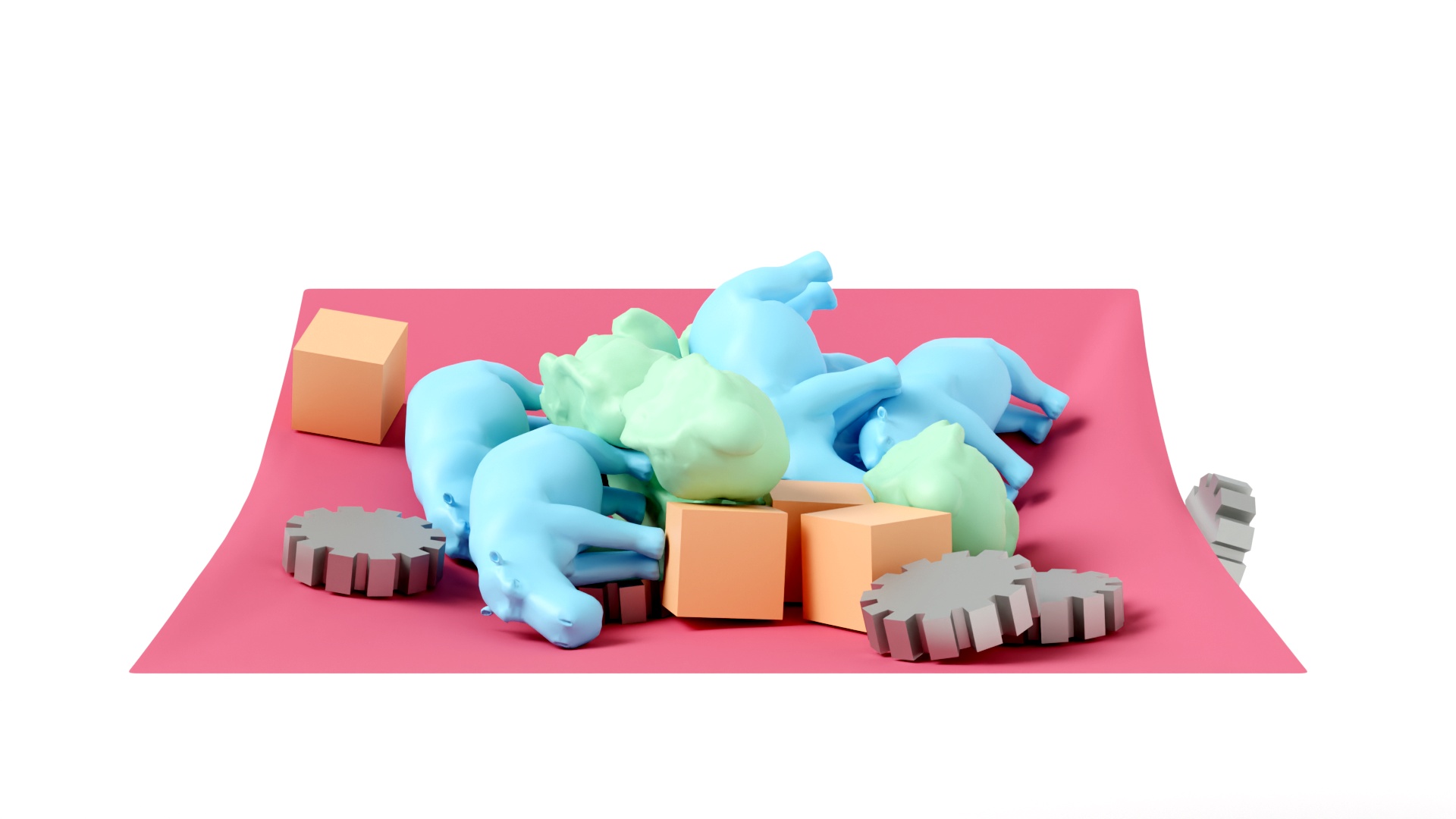}{$t\!=\!2.5$}
\caption{Multi-physics simulation with coupled soft body, cloth, and rigid body falling onto a cloth surface. Our method handles the complex interplay between cloth-soft body, soft body-soft body, and soft body-rigid body contacts while maintaining penetration-free constraints throughout the simulation. ($\Delta t=1/600$s, avg.\ step time: \del{1.1}\add{1.8}ms)}
\label{fig:multiphysics}
\end{figure*}

\autoref{fig:multiphysics} demonstrates a multi-physics scenario where soft body bunnies, deformable cubes, and rigid gears are dropped onto cloth, involving simultaneous cloth-soft, soft-soft, soft-rigid, and cloth-rigid contacts. Our unified approach handles all interaction types with the same Planar-DAT scheme, maintaining penetration-free constraints throughout.
\subsubsection{Qualitative Comparison to OGC (Isotropic-DAT)}

\begin{figure}[t]
\centering
\newcommand{\figtwist}[2]{%
\begin{tikzpicture}
\node[anchor=south west,inner sep=0] (image) at (0,0) {\includegraphics[width=0.24\linewidth,trim=600 150 600 80,clip]{Figures/ClothTwistComparison/#1.jpg}};
\ifx\relax#2\relax\else\node[anchor=south west,fill=white,fill opacity=0.7,text opacity=1,inner sep=2pt] at (image.south west) {\scriptsize #2};\fi
\end{tikzpicture}}
\begin{tabular}{@{}c@{\hspace{2pt}}c@{\hspace{2pt}}c@{\hspace{2pt}}c@{}}
\multicolumn{4}{c}{\scriptsize Isotropic-DAT (OGC)} \\[2pt]
\figtwist{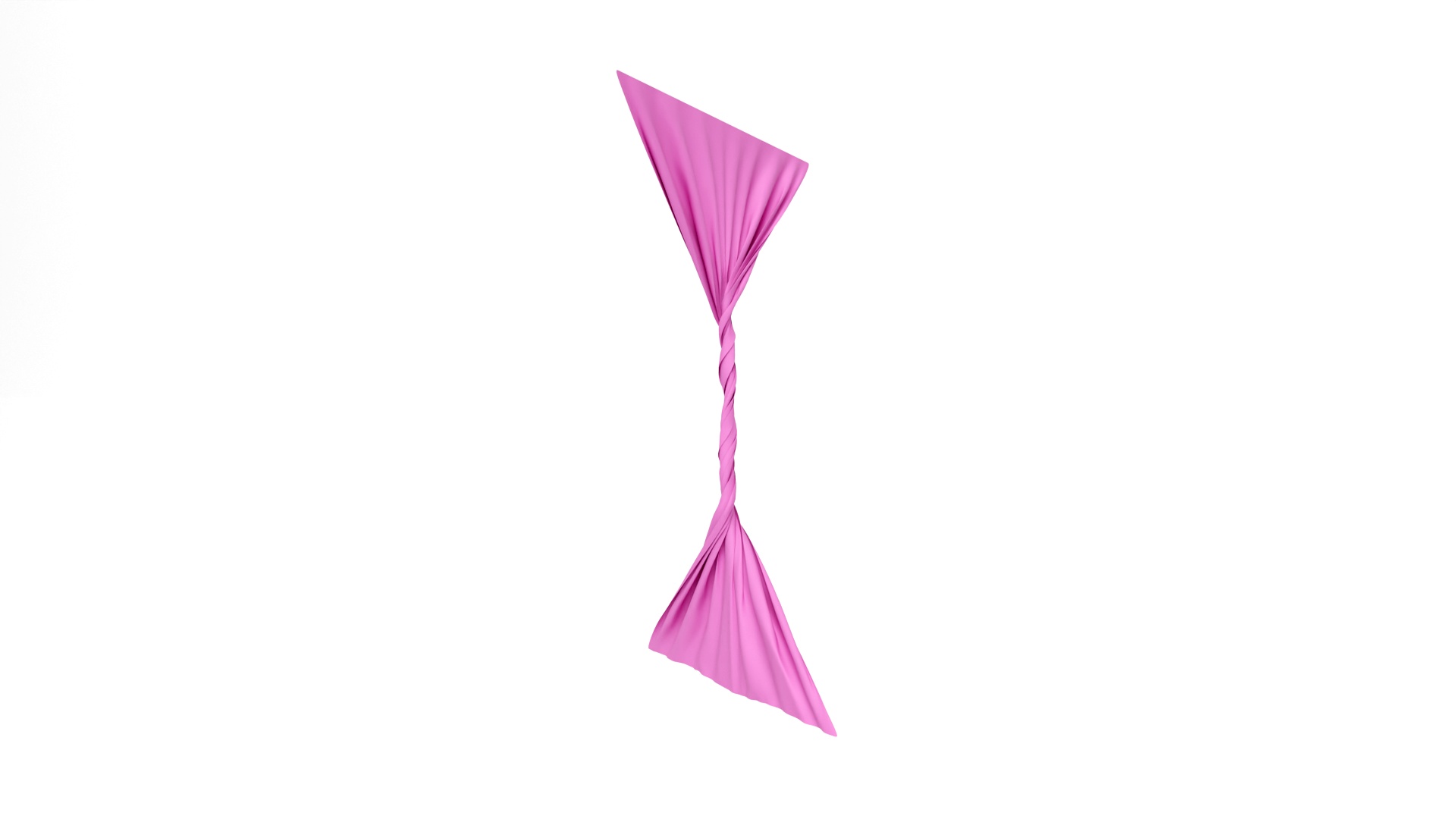}{$t\!=\!10$} & \figtwist{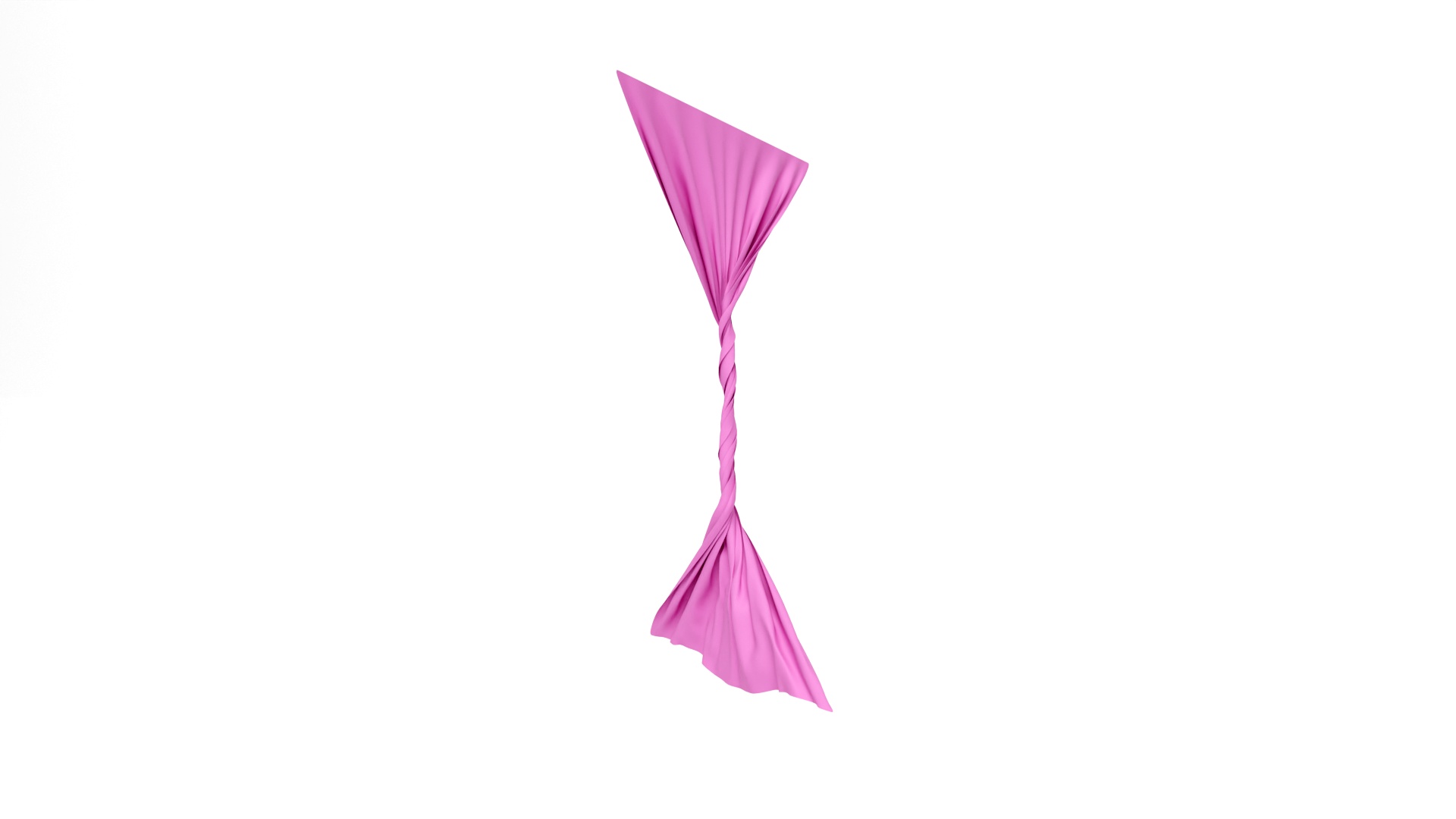}{$t\!=\!10.03$} & \figtwist{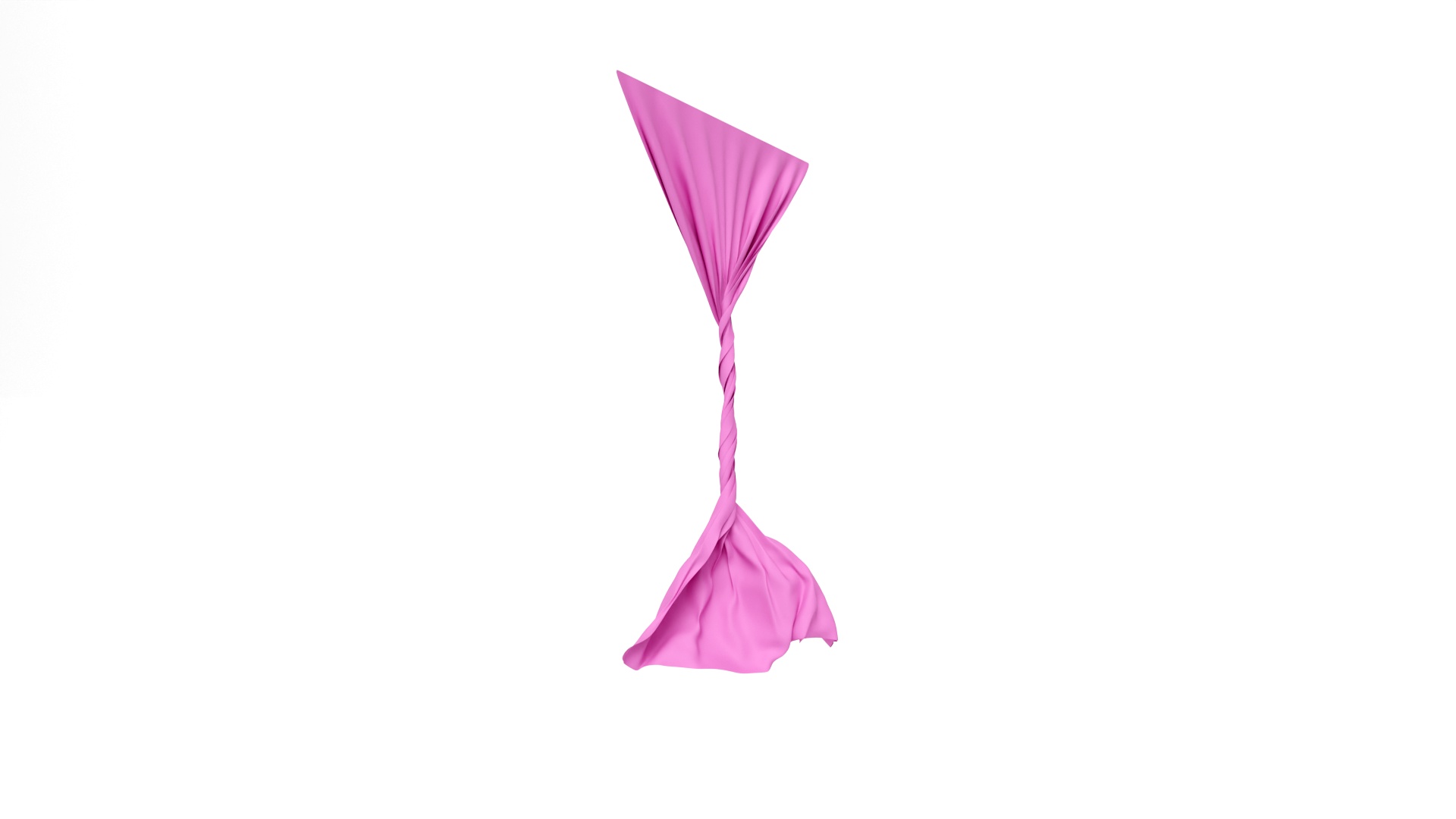}{$t\!=\!11$} & \figtwist{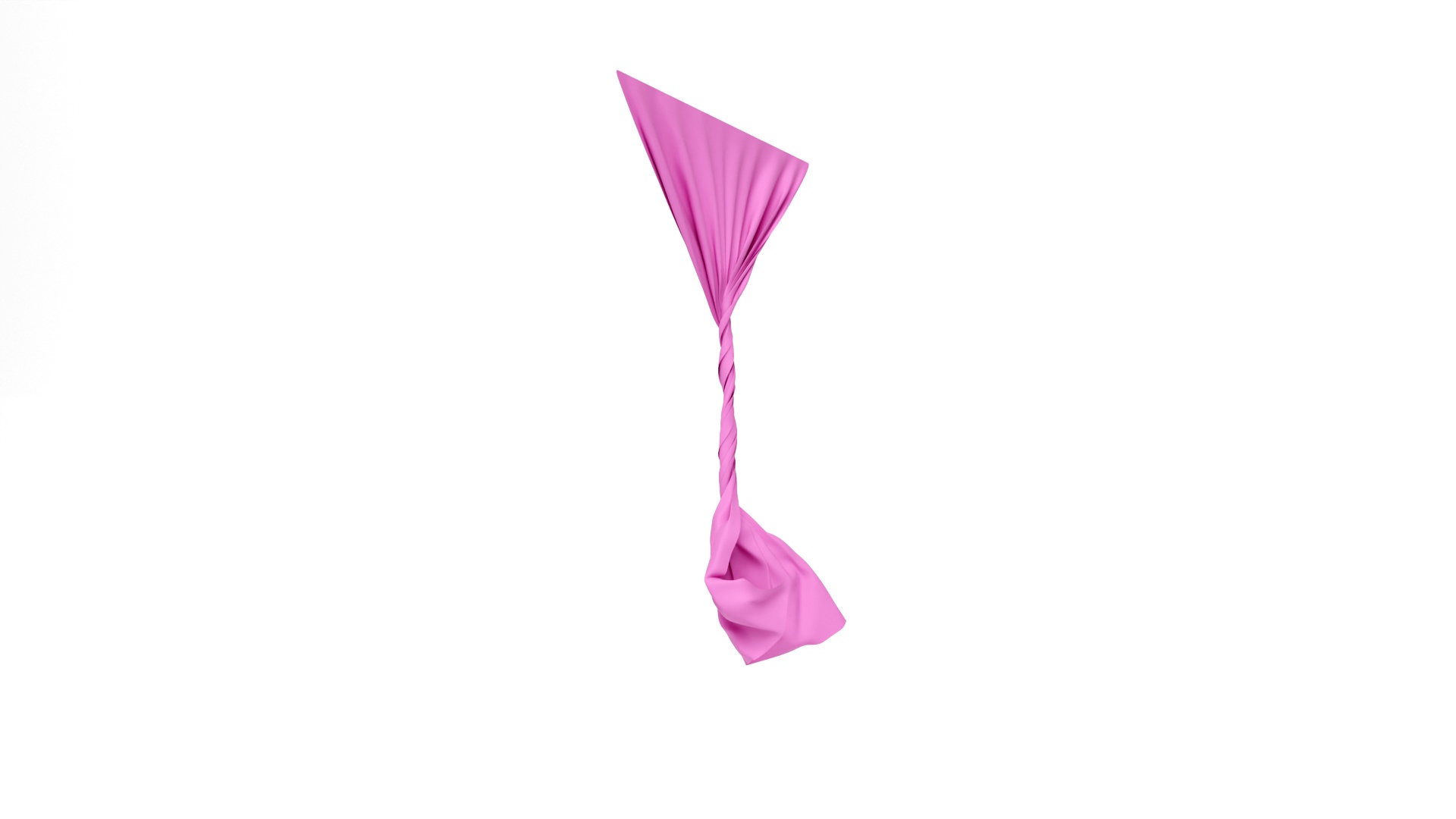}{$t\!=\!12.25$} \\[2pt]
\multicolumn{4}{c}{\scriptsize Planar-DAT (Ours)} \\[2pt]
\figtwist{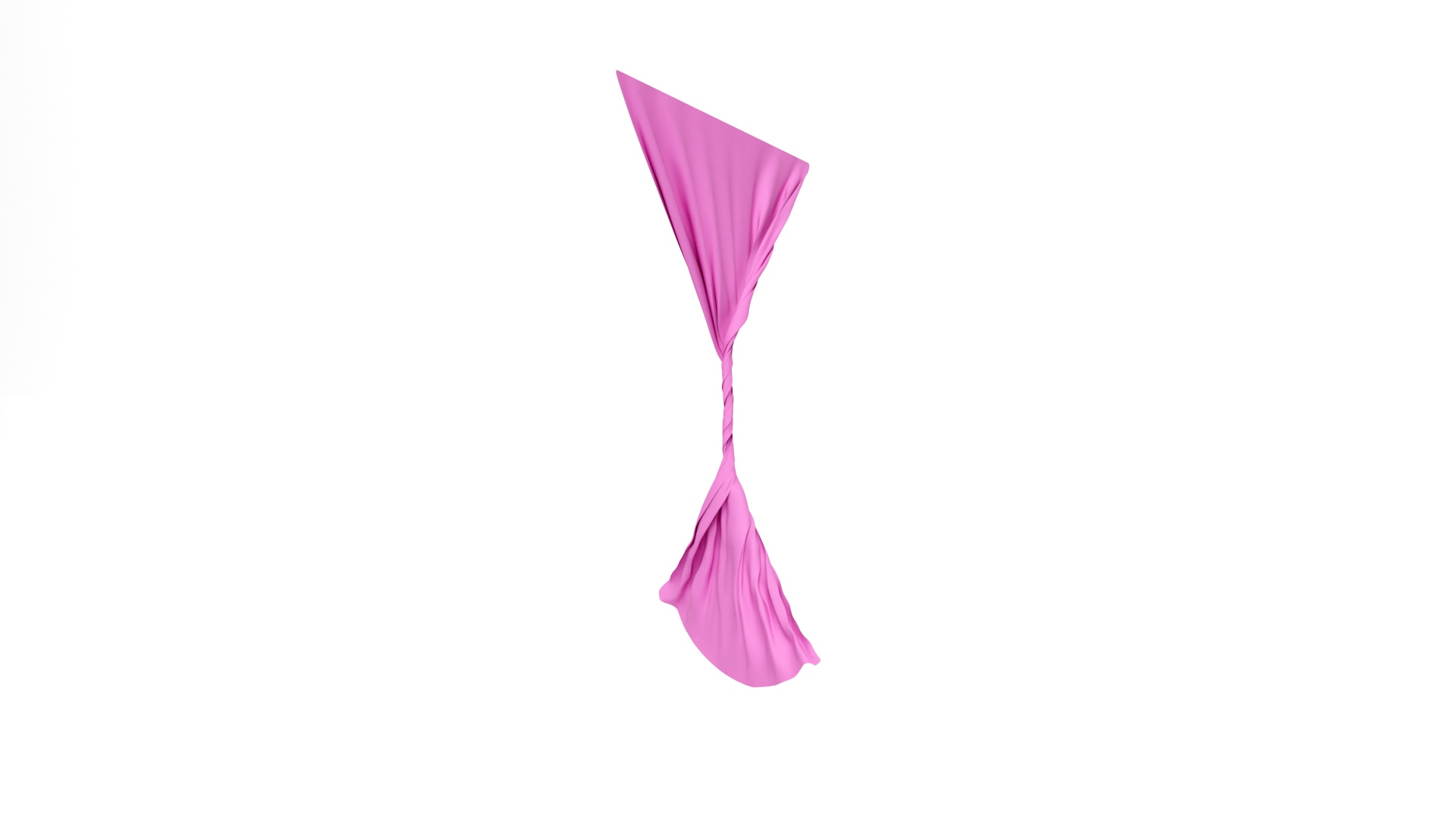}{$t\!=\!10$} & \figtwist{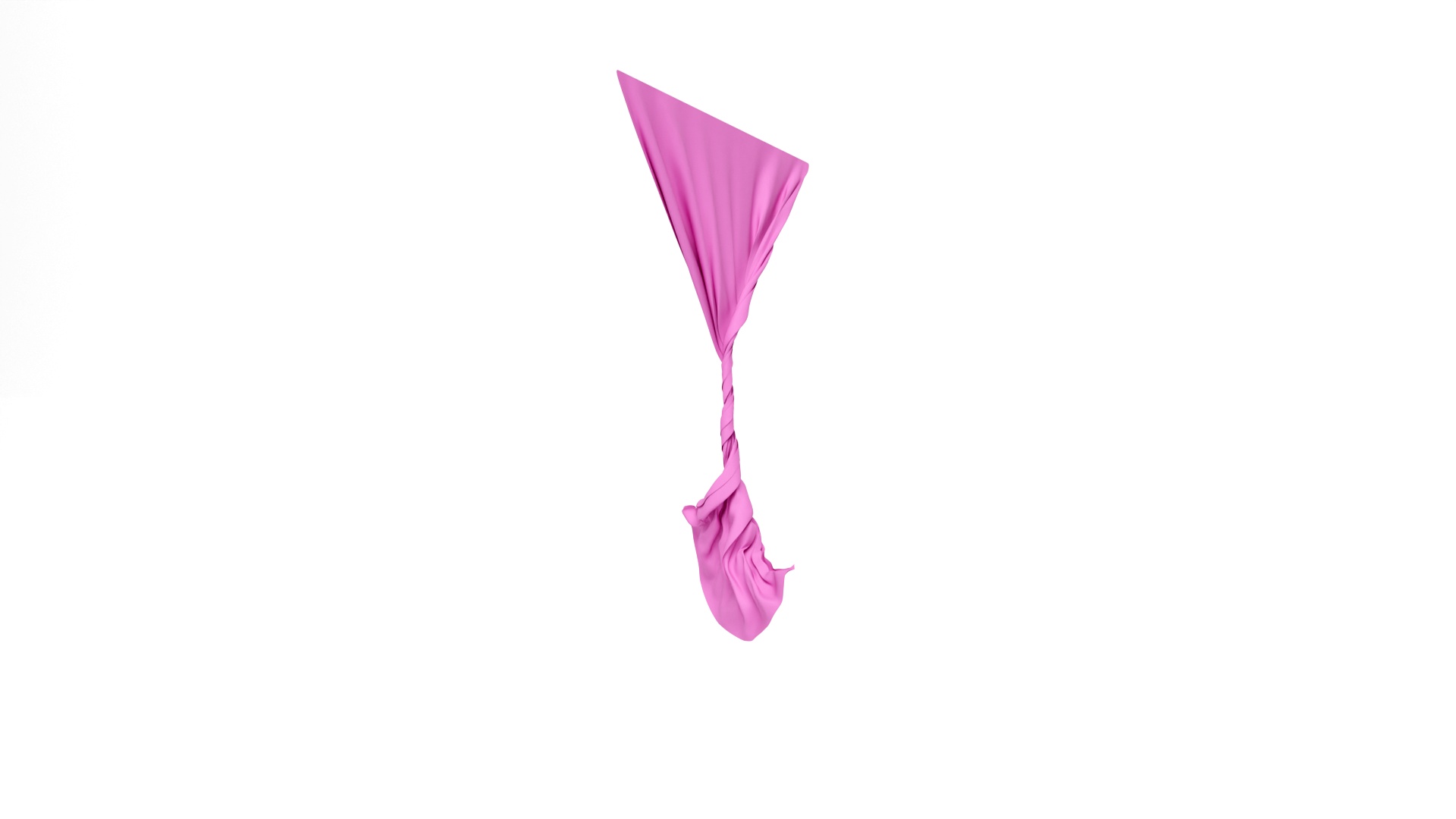}{$t\!=\!10.03$} & \figtwist{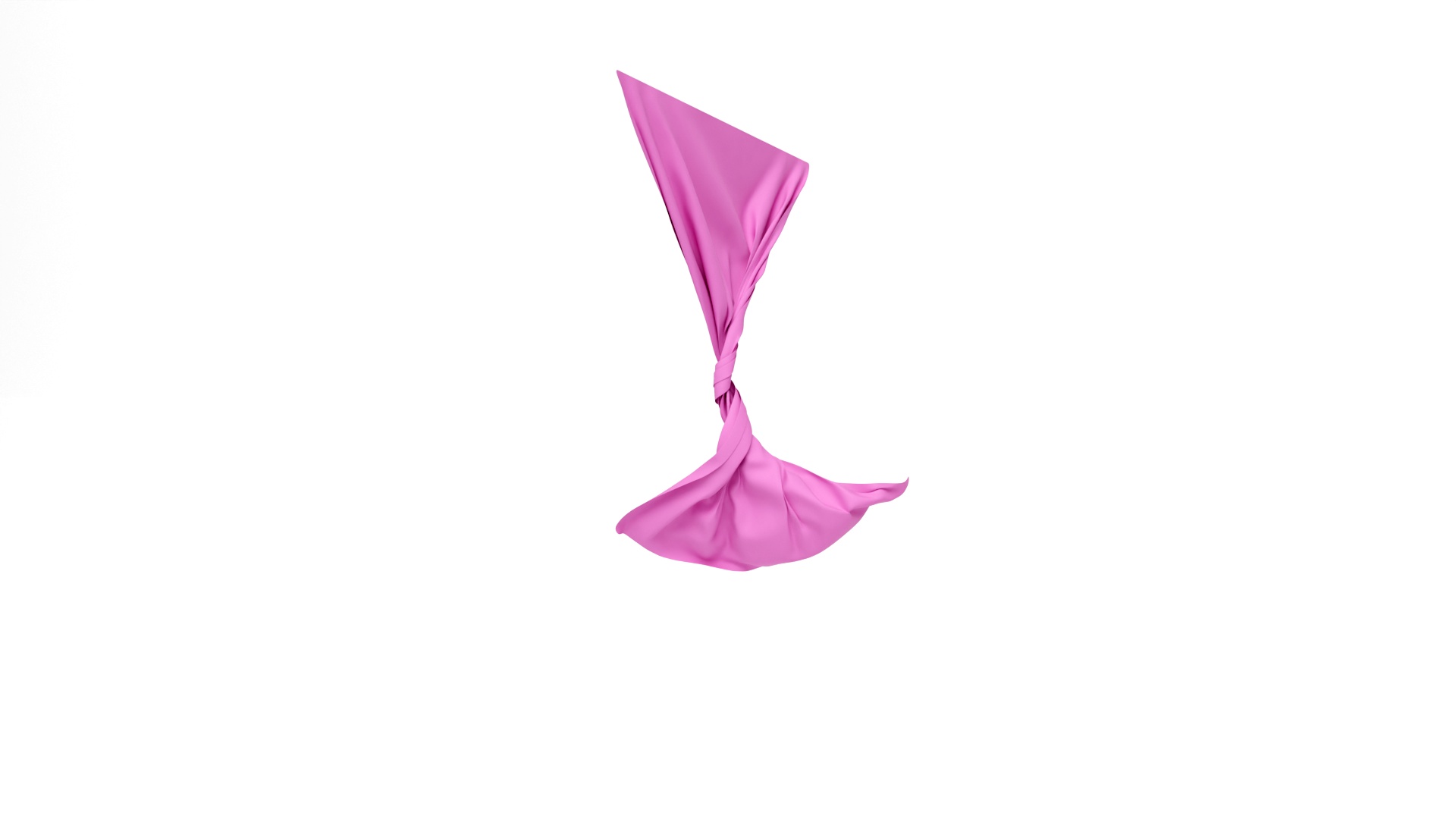}{$t\!=\!11$} & \figtwist{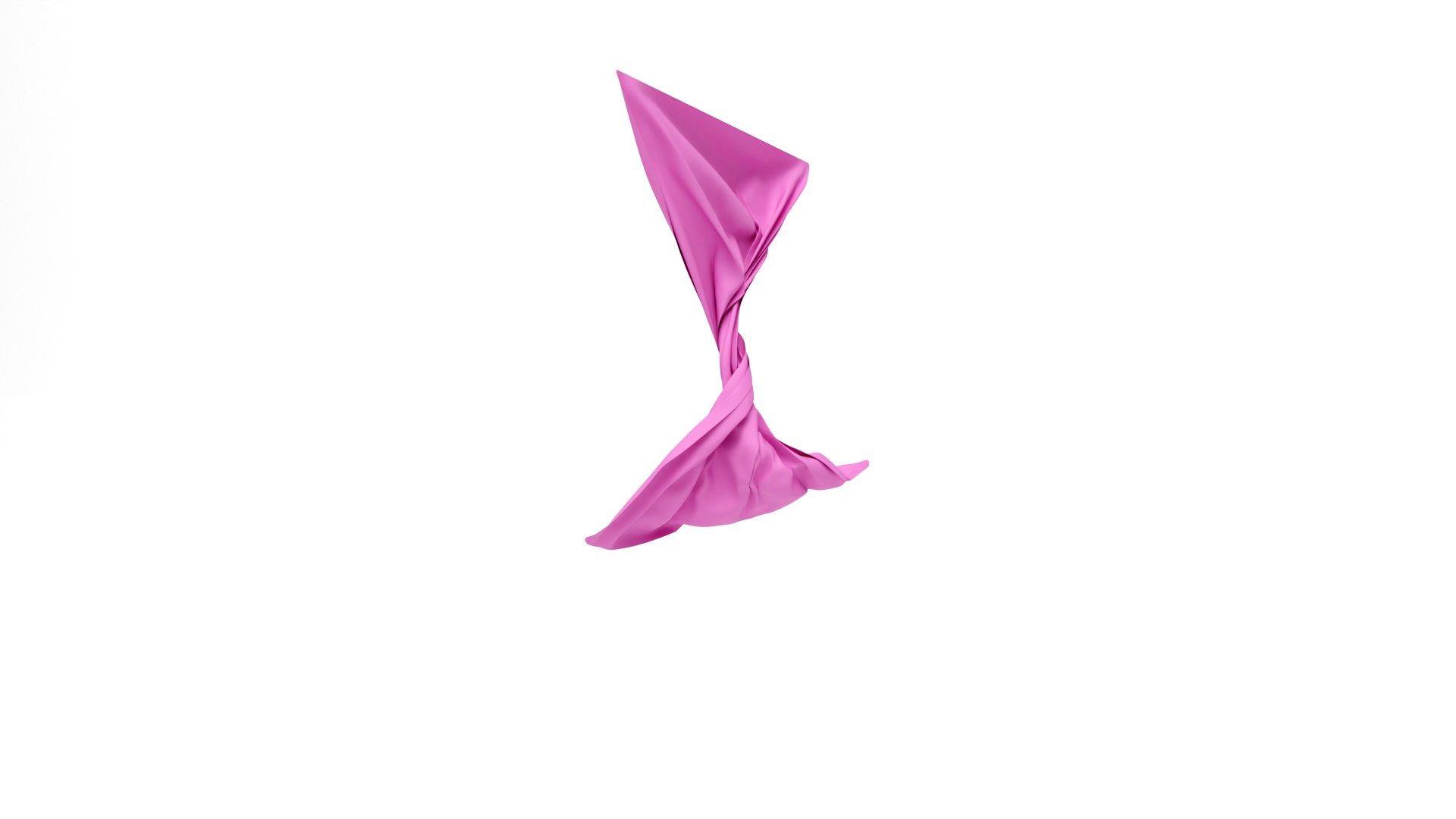}{$t\!=\!12.25$} \\
\end{tabular}
\caption{Cloth twist release comparison. A cloth is twisted 6 full rotations over 10 seconds, then released at $t=10$s. Top: Isotropic-DAT (OGC). Bottom: Planar-DAT (Ours). ($\Delta t=1/600$s, avg.\ step time: Planar-DAT=1.1ms, Isotropic-DAT=1.0ms)}
\label{fig:clothTwistComparison}
\end{figure}

\autoref{fig:clothTwistComparison} compares the rebound behavior after releasing a twisted cloth. Planar-DAT exhibits significantly faster dynamics because its half-space constraints only restrict motion \emph{toward} surfaces, whereas isotropic-DAT's spherical constraints restrict motion in all directions, introducing artificial damping and preventing the twisted cloth from expanding.

\begin{figure*}[t]
\centering
\begin{minipage}{0.85\linewidth}
\centering
\includegraphics[trim=300 0 250 0,clip, width=0.265\linewidth]{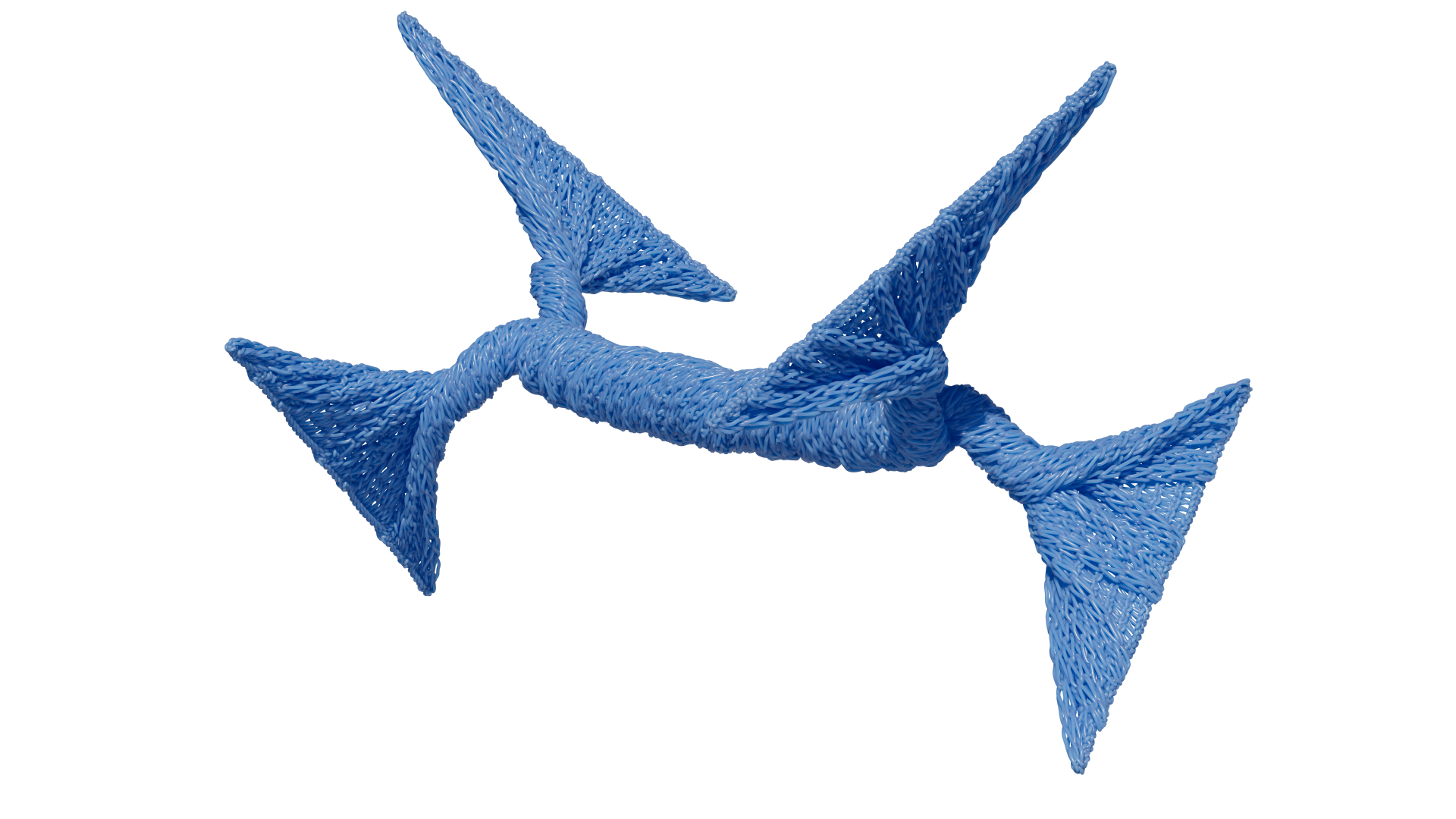}\hspace{-0.002\linewidth}%
\includegraphics[trim=300 0 250 0,clip, width=0.265\linewidth]{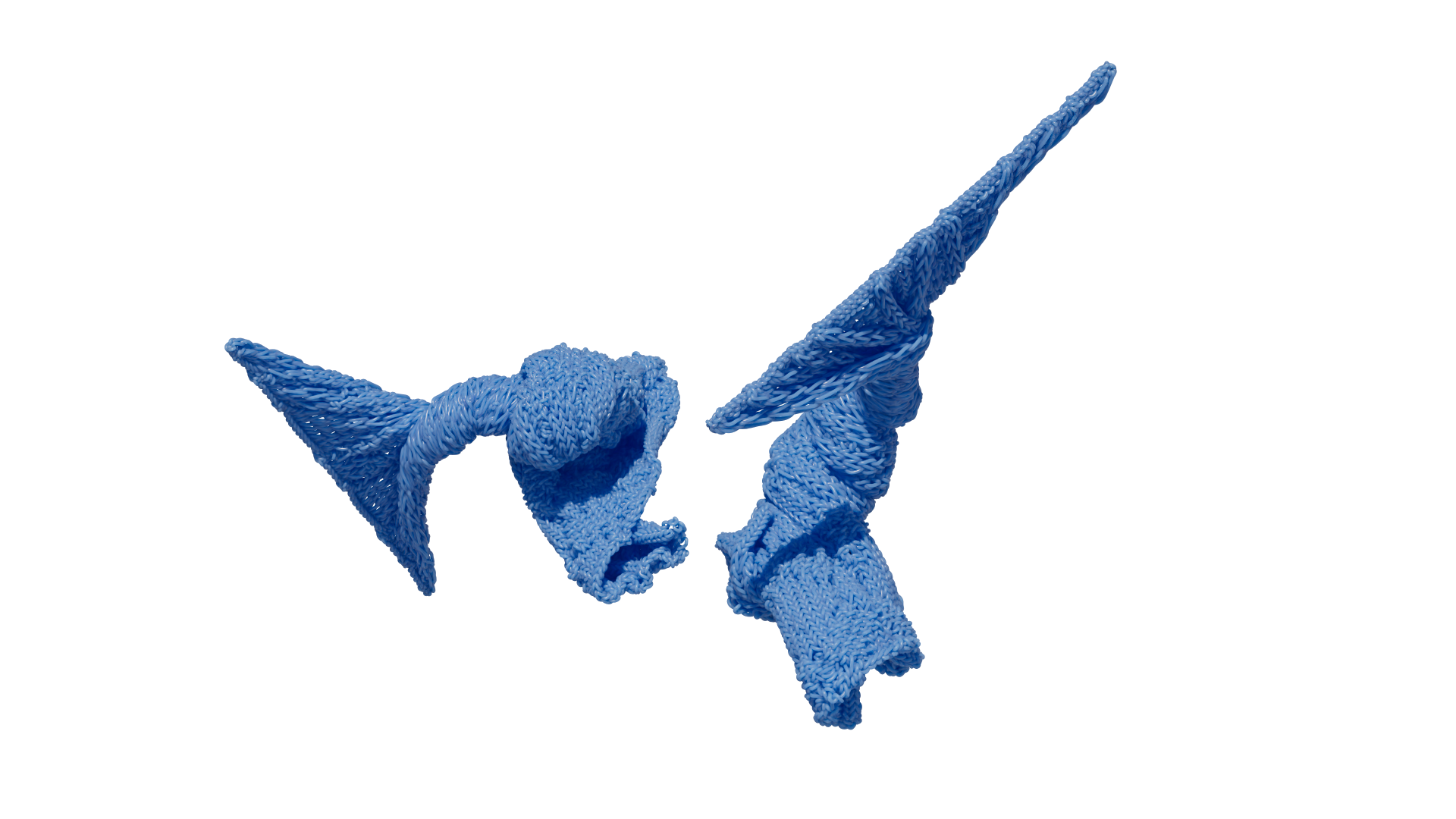}\hspace{-0.061\linewidth}%
\includegraphics[trim=300 0 250 0,clip, width=0.265\linewidth]{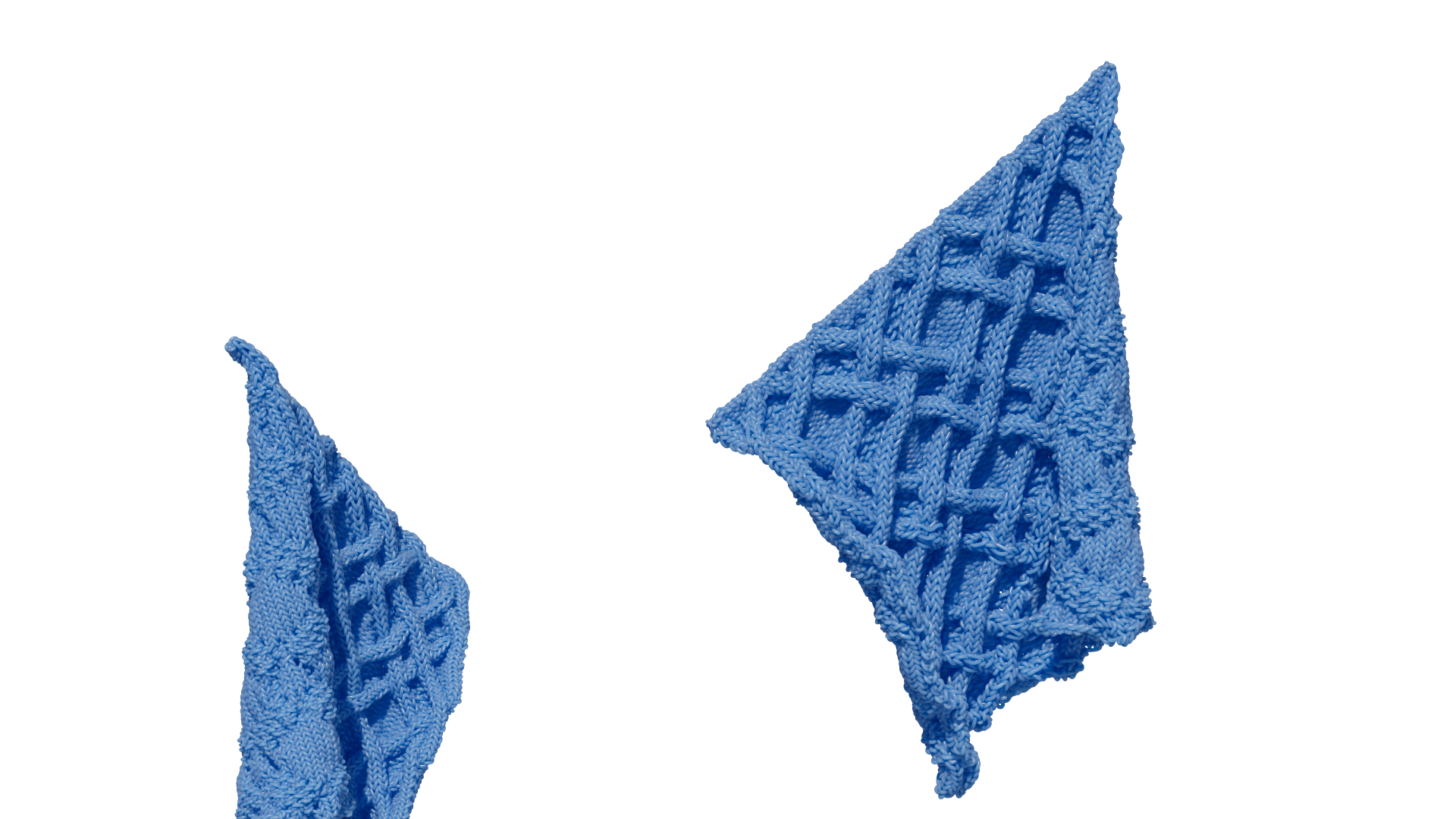}\hspace{-0.022\linewidth}%
\vrule width 0.5pt height 3.0cm depth 0pt\hspace{0.011\linewidth}%
\includegraphics[trim=300 0 250 0,clip, width=0.265\linewidth]{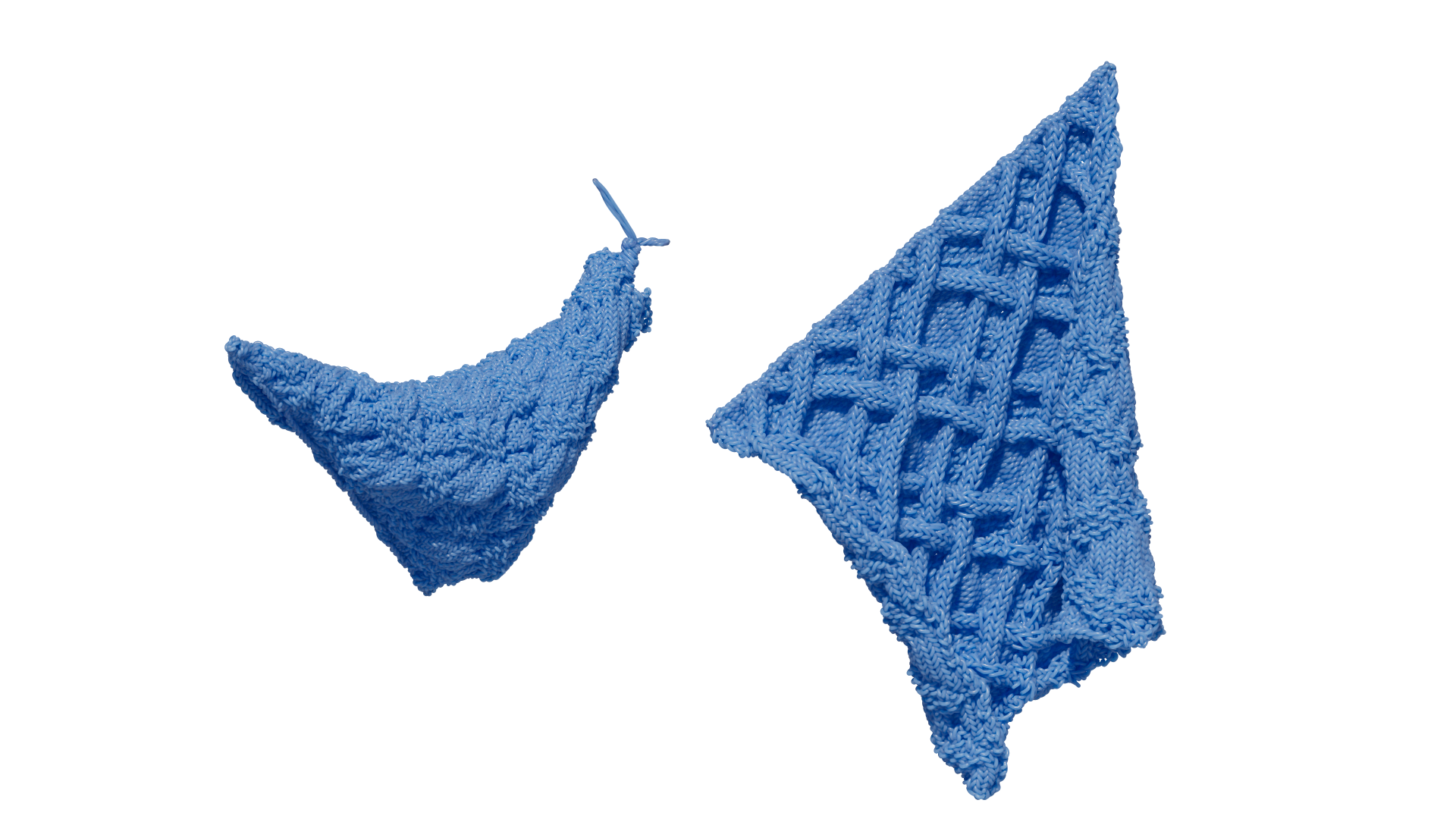}\\[-1.5em]
\makebox[0.778\linewidth][c]{\scriptsize Planar-DAT (Ours)}%
\makebox[0.167\linewidth][c]{\scriptsize Isotropic-DAT (OGC)}\hfill\\[-0.4em]
\end{minipage}
\caption{Isotropic-DAT can exhibit significant locking artifacts while our method successfully resolves the close contacts. When released, our method also produces significantly less artificial damping. ($\Delta t=1/2520$s, avg.\ step time: Planar-DAT=0.45ms, Isotropic-DAT=0.39ms)}
\Description{}
\label{fig:yarnTwist}
\end{figure*}

We further test the damping of close contacts in \autoref{fig:yarnTwist} by twisting two pieces of yarn cloth together. We only use edge segment collisions with yarn simulated using Cosserat rods \cite{hsu2025}. The close knitted nature of yarn cloth presents a significant challenge for isotropic-DAT, the use of which results in unstable deadlock artifacts. Despite the tight yarn-yarn contacts in all directions, our planar-DAT is able to correctly resolve the sliding and twisting when two of the four cloth ends are released. 

\begin{figure*}[t]
\centering
\newcommand{\figtread}[2]{%
\begin{tikzpicture}
\node[anchor=south west,inner sep=0] (image) at (0,0) {\includegraphics[width=0.19\linewidth,trim=280 300 200 250,clip]{Figures/Treadmill/#1.jpg}};
\ifx\relax#2\relax\else\node[anchor=south west,fill=white,fill opacity=0.7,text opacity=1,inner sep=2pt] at (image.south west) {\scriptsize #2};\fi
\end{tikzpicture}}
\begin{tabular}{@{}c@{\hspace{2pt}}c@{\hspace{2pt}}c@{\hspace{2pt}}c@{\hspace{2pt}}c@{}}
\multicolumn{5}{c}{\scriptsize Isotropic-DAT (OGC)} \\[-2pt]
\figtread{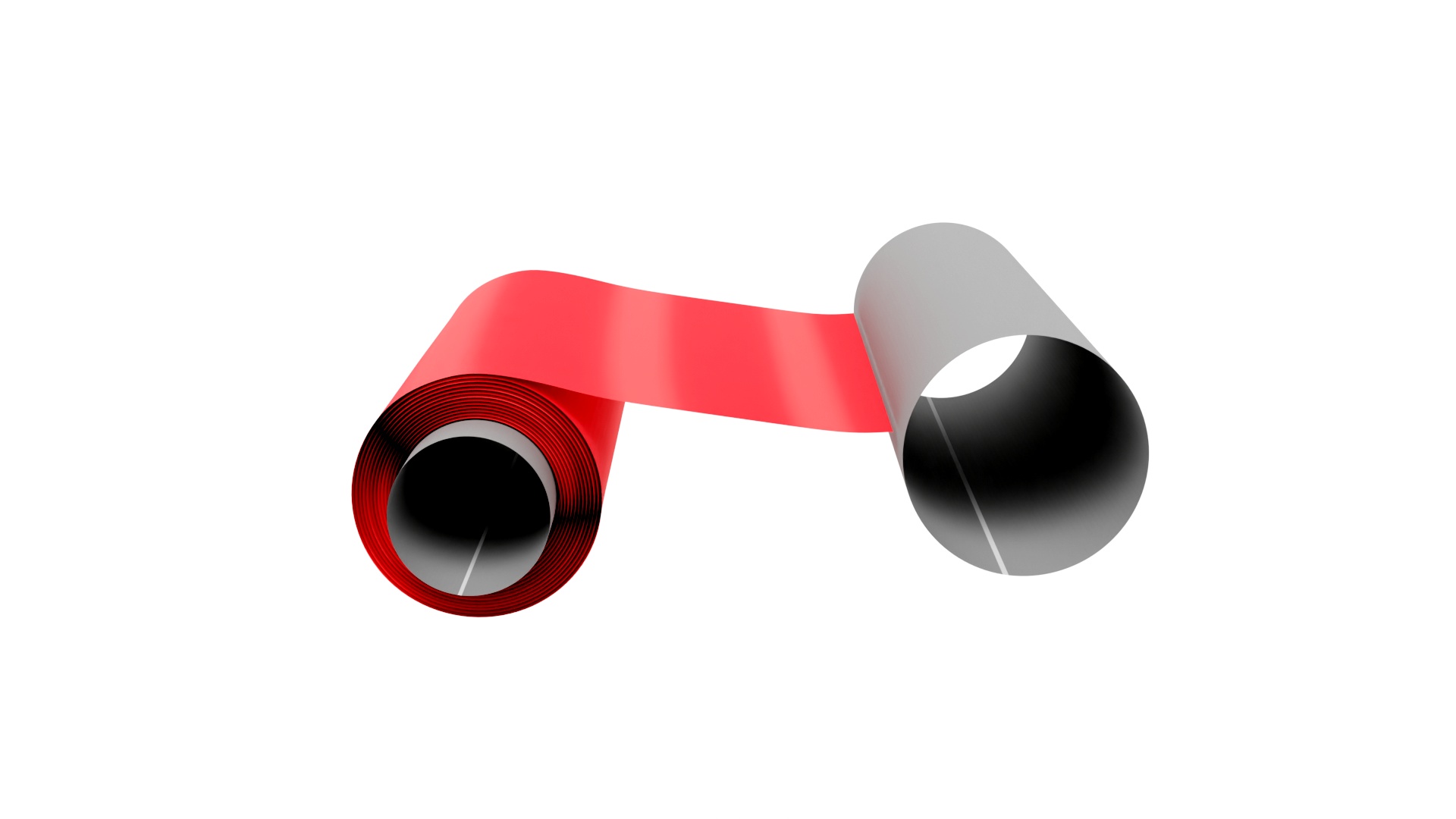}{} & \figtread{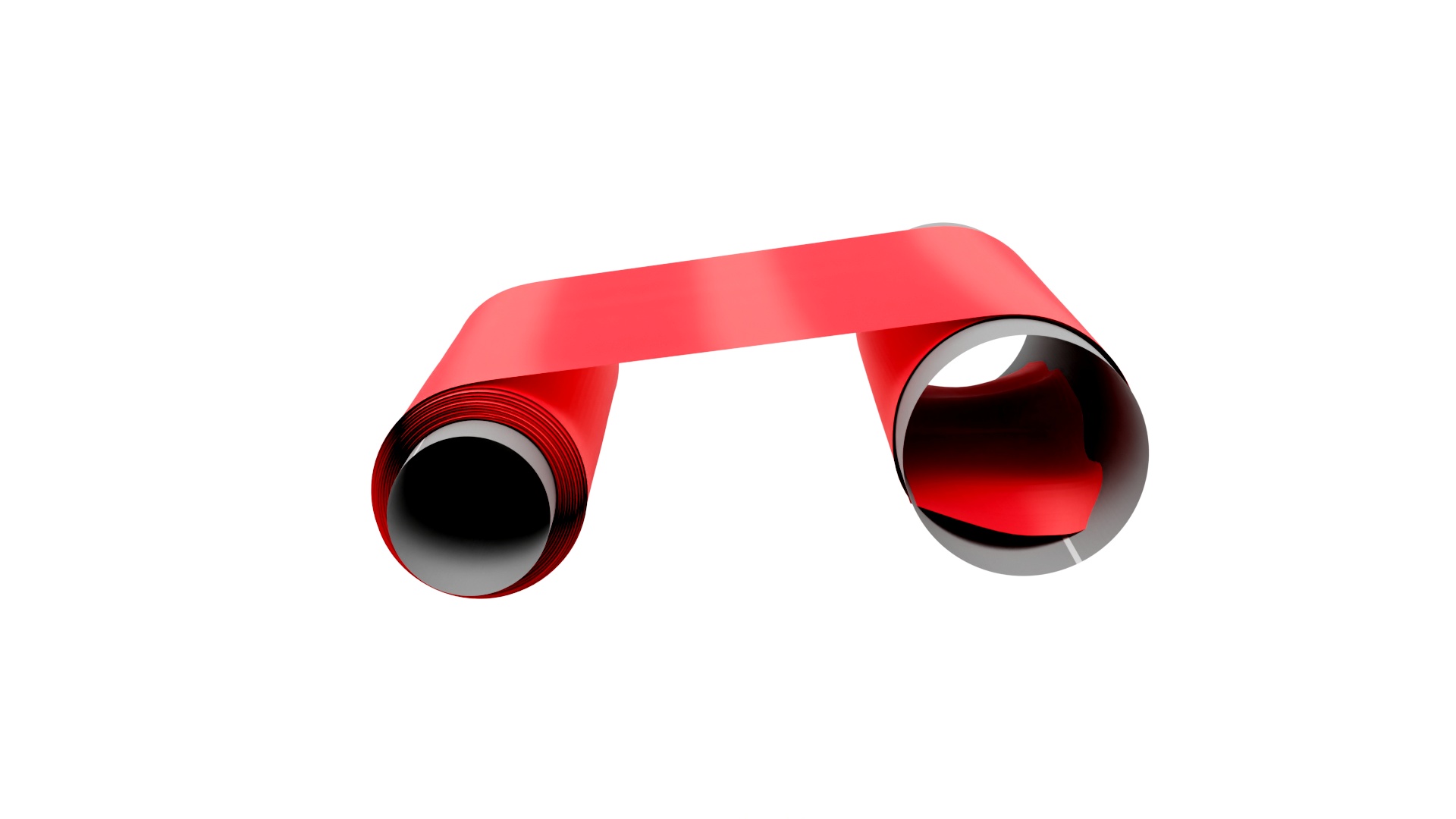}{} & \figtread{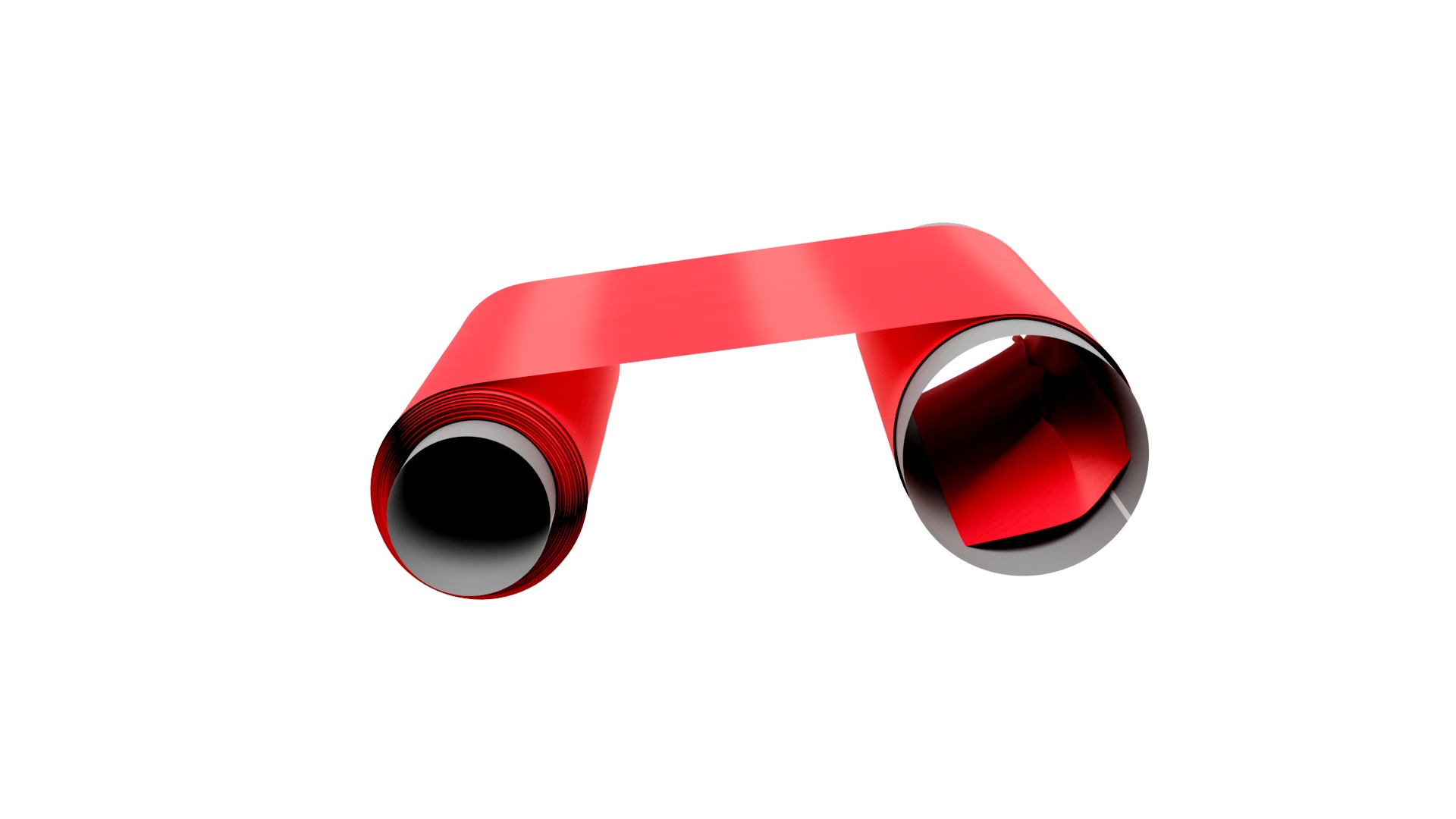}{} & \figtread{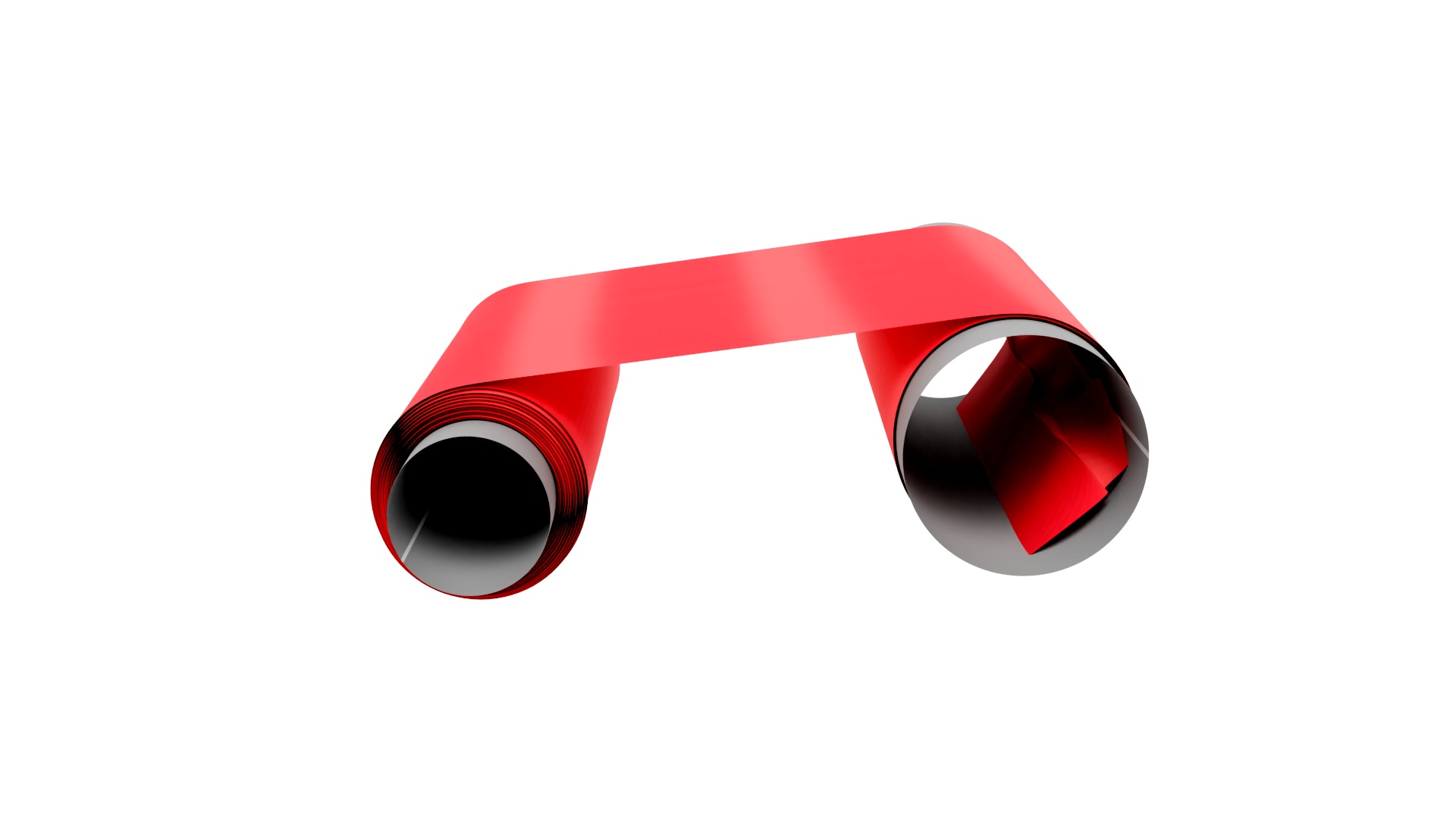}{} & \figtread{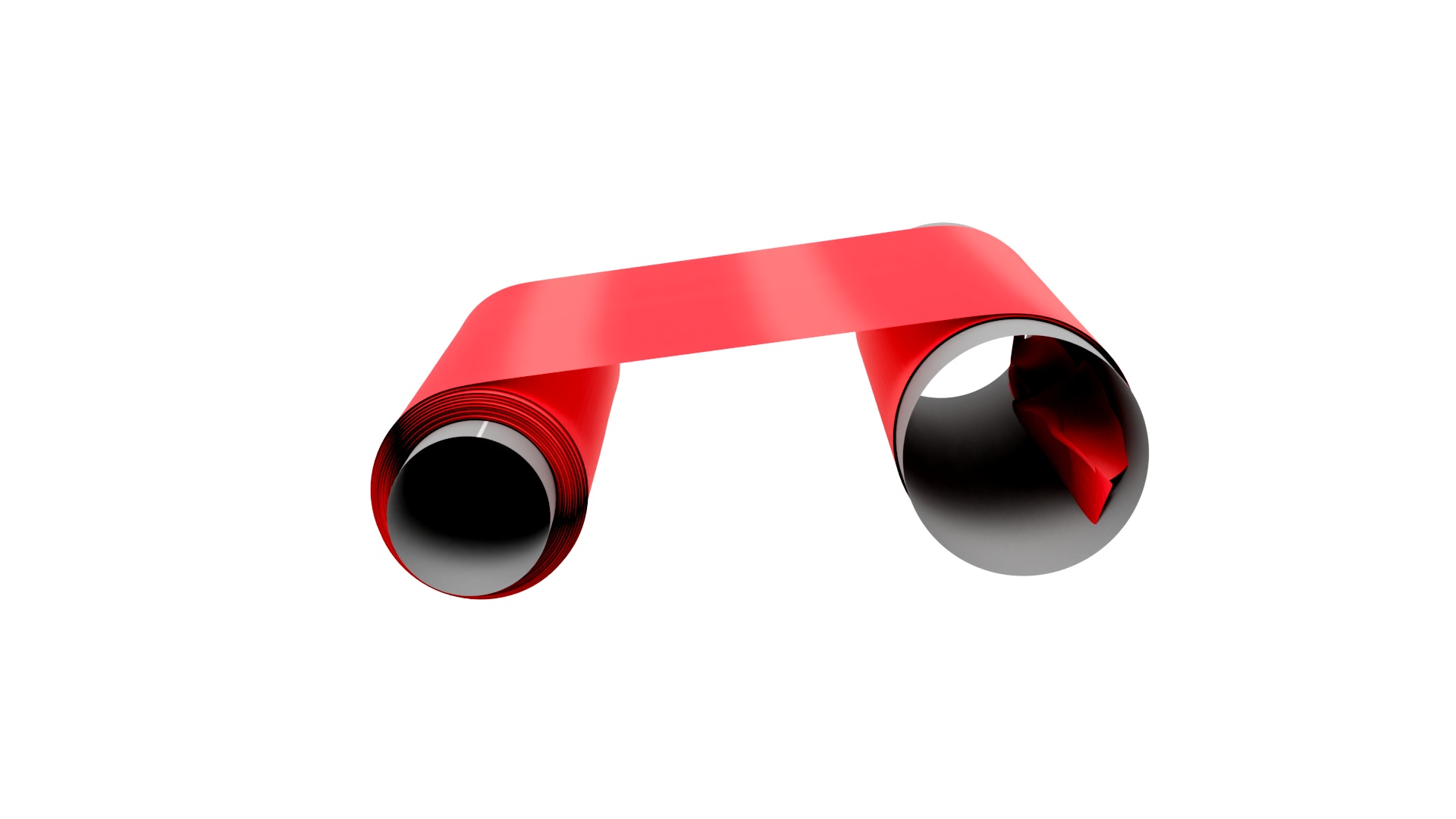}{} \\[2pt]
\multicolumn{5}{c}{\scriptsize Planar-DAT (Ours)} \\[-2pt]
\figtread{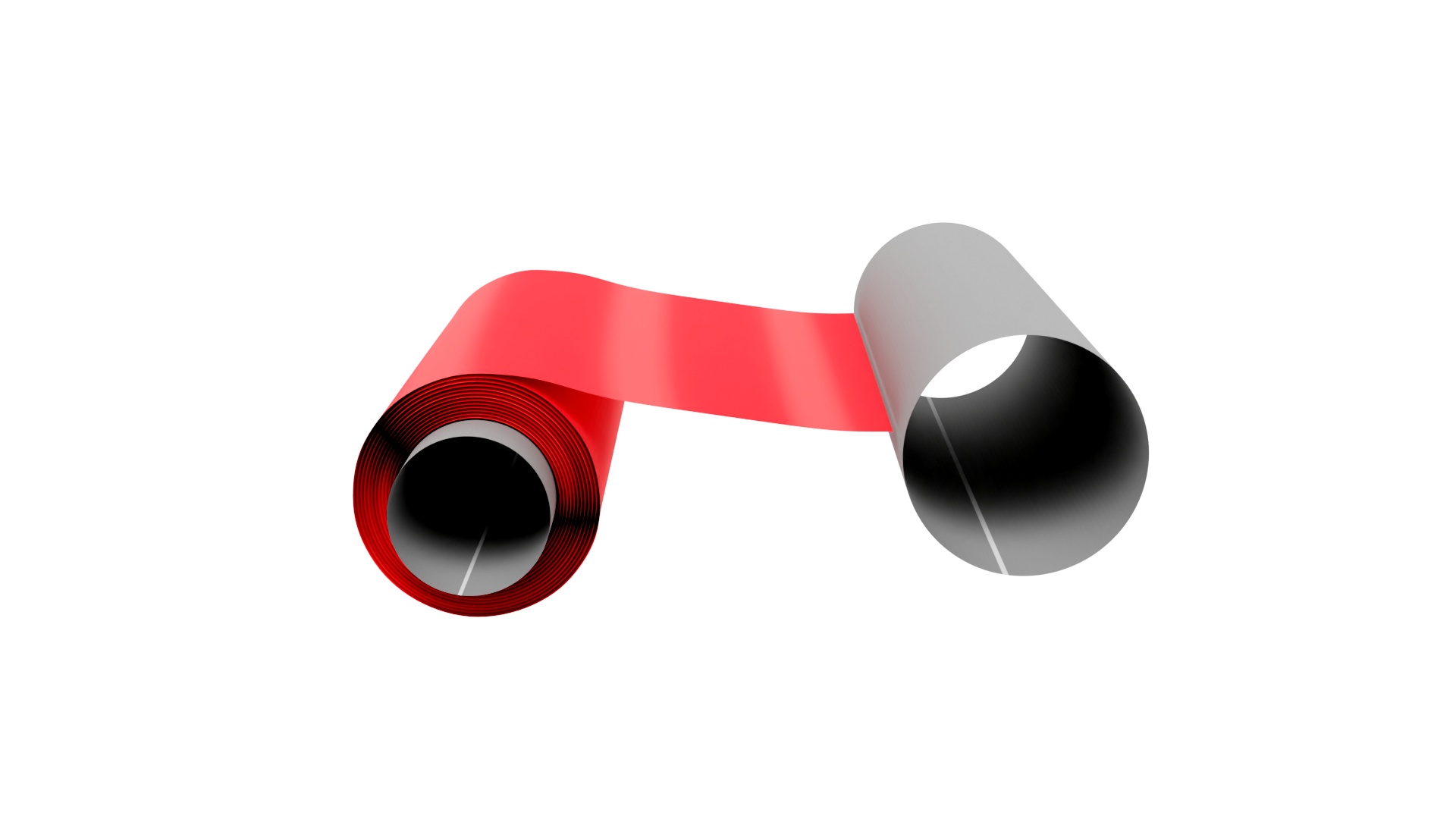}{$t\!=\!0$} & \figtread{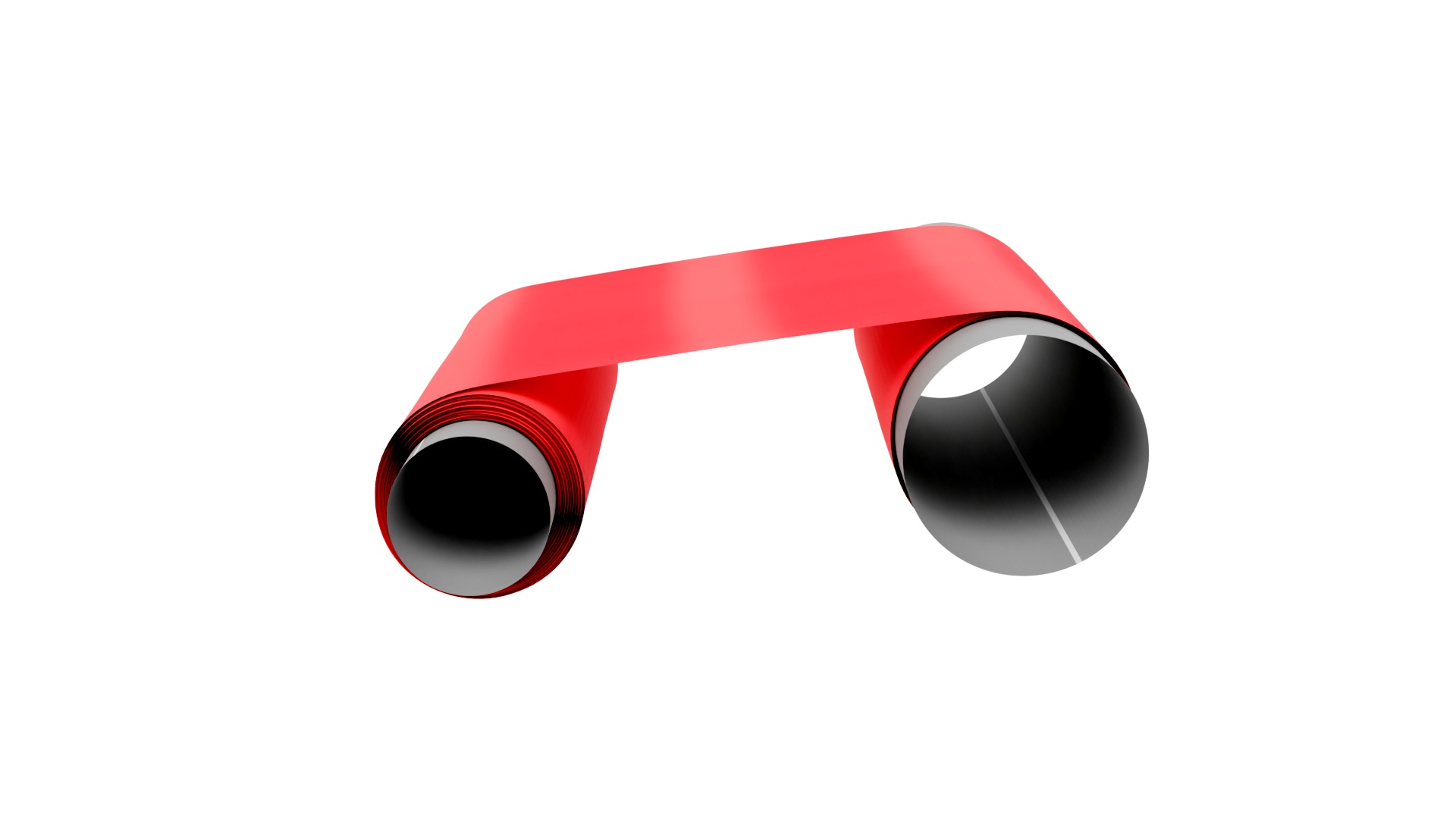}{$t\!=\!4.4$} & \figtread{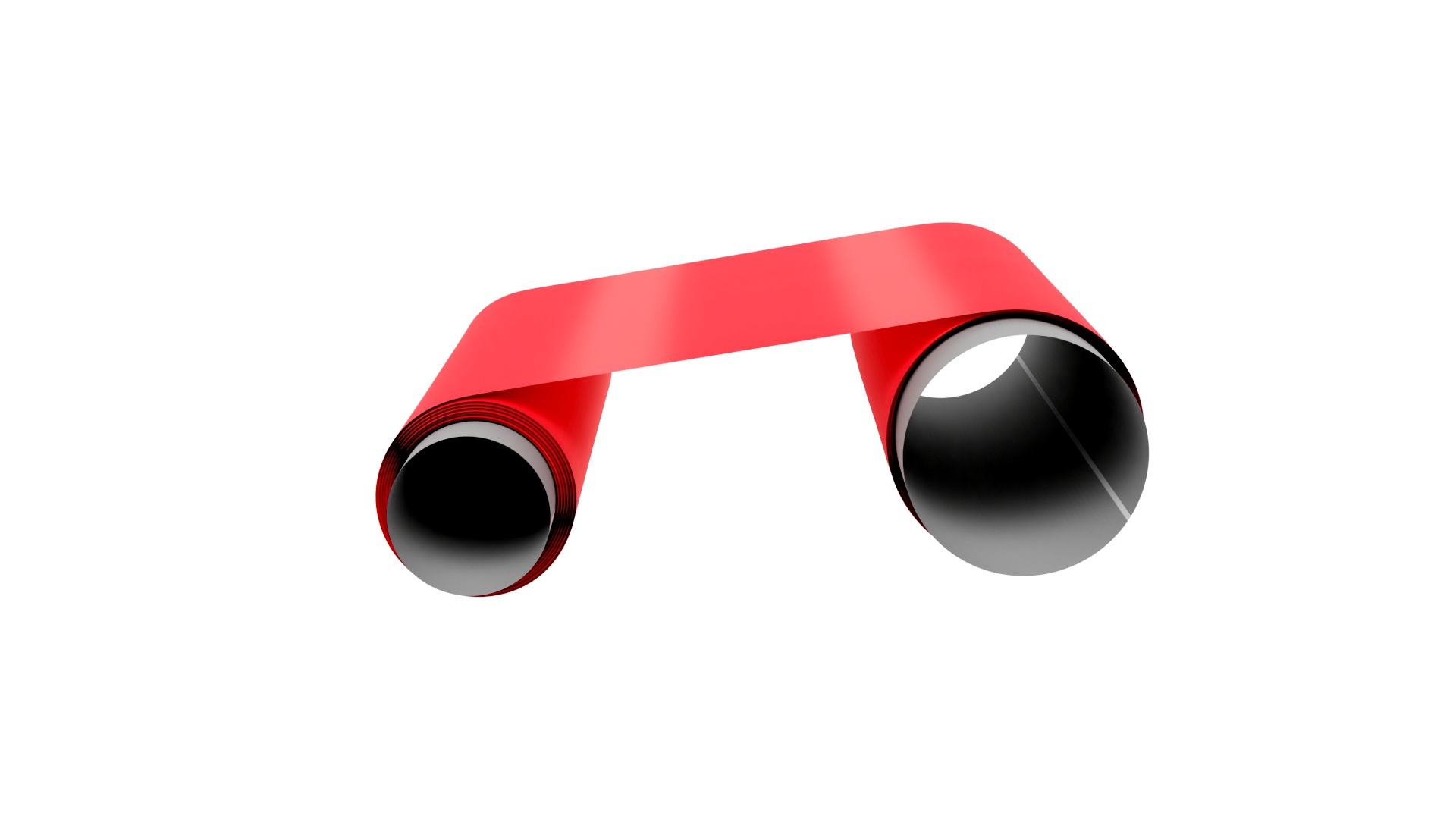}{$t\!=\!8.75$} & \figtread{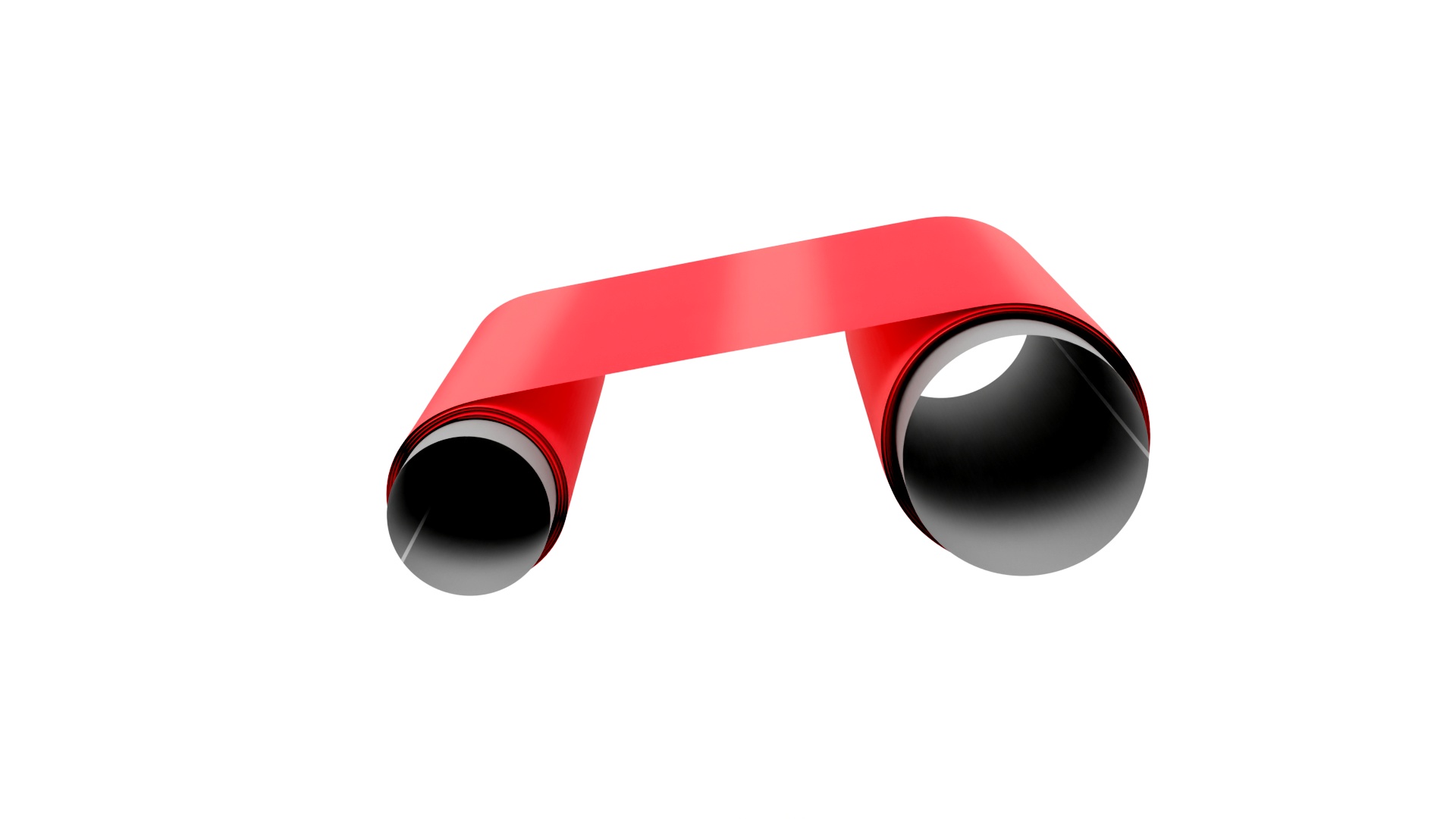}{$t\!=\!13.1$} & \figtread{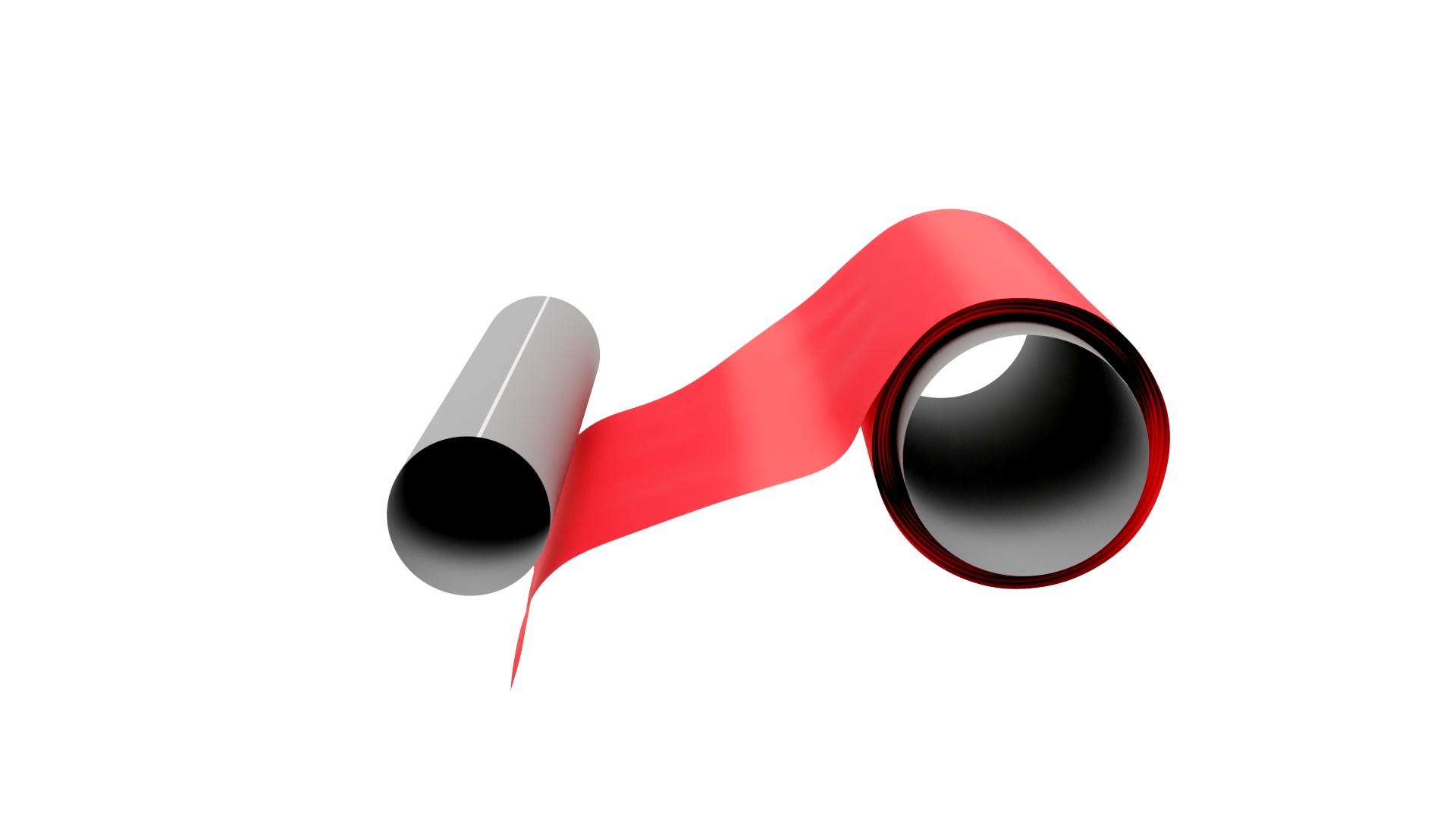}{$t\!=\!17.5$} \\
\end{tabular}
\caption{Treadmill cloth unrolling. A rolled cloth mesh is wound around a smaller rotating cylinder and attached at its edge to a larger rotating cylinder. As both cylinders rotate with matched surface velocity, the cloth unrolls. Top: Isotropic-DAT. Bottom: Planar-DAT. ($\Delta t=1/600$s, avg.\ step time: Planar-DAT=0.55ms, Isotropic-DAT=0.5ms)}
\label{fig:treadmill}
\end{figure*}

\autoref{fig:treadmill} shows a cloth drawn off a spiral winding by rotating cylinders, maintaining dense self-contact as layers slide past each other. Planar-DAT enables smooth layer separation and tangential motion, while Isotropic-DAT's spherical constraints resist sliding and can cause deadlocks, preventing full unwinding.

\begin{figure}[t]
\centering
\newcommand{\figunroll}[2]{%
\begin{tikzpicture}
\node[anchor=south west,inner sep=0] (image) at (0,0) {\includegraphics[width=0.32\linewidth,trim=300 200 300 200,clip]{Figures/UnrollCloth/#1.jpg}};
\ifx\relax#2\relax\else\node[anchor=south west,fill=white,fill opacity=0.7,text opacity=1,inner sep=2pt] at (image.south west) {\scriptsize #2};\fi
\end{tikzpicture}}
\begin{tabular}{@{}c@{\hspace{2pt}}c@{\hspace{2pt}}c@{}}
\multicolumn{3}{c}{\scriptsize Isotropic-DAT (OGC)} \\[2pt]
\figunroll{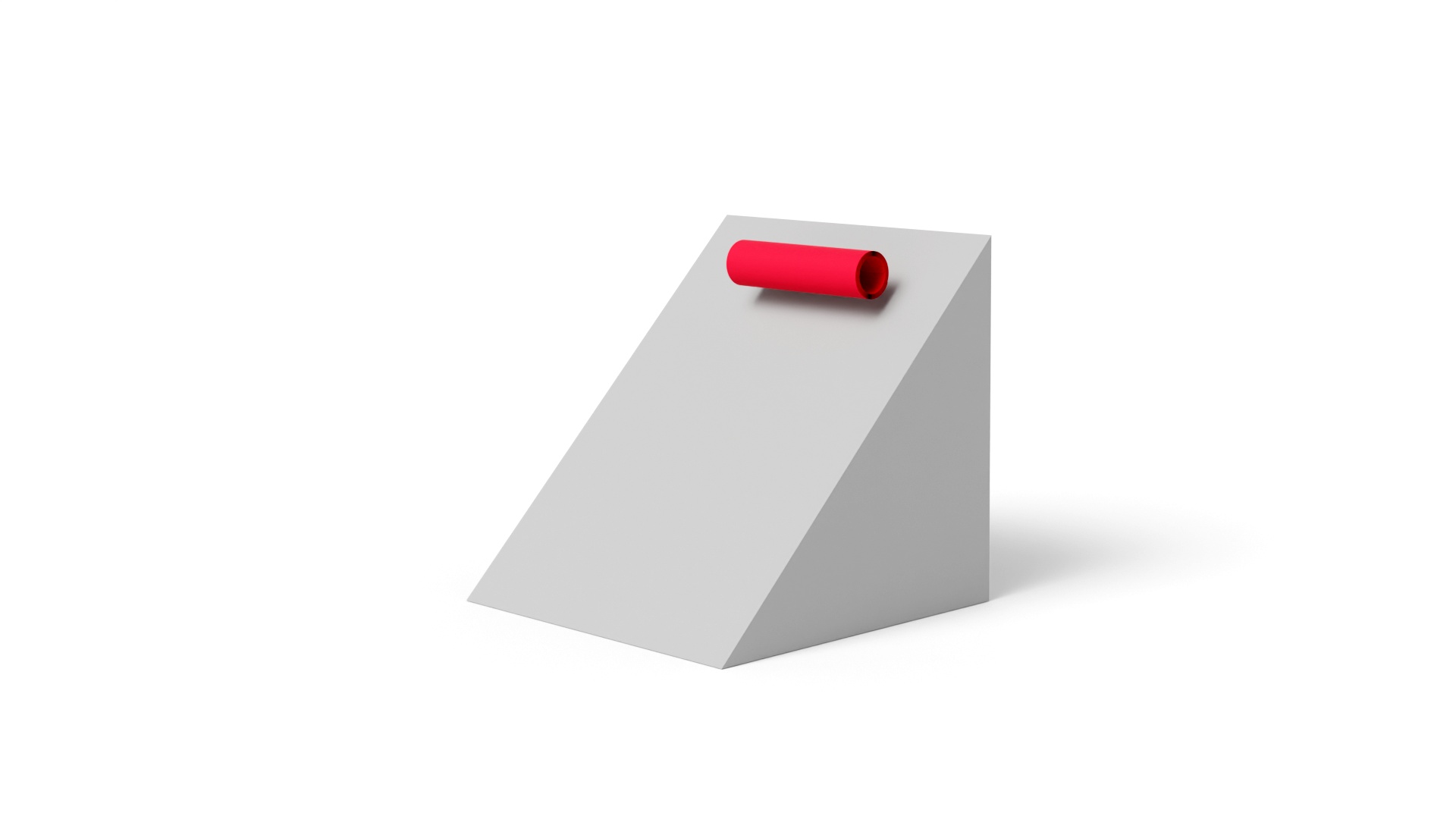}{} & \figunroll{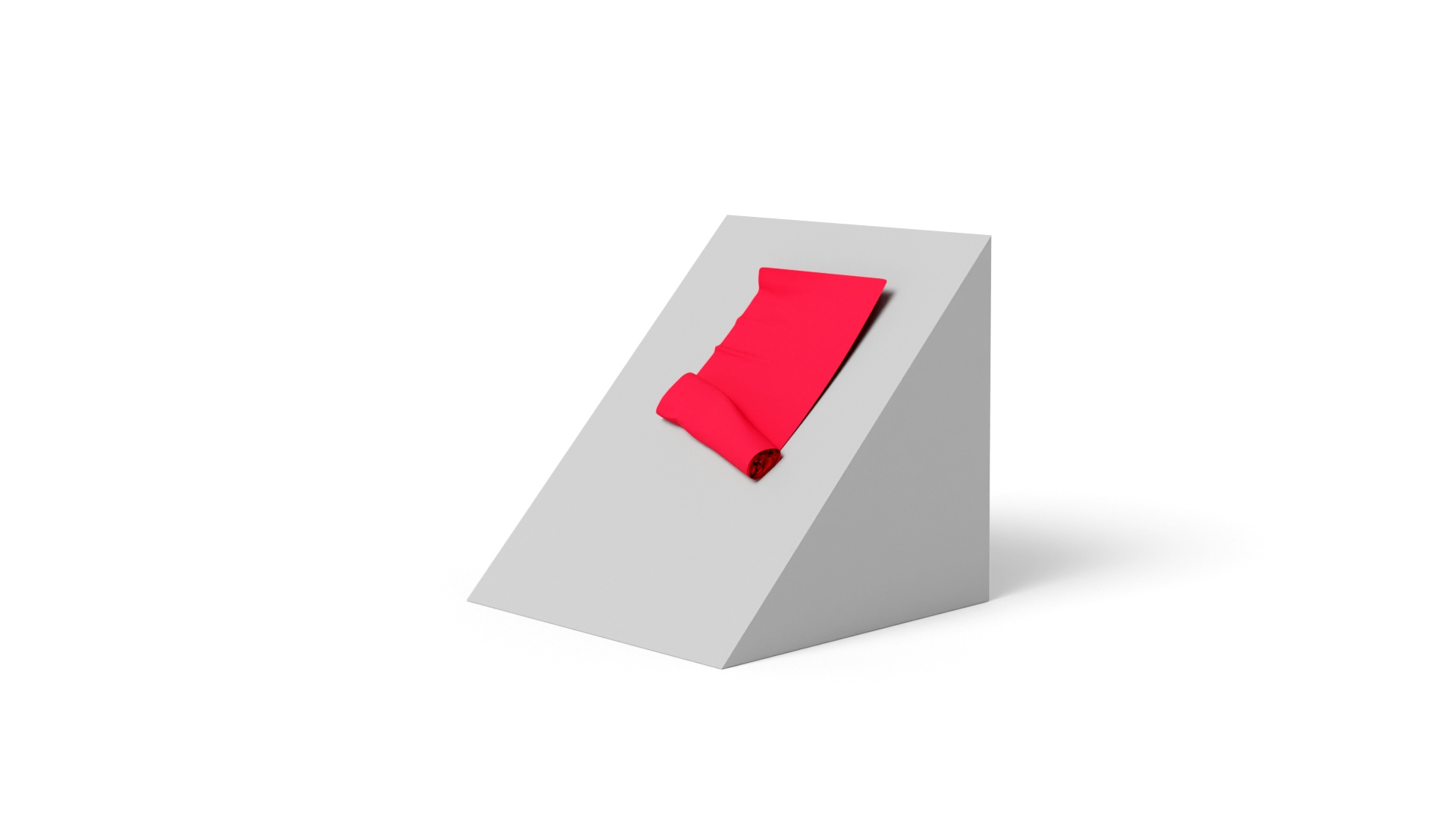}{} & \figunroll{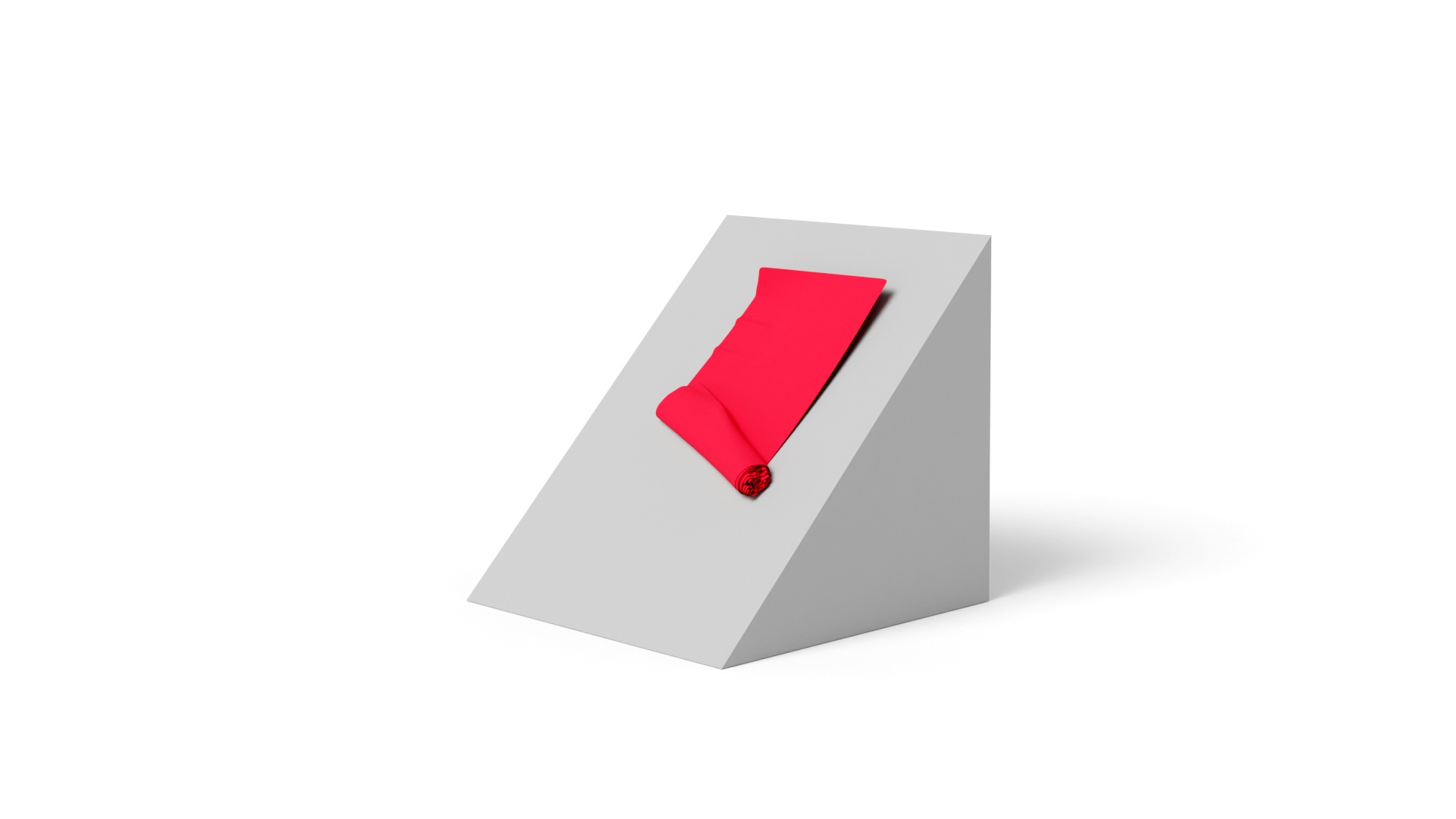}{} \\[2pt]
\multicolumn{3}{c}{\scriptsize Planar-DAT (Ours)} \\[2pt]
\figunroll{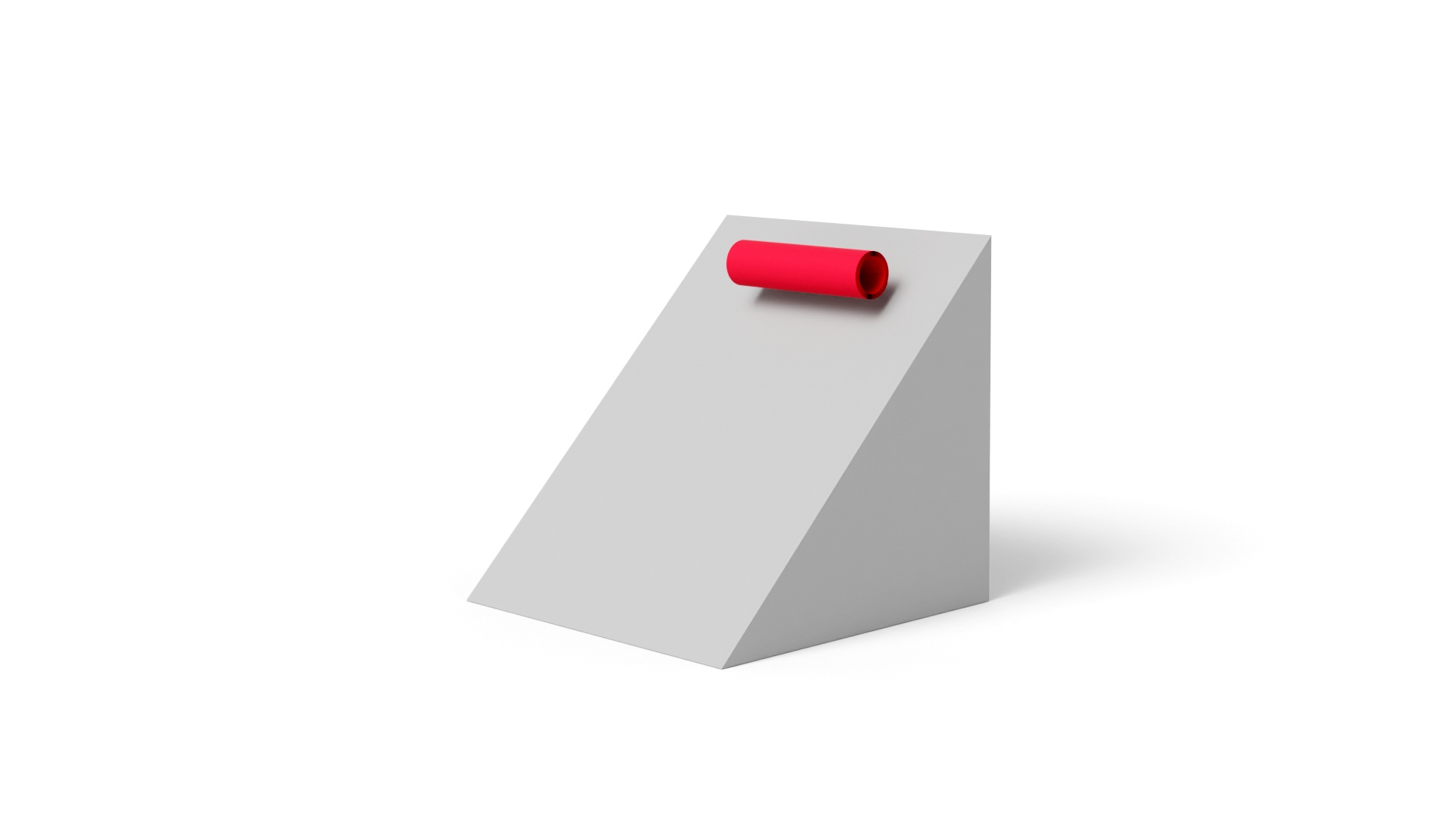}{$t\!=\!0$} & \figunroll{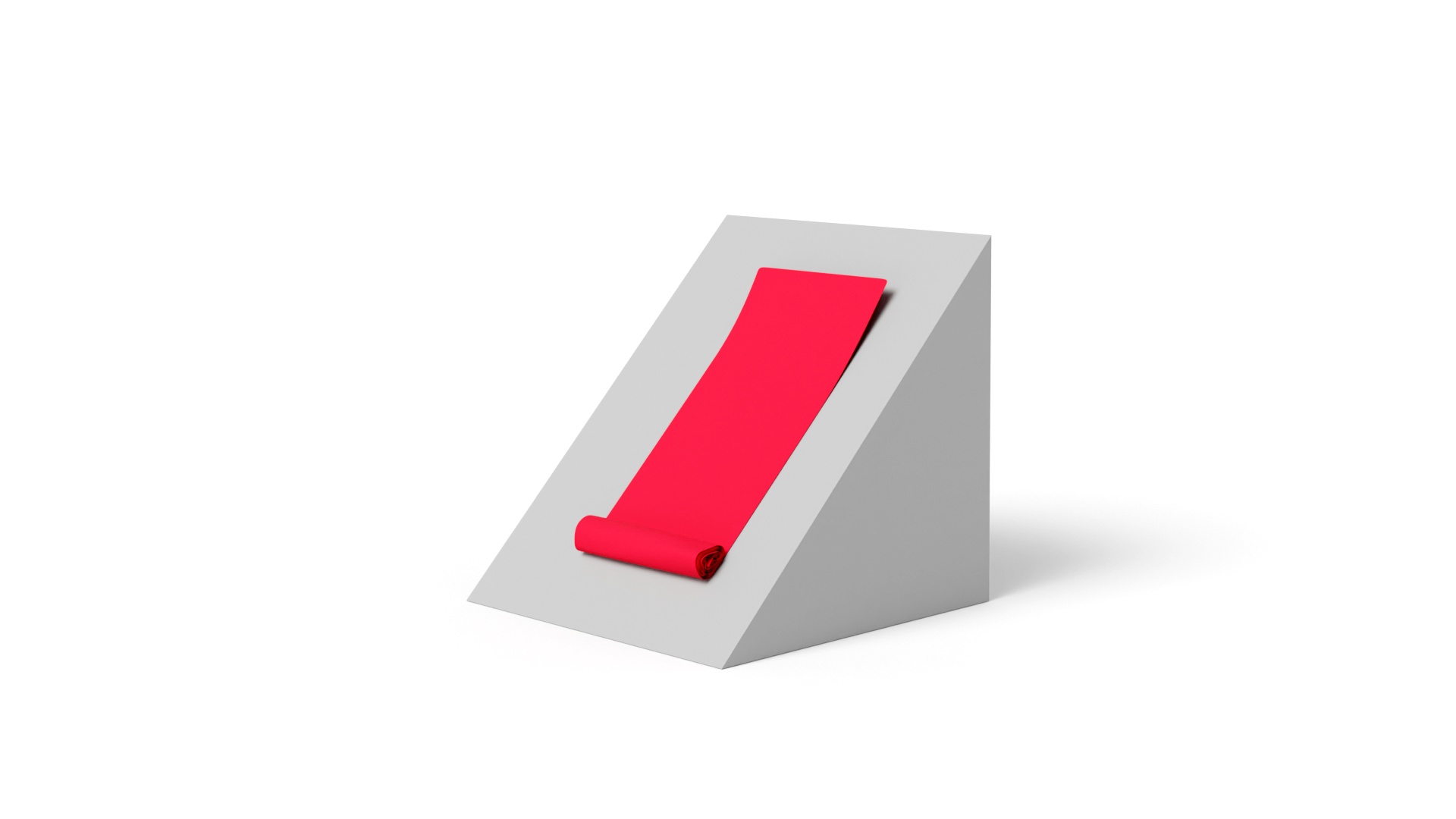}{$t\!=\!1.25$} & \figunroll{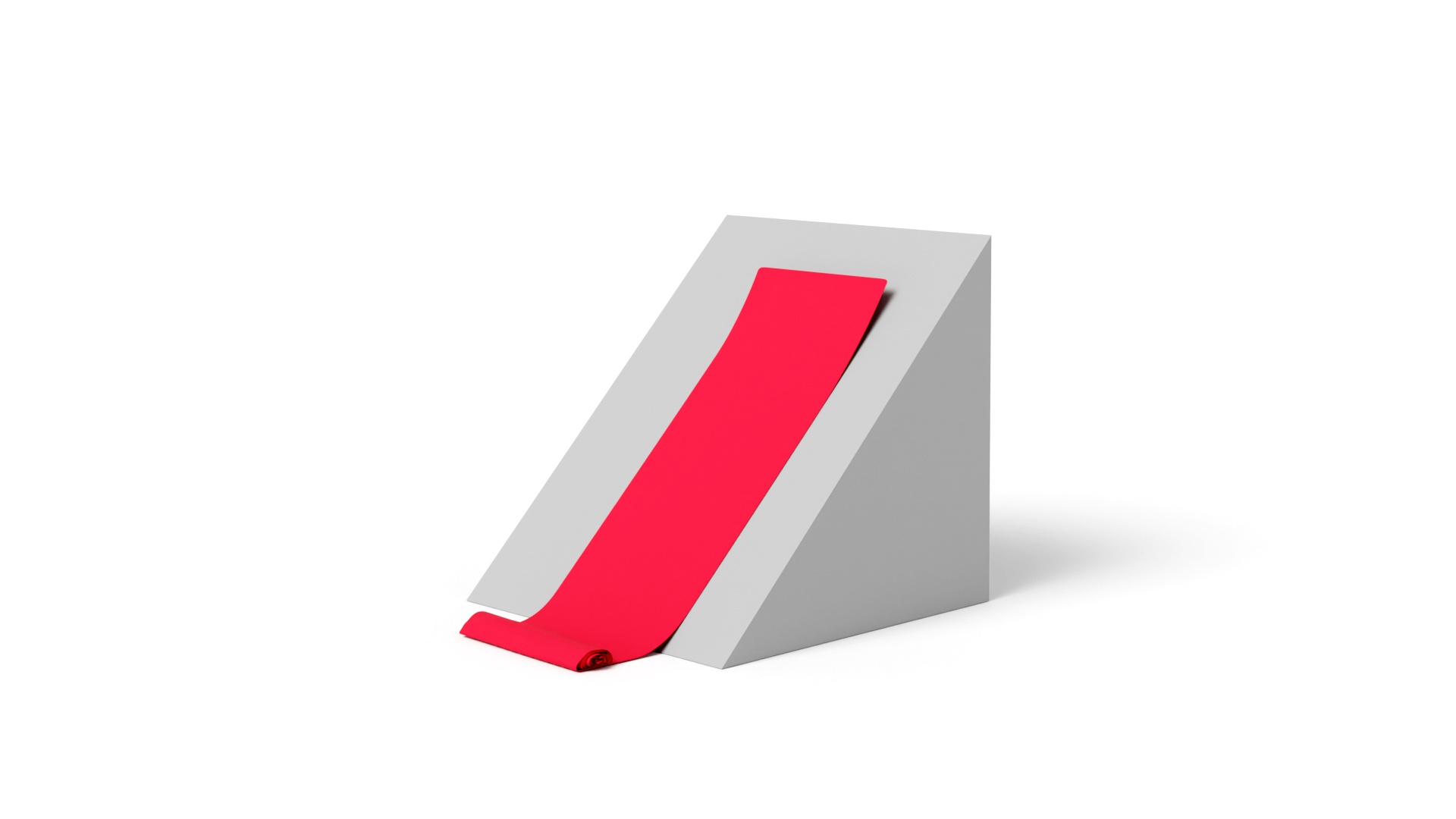}{$t\!=\!2.5$} \\
\end{tabular}
\caption{Cloth unrolling under gravity. A rolled cloth is released above a prism-shaped collider and unrolls as it falls. Top: Isotropic-DAT. Bottom: Planar-DAT. ($\Delta t=1/600$s, avg.\ step time: Planar-DAT=0.45ms, Isotropic-DAT=0.4ms)}
\label{fig:unrollCloth}
\end{figure}

\autoref{fig:unrollCloth} shows a rolled cloth unrolling under gravity onto a prism collider. As the outer layers separate and fall, dense self-contact occurs between adjacent cloth layers. Planar-DAT preserves natural falling dynamics, while Isotropic-DAT's isotropic constraints introduce additional resistance that slows the unrolling motion.

\begin{figure}[t]
\centering
\newcommand{\figosc}[2]{%
\begin{tikzpicture}
\node[anchor=south west,inner sep=0] (image) at (0,0) {\includegraphics[width=0.32\linewidth,trim=100 350 100 350,clip]{Figures/OscillatingCloth/#1.jpg}};
\ifx\relax#2\relax\else\node[anchor=south west,fill=white,fill opacity=0.7,text opacity=1,inner sep=2pt] at (image.south west) {\scriptsize #2};\fi
\end{tikzpicture}}
\begin{tabular}{@{}c@{\hspace{2pt}}c@{\hspace{2pt}}c@{}}
\multicolumn{3}{c}{\scriptsize Isotropic-DAT (OGC)} \\[2pt]
\figosc{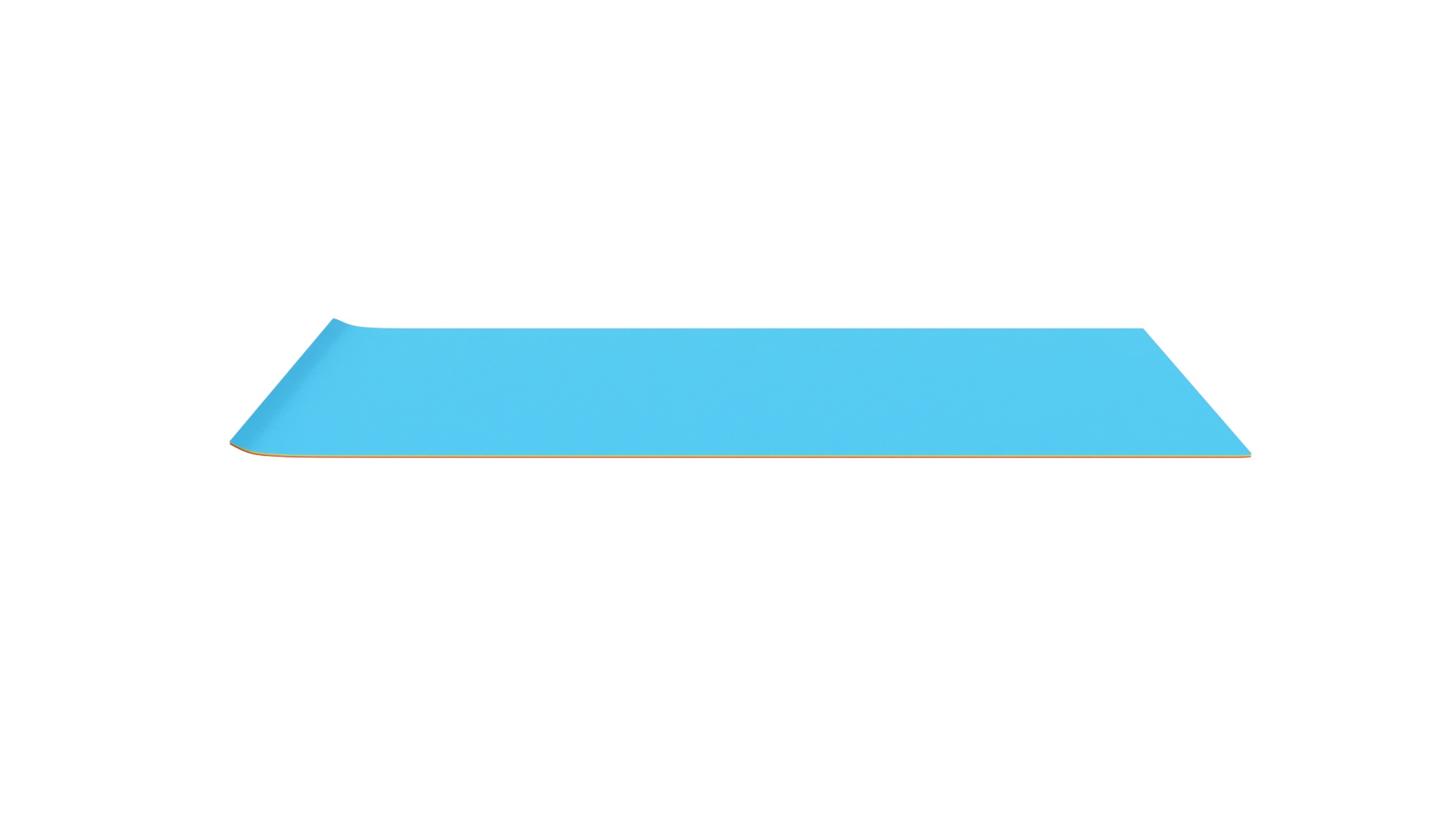}{} & \figosc{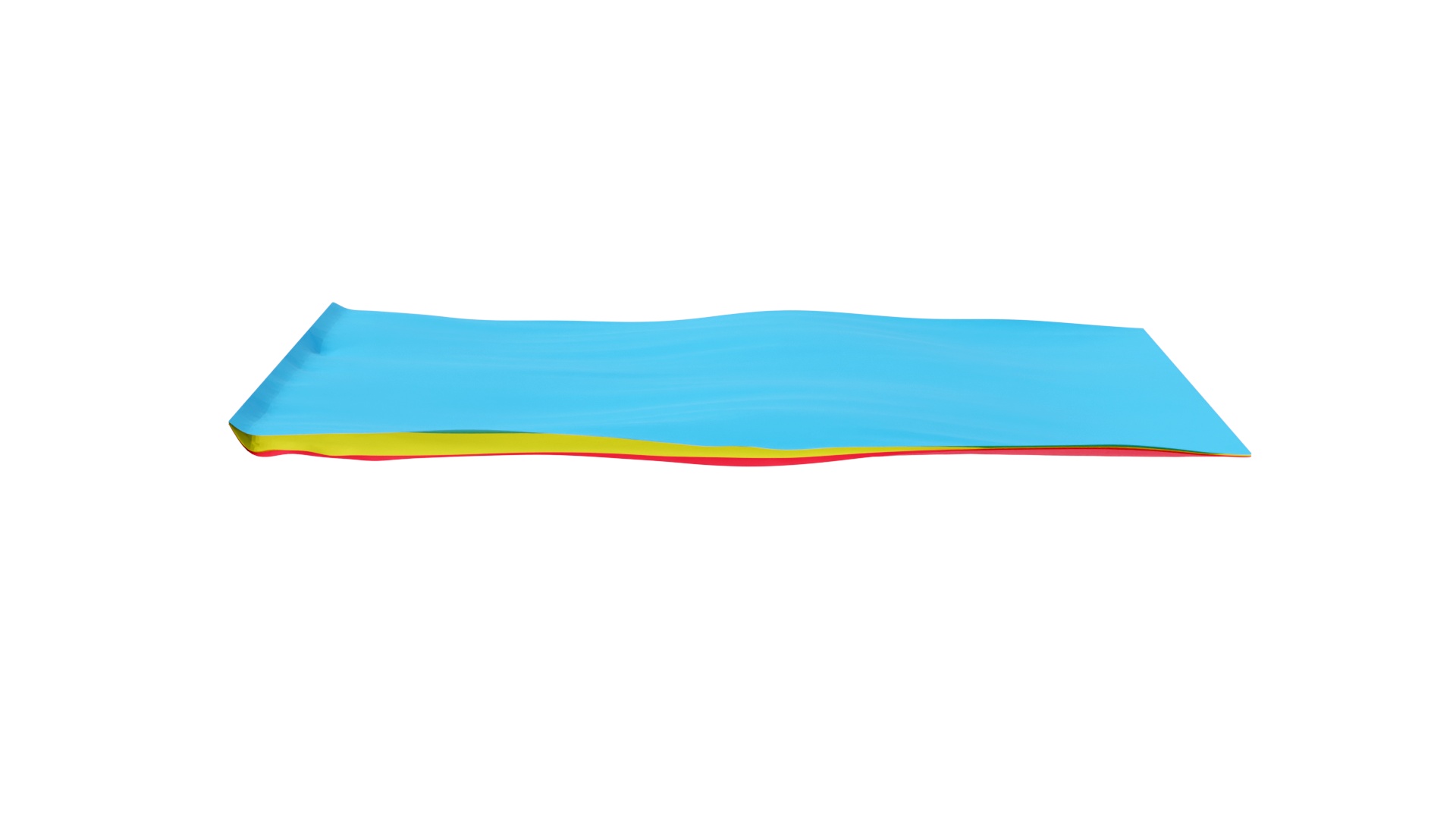}{} & \figosc{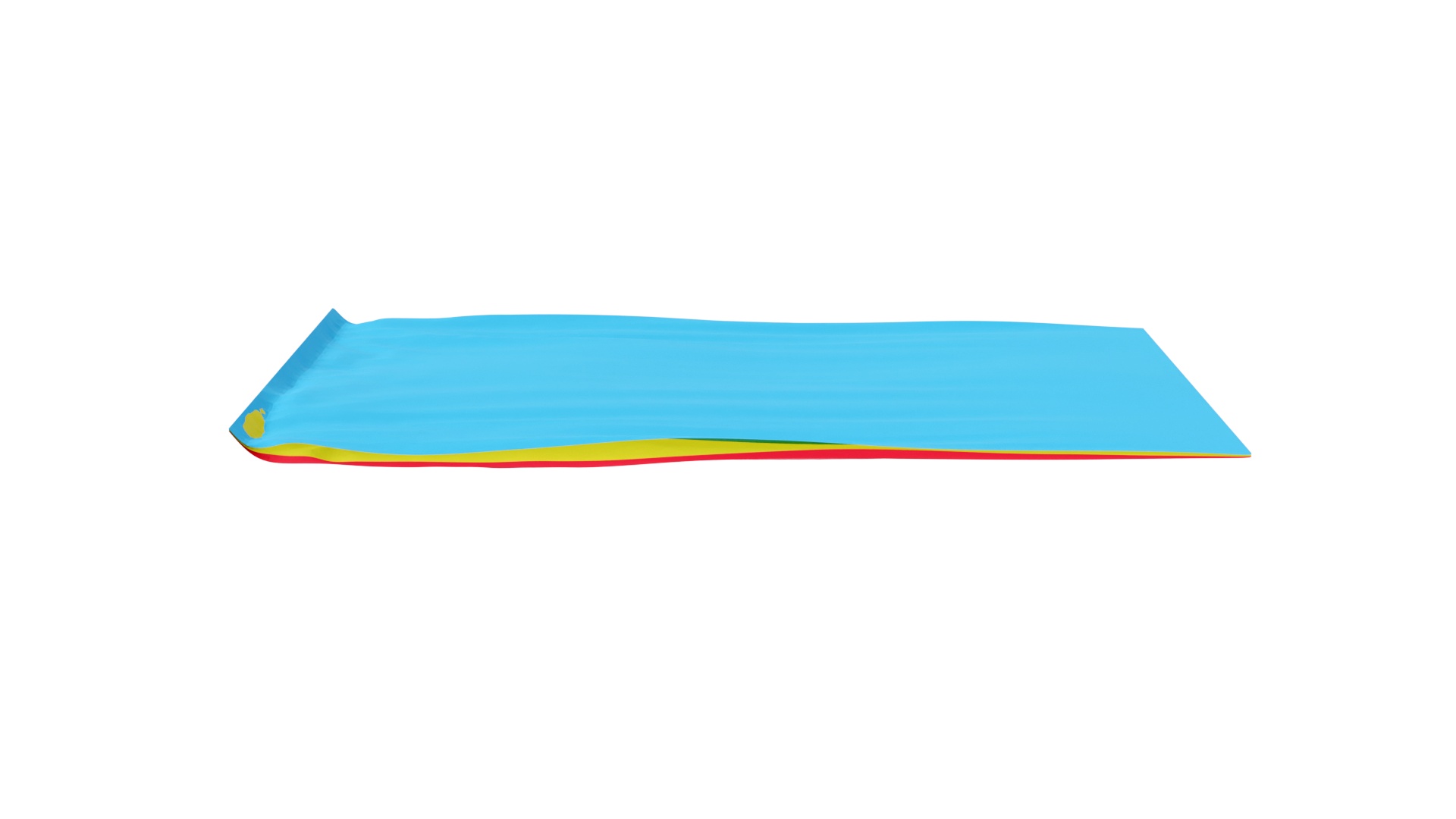}{} \\[2pt]
\multicolumn{3}{c}{\scriptsize Planar-DAT (Ours)} \\[2pt]
\figosc{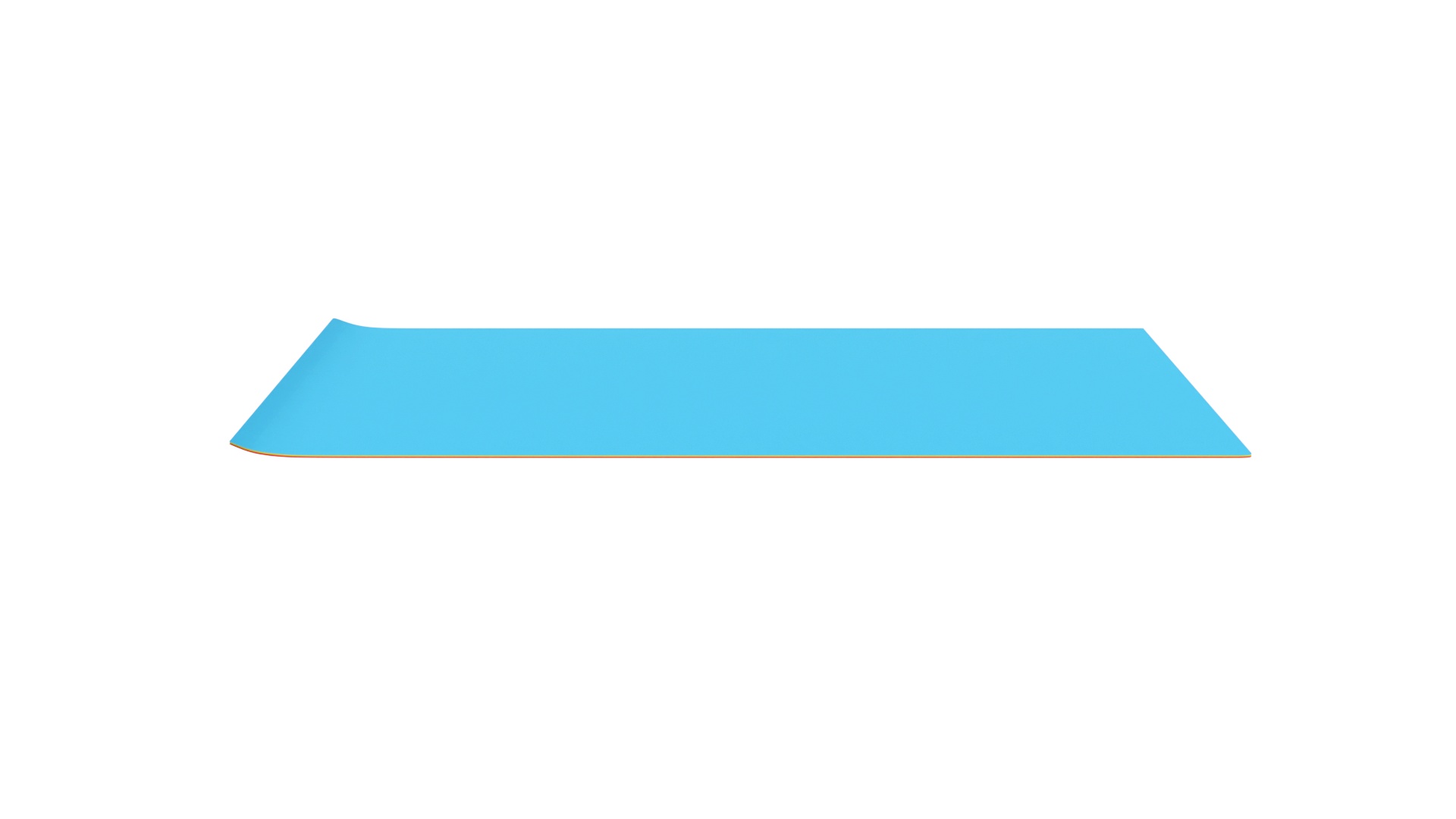}{$t\!=\!0$} & \figosc{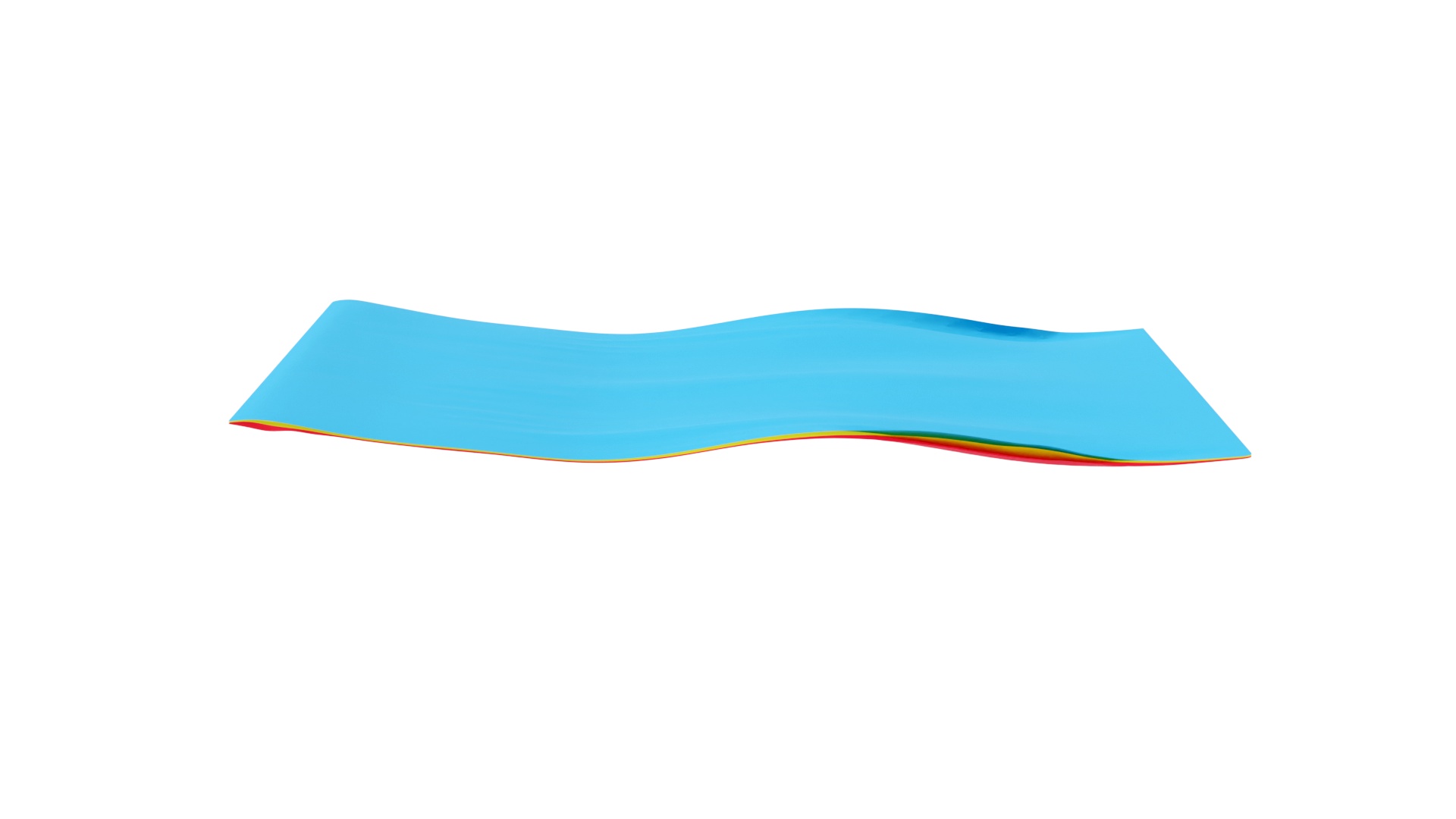}{$t\!=\!3.33$} & \figosc{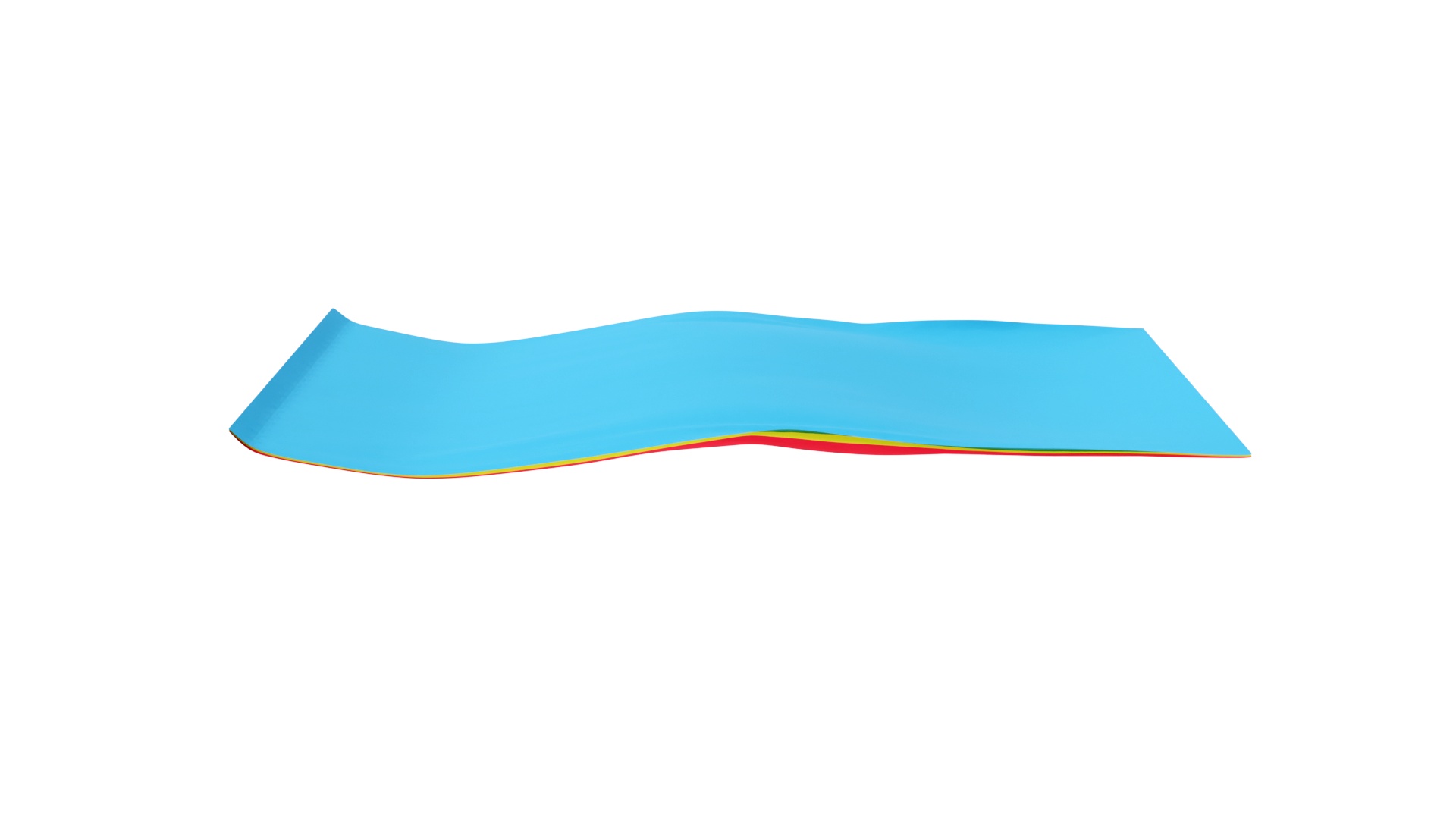}{$t\!=\!10.5$} \\
\end{tabular}

\caption{Oscillating cloth layers. Three stacked cloth layers with one edge fixed and the opposite edge oscillating sinusoidally. Top: Isotropic-DAT. Bottom: Planar-DAT. ($\Delta t=1/300$s, avg.\ step time: Planar-DAT=0.65ms, Isotropic-DAT=0.6ms)}
\label{fig:oscillatingCloth}
\end{figure}

\autoref{fig:oscillatingCloth} shows three stacked cloth layers undergoing sinusoidal oscillation at one edge. The layers maintain close proximity throughout the motion, creating continuous self-contact. Planar-DAT produces more responsive wave propagation through the cloth, while Isotropic-DAT's constraints damp the oscillatory motion, resulting in reduced amplitude over time.

\begin{figure}[t]
\centering
\newcommand{\figgift}[2]{%
\begin{tikzpicture}
\node[anchor=south west,inner sep=0] (image) at (0,0) {\includegraphics[width=0.24\linewidth,trim=650 50 650 50,clip]{Figures/GiftBoxDrop/#1.jpg}};
\ifx\relax#2\relax\else\node[anchor=south west,fill=white,fill opacity=0.7,text opacity=1,inner sep=2pt] at (image.south west) {\scriptsize #2};\fi
\end{tikzpicture}}
\begin{tabular}{@{}c@{\hspace{2pt}}c@{\hspace{2pt}}c@{\hspace{2pt}}c@{}}
\multicolumn{4}{c}{\scriptsize Isotropic-DAT (OGC)} \\[2pt]
\figgift{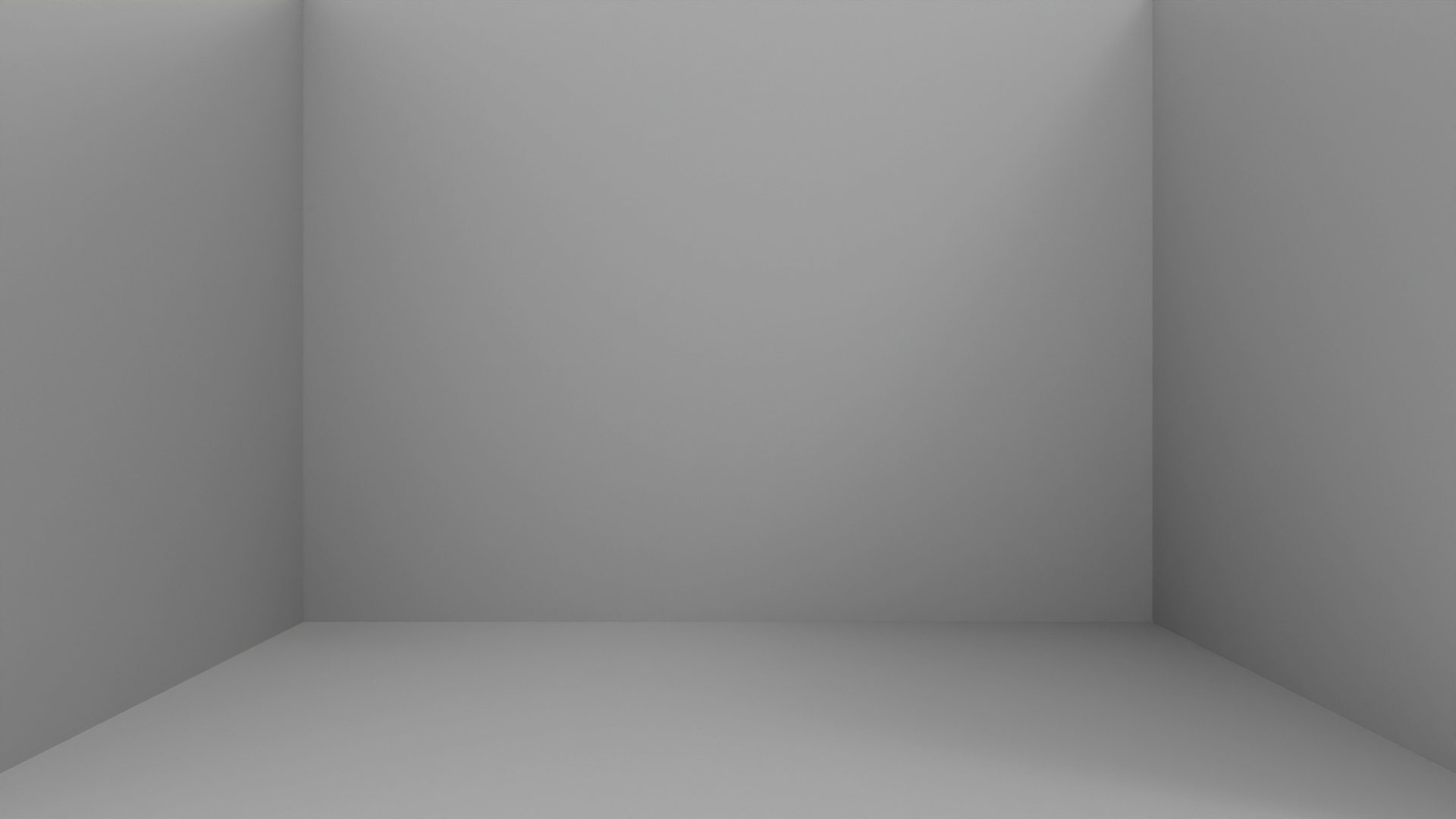}{} & \figgift{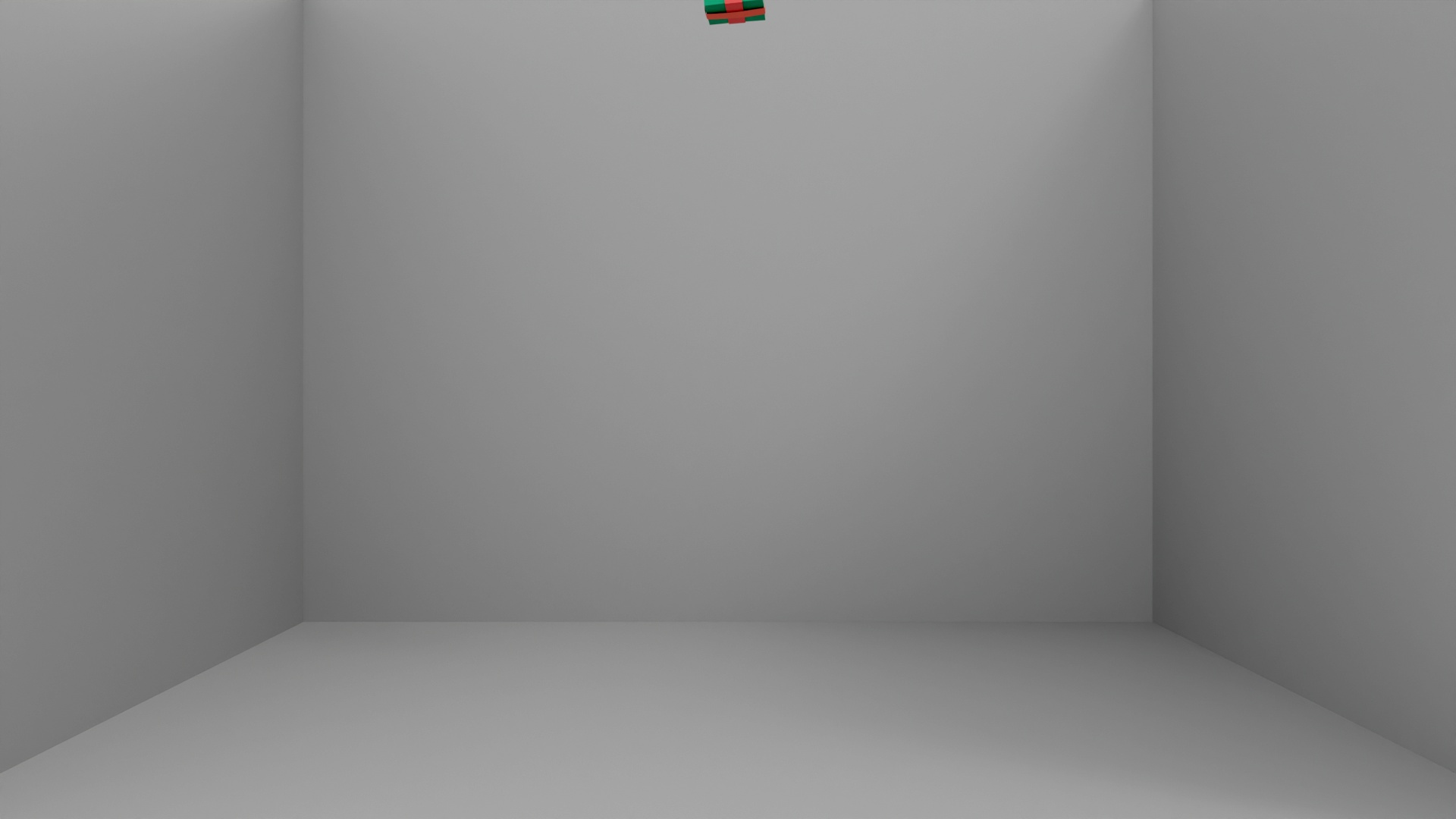}{} & \figgift{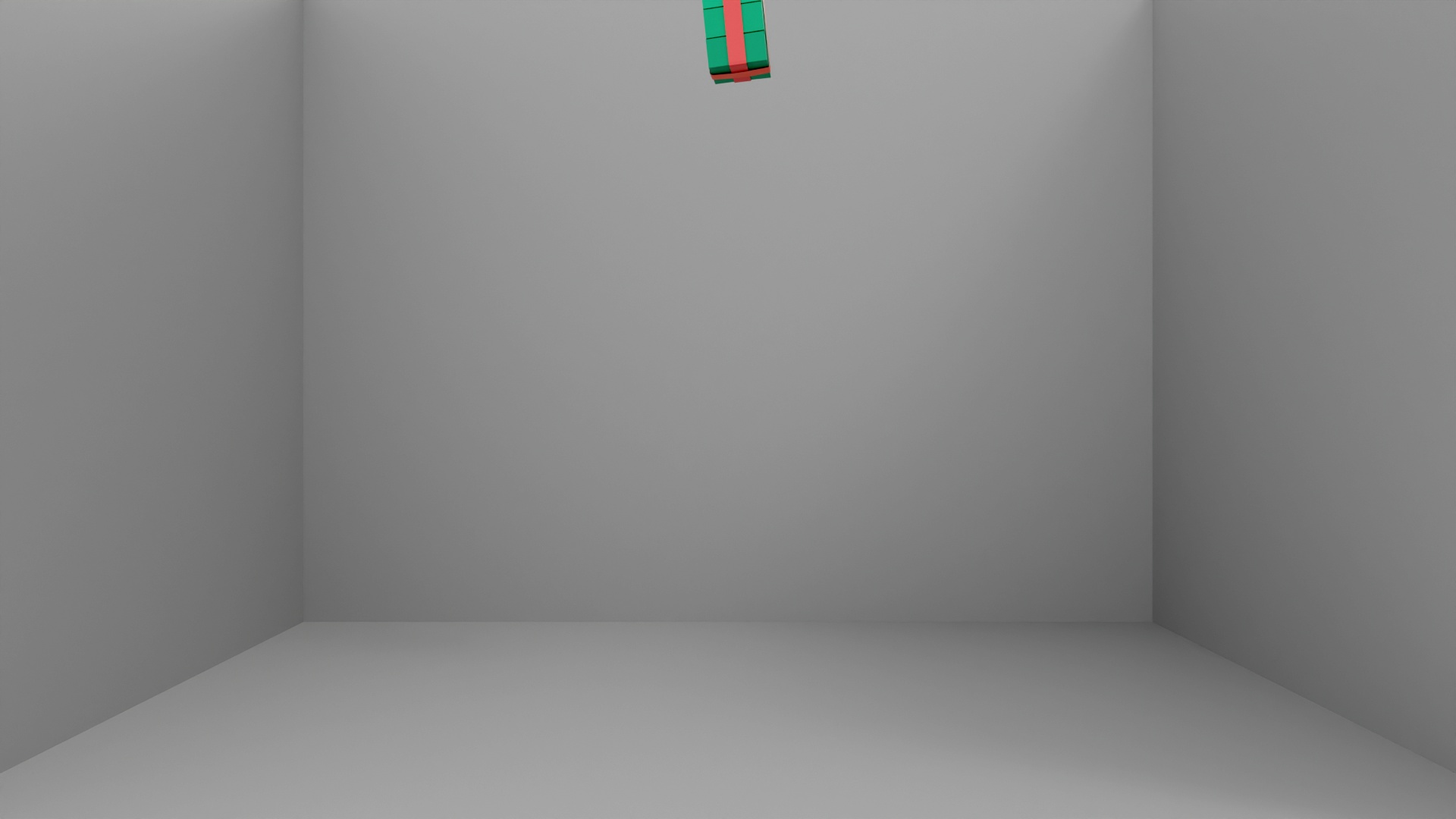}{} & \figgift{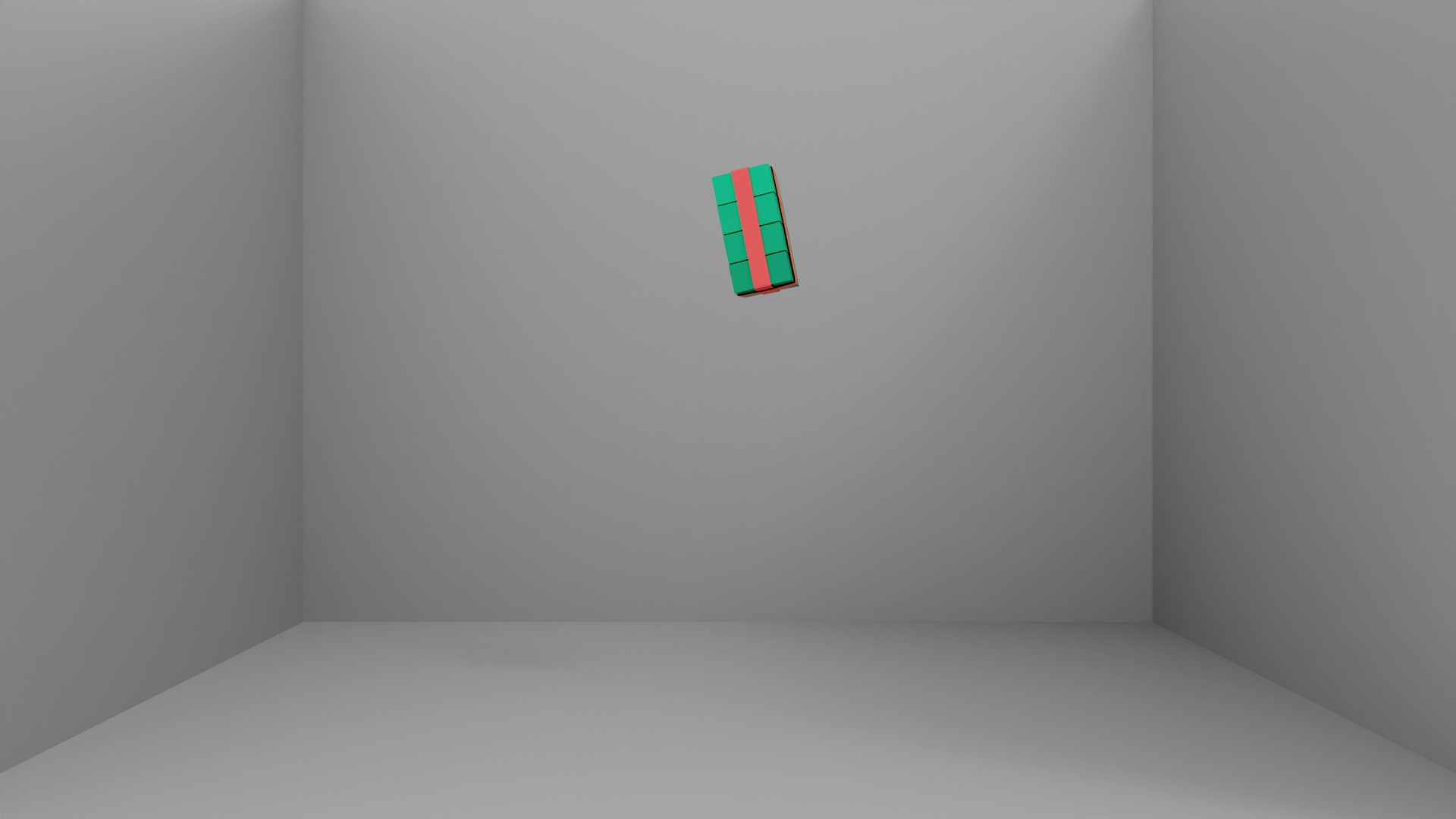}{} \\[2pt]
\multicolumn{4}{c}{\scriptsize Planar-DAT (Ours)} \\[2pt]
\figgift{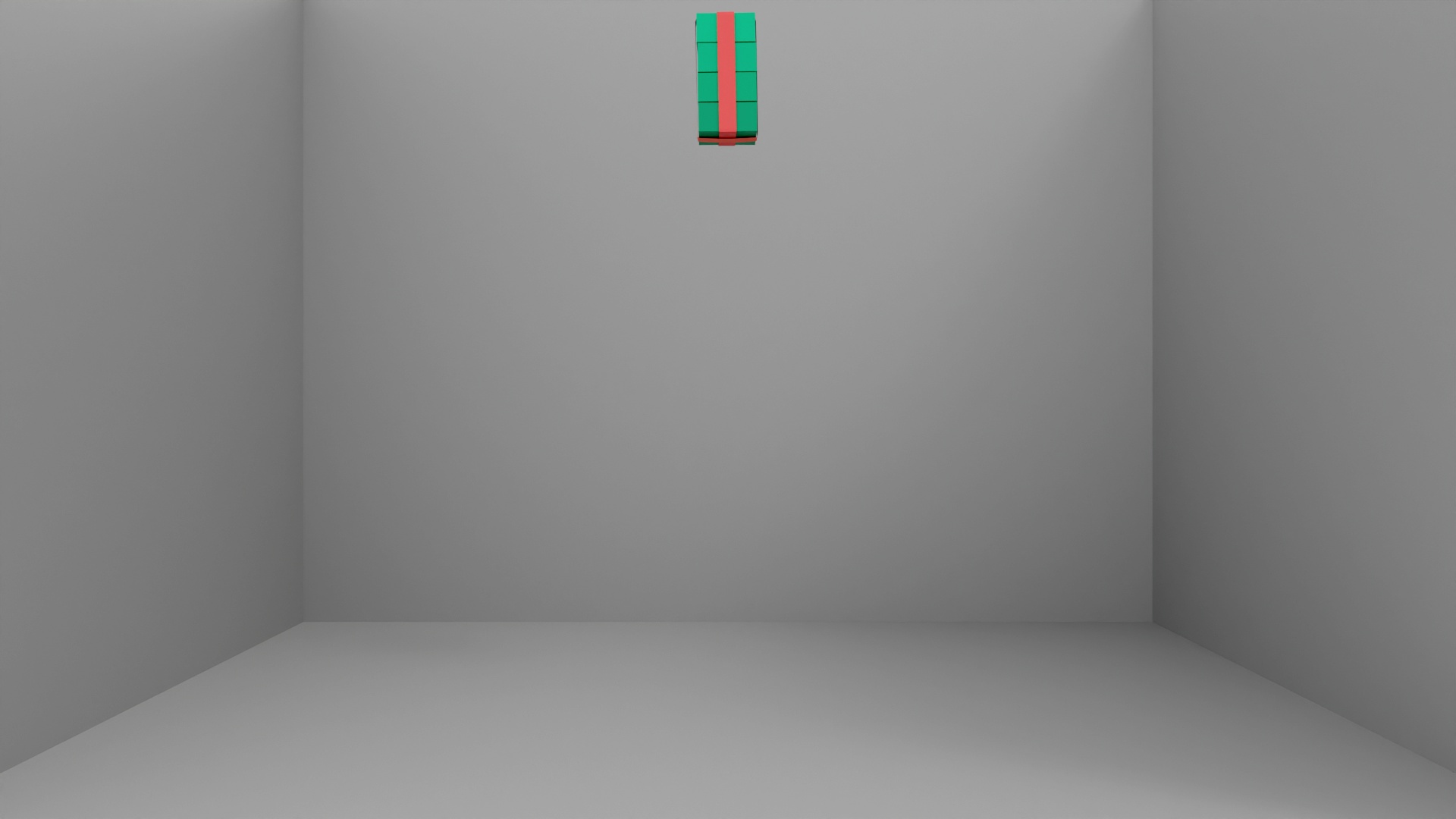}{} & \figgift{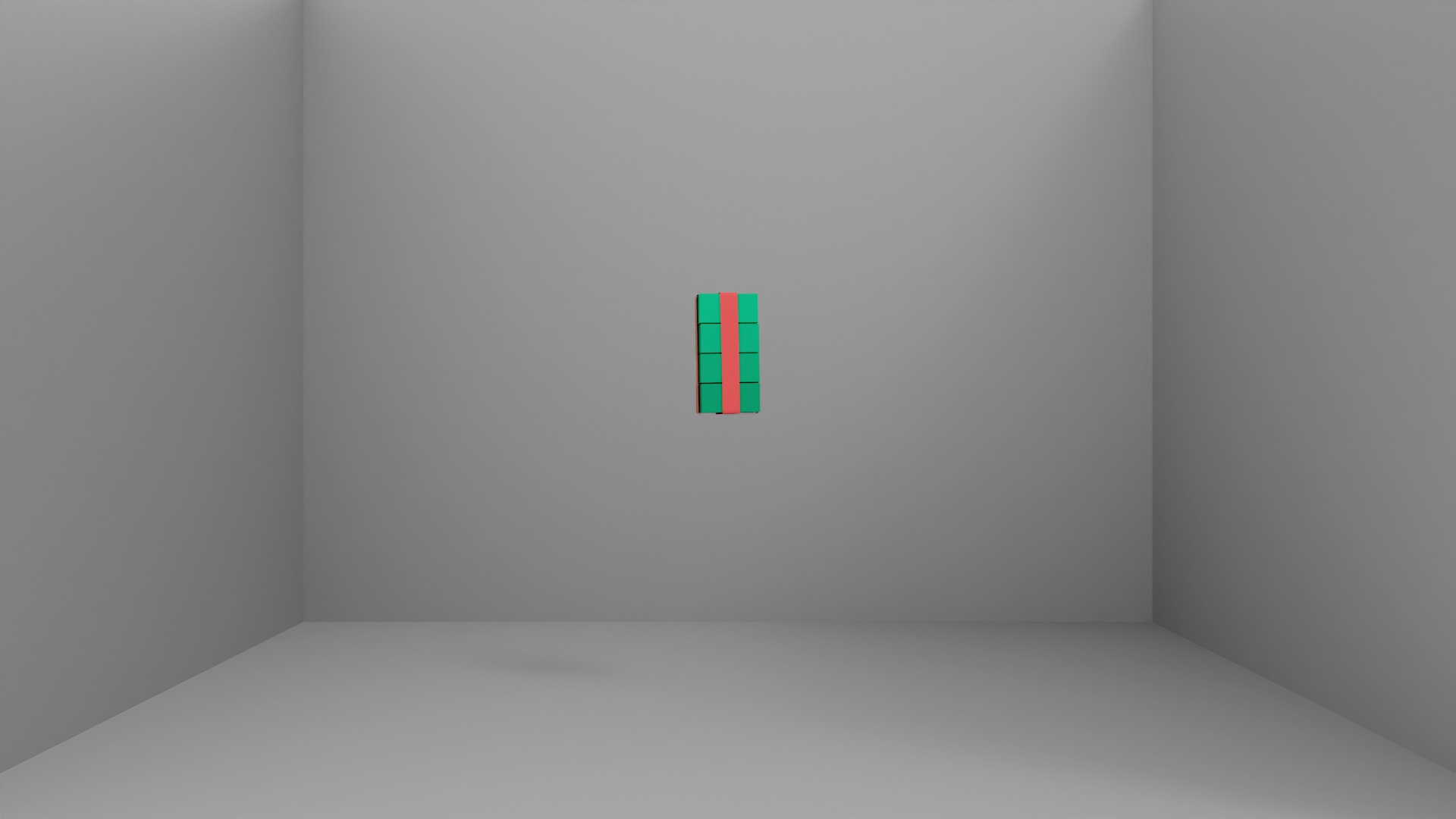}{} & \figgift{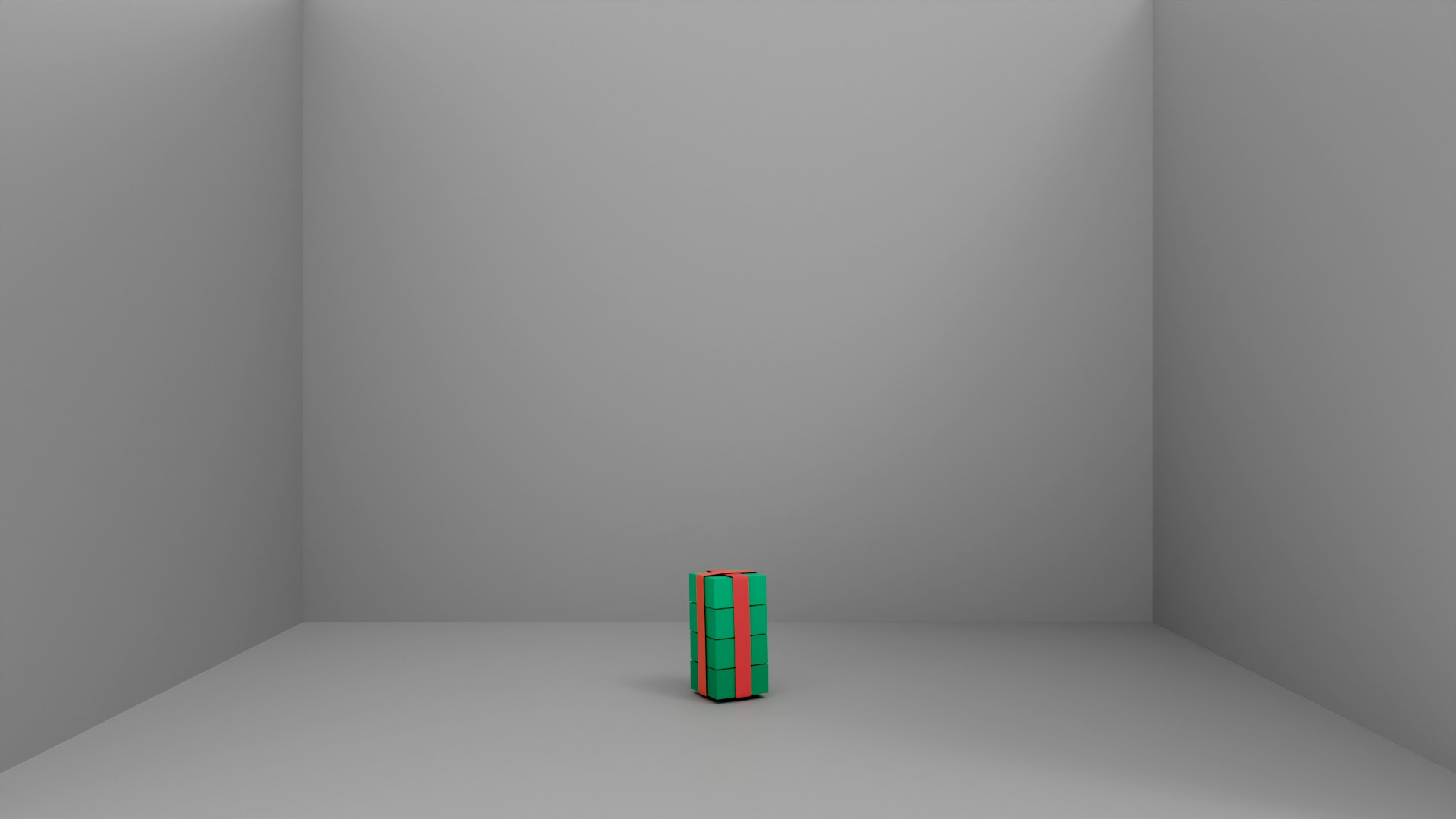}{} & \figgift{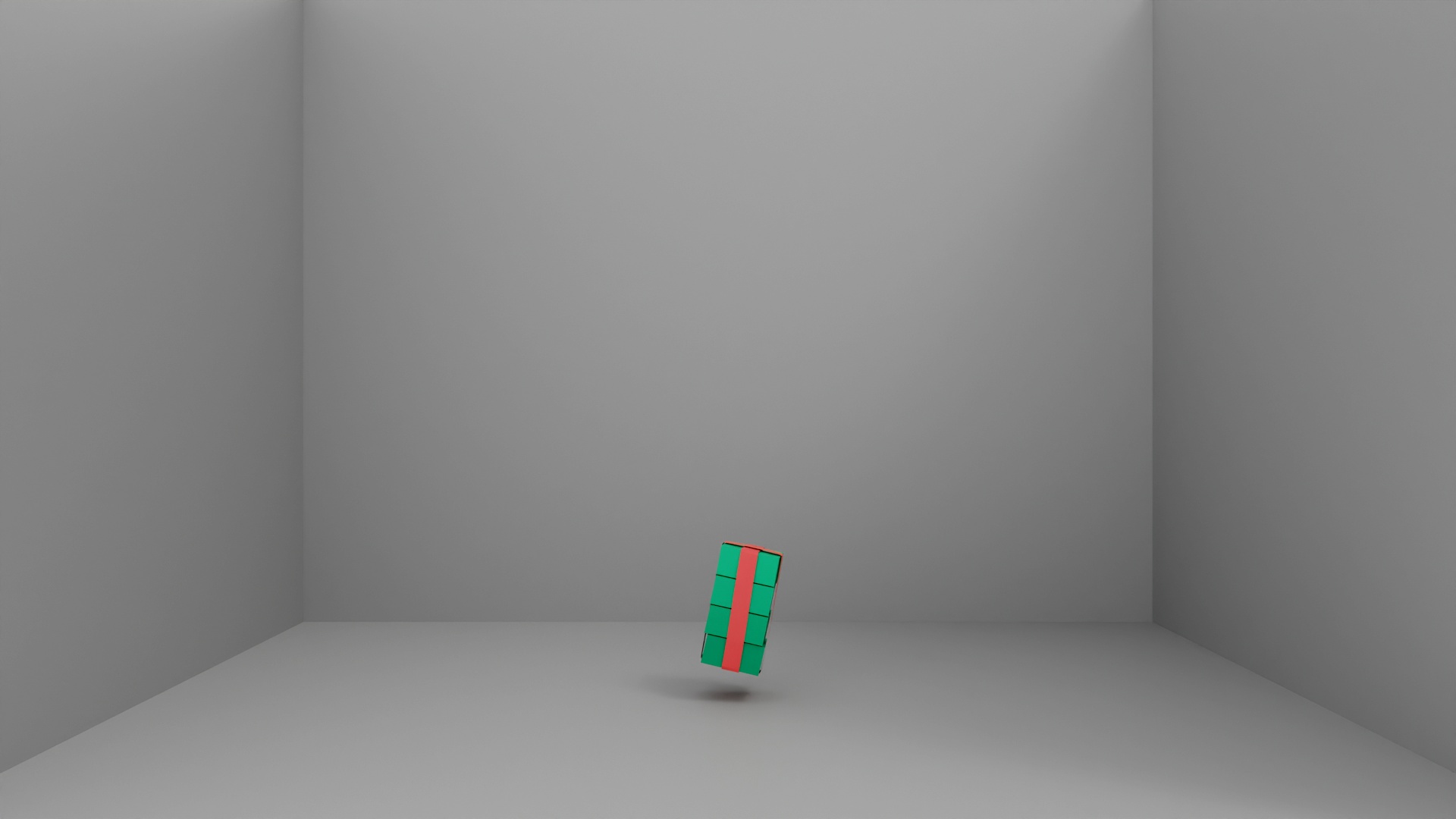}{} \\[2pt]
\multicolumn{4}{c}{\scriptsize Reference (Converged)} \\[2pt]
\figgift{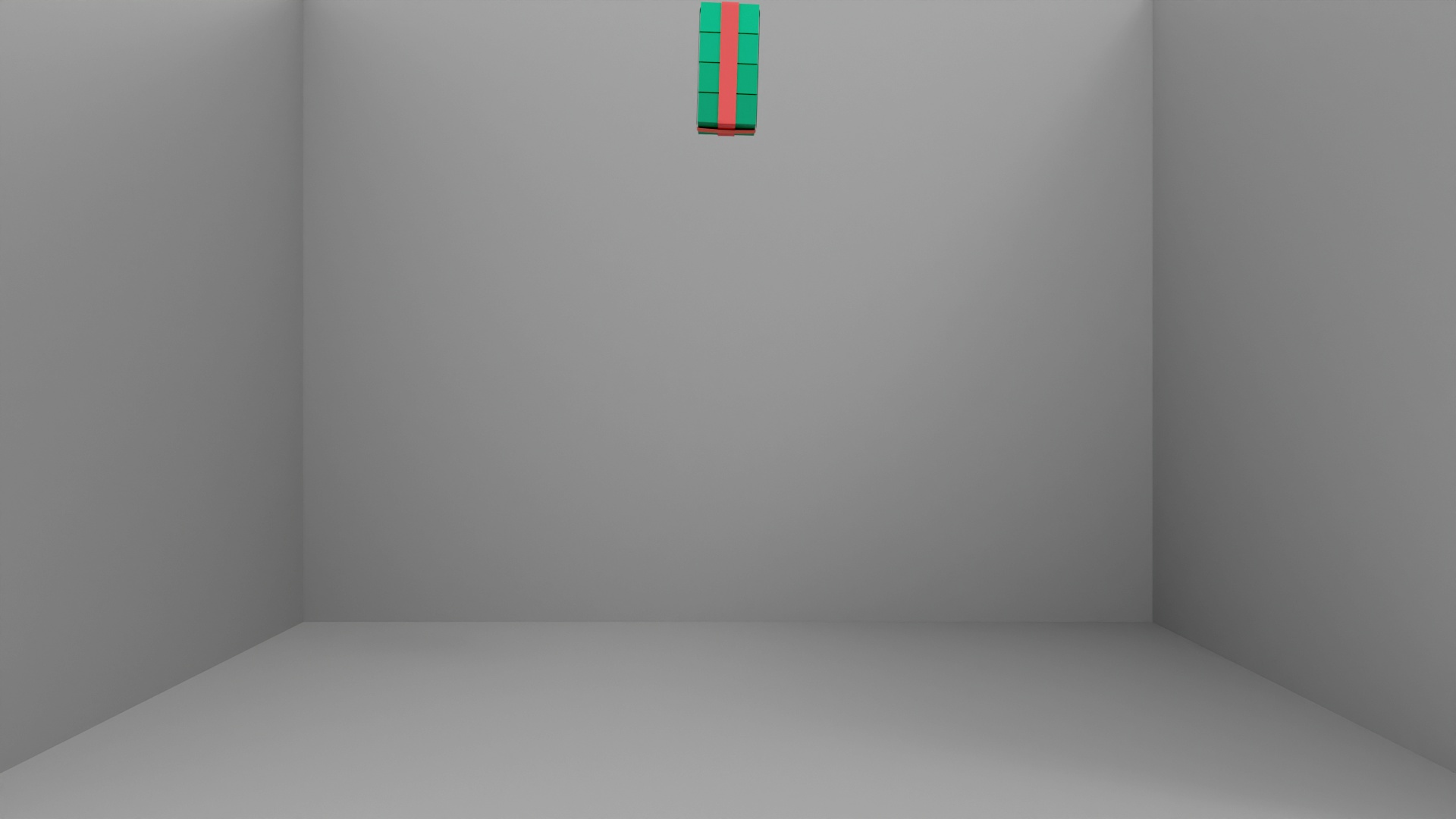}{$t\!=\!1.43$} & \figgift{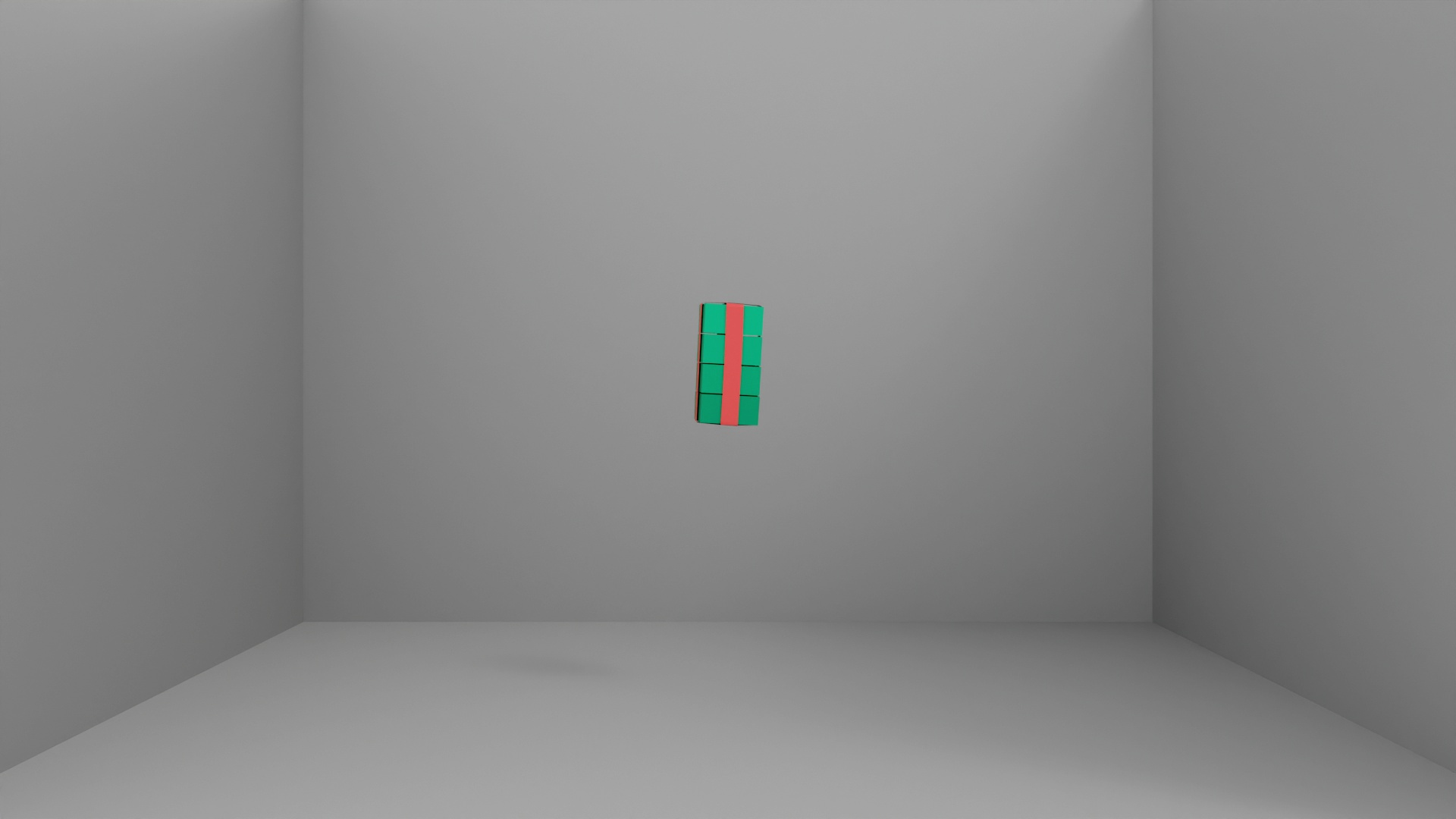}{$t\!=\!2.07$} & \figgift{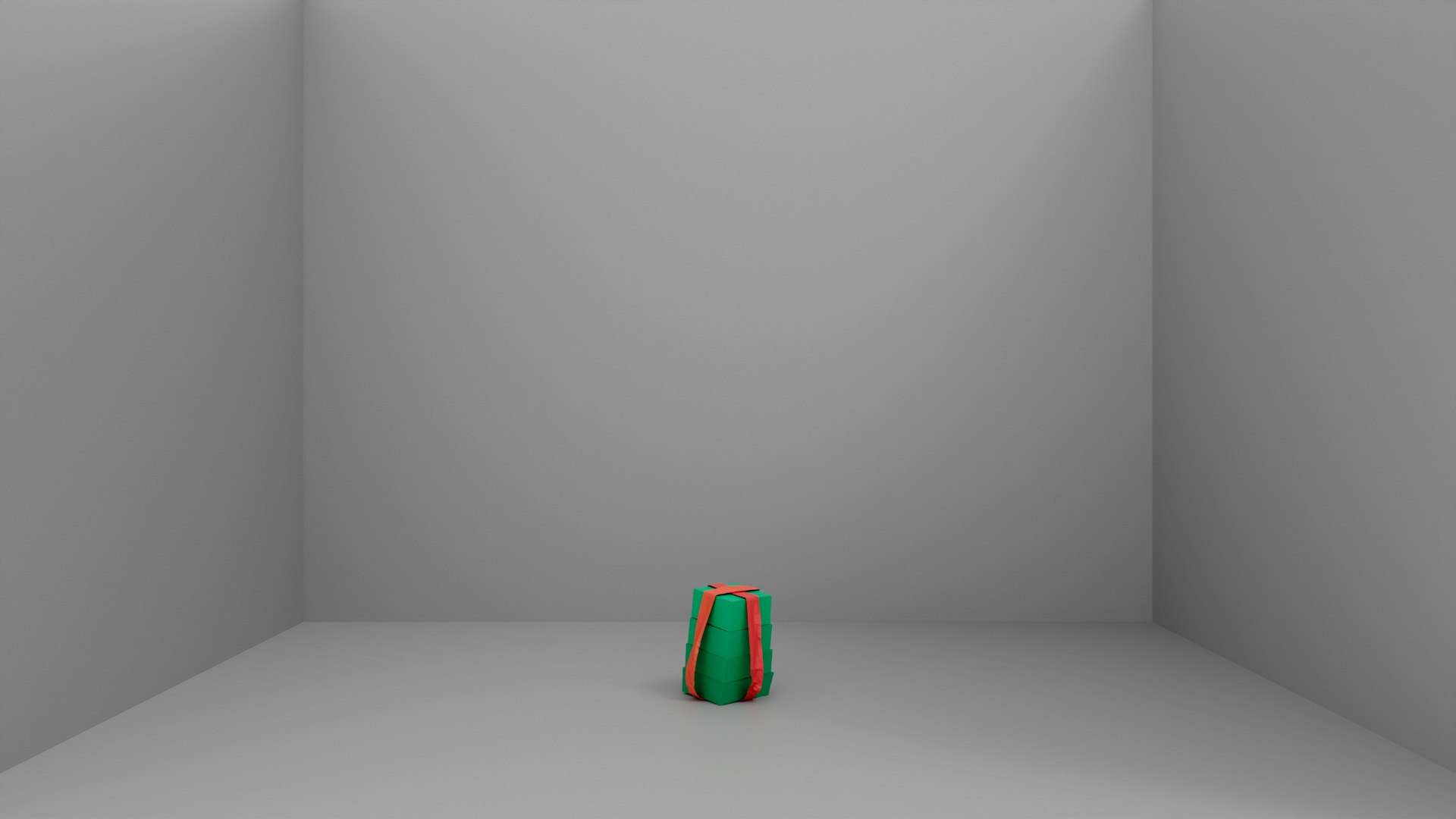}{$t\!=\!2.7$} & \figgift{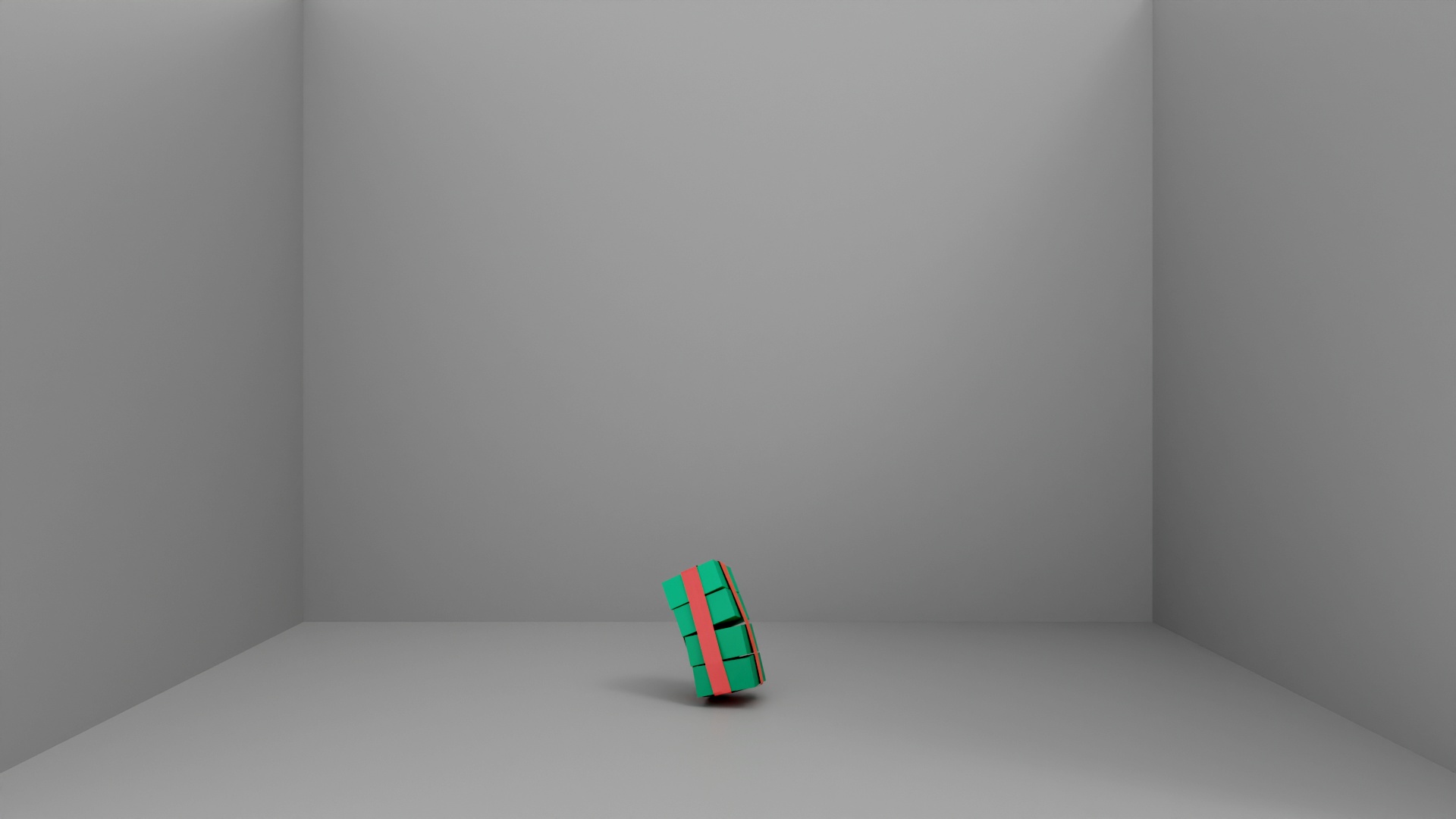}{$t\!=\!5$} \\
\end{tabular}
\caption{Wrapped gift box drop. Four stacked soft body blocks wrapped with two cloth straps are dropped onto the ground. Top: Isotropic-DAT. Middle: Planar-DAT. Bottom: Reference (Converged). ($\Delta t=1/600$s, avg.\ step time: Planar-DAT=0.4ms, Isotropic-DAT=0.35ms)}
\label{fig:giftBoxDrop}
\end{figure}

\autoref{fig:giftBoxDrop} demonstrates coupled soft body and cloth simulation with multiple contact interactions. Four stacked deformable blocks, wrapped with two perpendicular cloth straps, fall and collide with the ground. We include a converged reference simulation using Planar-DAT, representing the ground-truth behavior. Planar-DAT closely matches this reference, producing free falling, while Isotropic-DAT exhibits overdamped behavior due to its isotropic motion constraints.

\subsection{Quantitative Evaluation}
We conduct a controlled benchmark to quantify the improvement of Planar-DAT over previous works. We generate 10 random meshes, each with 3000 disjoint triangles, and for each mesh apply 100 random deformations: either per-vertex displacements (Linear) or per-triangle rigid motions with random rotation and translation (Rot.+Lin.), both scaled to 1\% of the mesh extent. We then apply both truncation schemes and measure the per-vertex displacement magnitude preservation ratio. We use $r_q$ two times as large as the deformation scale in this experiment.

\autoref{table:TruncationComparison} reports the relative improvement of Planar-DAT over Isotropic-DAT and CCD across all 9,000,000 displacements in the benchmark. The mean improvement of more than 2.22$\times$ in both linear and linear+rotational setup  demonstrate that Planar-DAT consistently preserves significantly more of the intended motion. Notably, the vertices where Planar-DAT shows the largest improvement are exactly those in the most challenging regions, where the mesh is embraced by complex, overlapping contact and large truncations would otherwise be imposed by Isotropic-DAT. In these difficult cases, Planar-DAT's directional awareness allows it to unlock hundreds of times more safe deformation, resulting in substantially better fidelity for highly-colliding, tightly-coupled configurations.
Compared to global truncation with CCD, Planar-DAT preserves over 100$\times$ more displacement than CCD on average, demonstrating the critical importance of per-vertex directional constraints over global uniform scaling. This is due to CCD being extremely conservative: a single early collision anywhere in the mesh restricts the entire system. 

\begin{figure}[t]
    \centering
    \includegraphics[width=\linewidth]{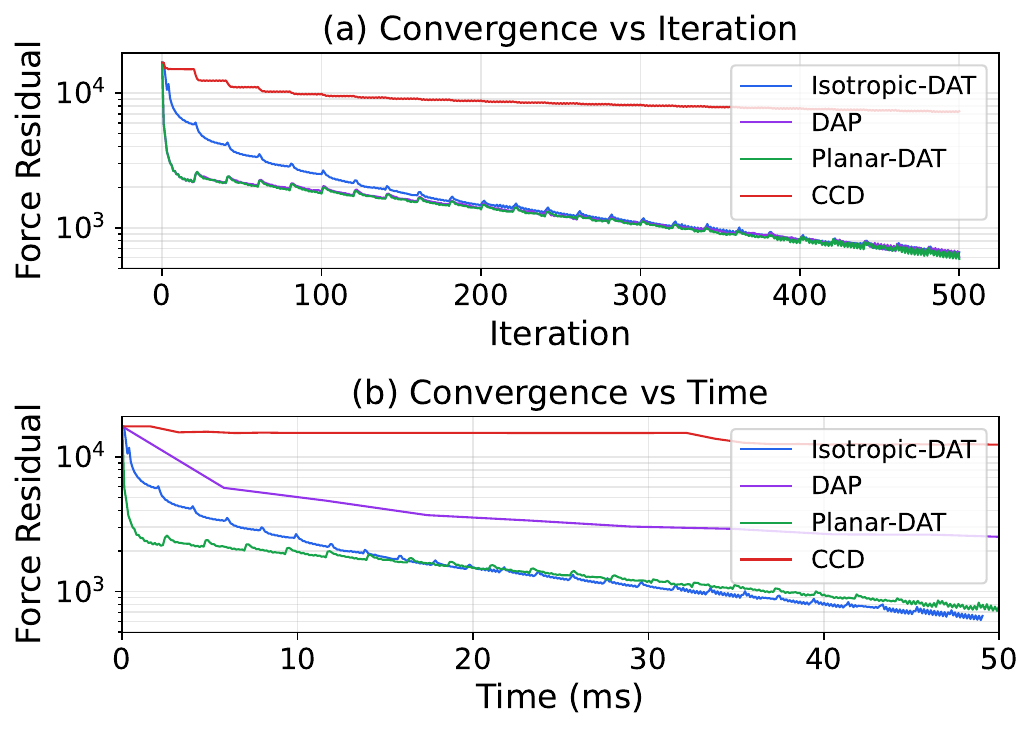}
    \caption{Convergence comparison on the cloth twist scenario (t=10s in \autoref{fig:clothTwistComparison}). (a)~Per-vertex average force residual vs.\ solver iteration. (b)~Force residual vs.\ wall-clock time. Isotropic-DAT and Planar-DAT have similar per-iteration cost, while CCD is $\sim$16$\times$ slower and DAP is $\sim$50$\times$ slower per iteration. Planar-DAT achieves the best convergence per unit time.}
    \label{fig:convergence}
\end{figure}

We evaluate solver convergence on the cloth twist scenario (\autoref{fig:clothTwistComparison}), where a sheet is twisted until packed with dense self-contact and then released. We record per-vertex average force residual across 500 VBD iterations at the step where release happens, with collision detection happens in every 20 iterations. As shown in \autoref{fig:convergence}(a), Planar-DAT reduces residual energy faster than Isotropic-DAT, especially in early iterations—important for real-time simulation. Although Planar-DAT incurs $\sim$10\% higher per-iteration cost (due to per-contact half-space evaluation), its directional constraints enable larger effective steps, requiring fewer iterations to reach a given accuracy.
By contrast, CCD-based truncation runs at $\sim$16\,ms per iteration—over $10\times$ slower than either DAT scheme (\autoref{fig:convergence}(b)). CCD's error curve also plateaus very quickly until the next collision happens, because the earliest collision that locks up the entire optimization is not detected by discrete collision detection and therefore lack\del{s of}\add{s} repulsive force\add{. We note that Local CCD~\citep{lan2023Stencil} can mitigate this TOI locking by confining CCD to per-stencil collision pairs, though a global CCD pass is still required to detect new pairs}.
Divide and Project (DAP), which iteratively projects onto half-space intersections, achieves similar per-iteration convergence to Planar-DAT but requires $\sim$50$\times$ more time, making it impractical for real-time use. This gap arises as DAP performs expensive and sequential iterative computations every solver iteration while both DAT variants efficiently reuse cached collision detection results. Planar-DAT thus achieves the best balance of convergence rate and computational cost.

\begin{table}[h]
\centering
\begin{minipage}[c]{0.18\linewidth}
    \centering
    \includegraphics[width=\linewidth,trim=500 50 550 50,clip]{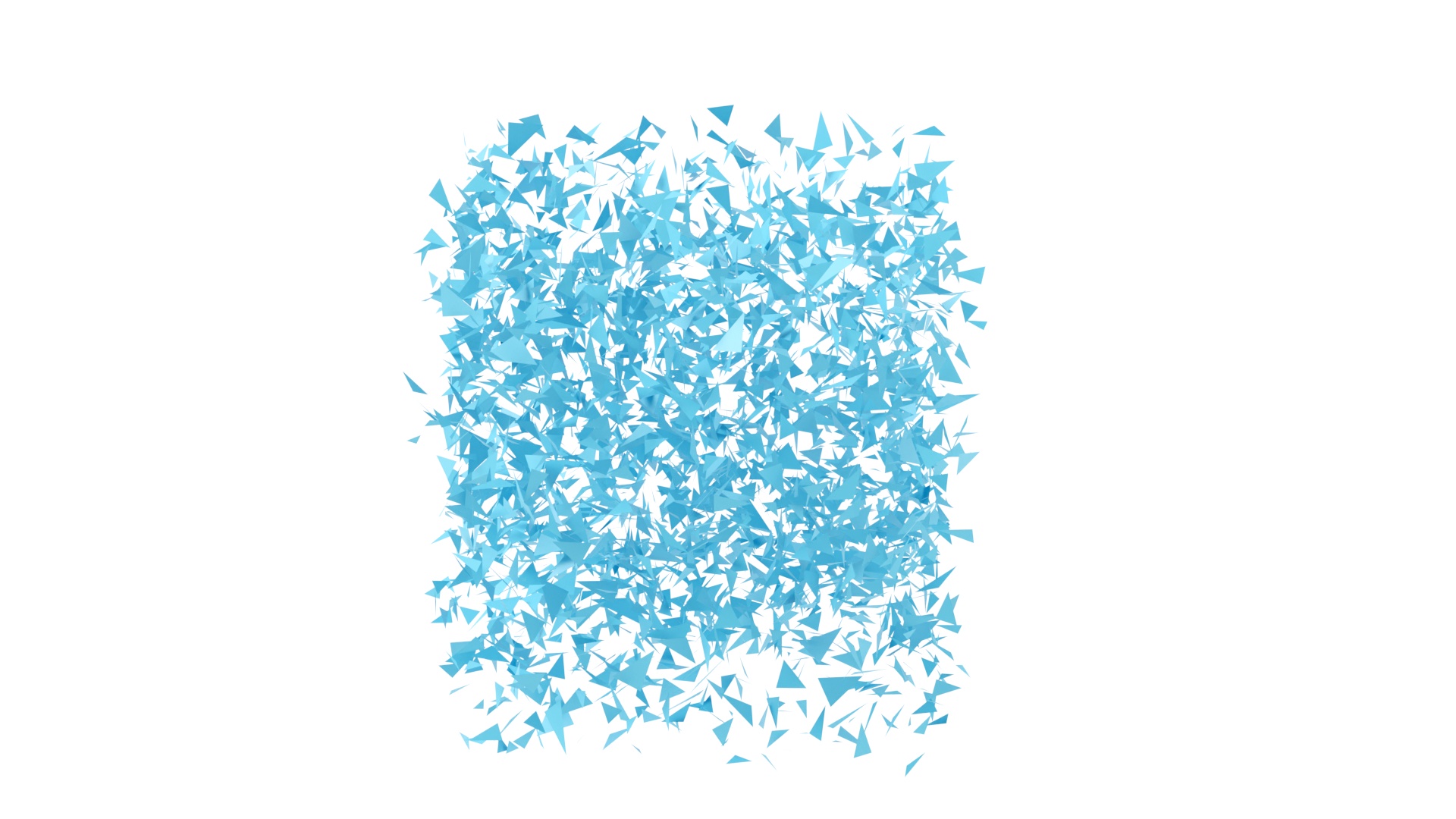}
\end{minipage}%
\hspace{0.01\linewidth}%
\begin{minipage}[c]{0.80\linewidth}
    \centering
    \small
    \begin{tabular}{|l|c|c|c|c|}
    \hline
    & \multicolumn{2}{c|}{\textbf{vs Isotropic-DAT}} & \multicolumn{2}{c|}{\textbf{vs Global CCD}} \\
    \cline{2-5}
    \textbf{Statistic} & Linear & Rot.+Lin. & Linear & Rot.+Lin. \\
    \hline
    Mean            & 2.22$\times$ & 2.34$\times$ & 265$\times$ & 289$\times$ \\
    Median          & 1.42$\times$ & 1.65$\times$ & 265$\times$ & 282$\times$ \\
    Max             & 157.6$\times$ & 178.2$\times$ & 460$\times$ & 512$\times$ \\
    \hline
    50th Percentile & 1.42$\times$ & 1.65$\times$ & 265$\times$ & 282$\times$ \\
    75th Percentile & 2.28$\times$ & 2.58$\times$ & 360$\times$ & 392$\times$ \\
    90th Percentile & 4.03$\times$ & 4.72$\times$ & 411$\times$ & 448$\times$ \\
    95th Percentile & 6.04$\times$ & 7.18$\times$ & 460$\times$ & 492$\times$ \\
    99th Percentile & 14.1$\times$ & 17.4$\times$ & 460$\times$ & 512$\times$ \\
    \hline
    \end{tabular}
\end{minipage}
\caption{Planar-DAT improvement over Isotropic-DAT and CCD in per-vertex displacement preservation, evaluated on randomly generated triangles (left).}
\label{table:TruncationComparison}
\end{table}

\section{Discussion}
We presented Divide and Truncate (DAT), a framework that guarantees penetration- and inversion-free simulation by partitioning space into exclusive regions and truncating displacements accordingly. Our Planar-DAT formulation incorporates displacement direction to avoid the artificial damping and deadlock of prior isotropic methods, while naturally extending to rotational degrees of freedom and animated objects. Integrated with a GPU-based solver, DAT enables robust real-time simulation of challenging multi-physics scenarios.
\paragraph{Limitations.}
Our method has several limitations worth noting. First, Planar-DAT relies on a query radius $r_q$ to limit the neighborhood search; contacts beyond this radius fall back to isotropic truncation, which can reintroduce damping in scenarios with large displacements. Second, while truncation guarantees penetration-free configurations, it does not preserve momentum or energy if the solver iteration is cut off before convergence, which may affect physical accuracy in high-precision applications. Third, for rigid bodies with rotational degrees of freedom, we rely on sampling and interval arithmetic to bound curved trajectories, which adds computational overhead and may require parameter tuning for very fast rotations.
\add{Finally, since we use OGC as the underlying contact model, our method inherits its contact energy discontinuity at non-convex regions of the boundary. This is an inherited limitation of the OGC formulation that is not addressed by this method.}

\bibliographystyle{ACM-Reference-Format}
\bibliography{my_references,references}
\appendix

\section{Proofs}

\subsection{Proof of Theorem 1: Vertex-Level Equivalence}
\label{proof:vertex_level_eq}

\begin{proof}
The exclusive regions are intersections of half-spaces, hence convex. Triangles and edges are also convex.

\textit{If:} If all faces stay in their regions, then each vertex (as part of its incident faces) must satisfy the constraints.

\textit{Only if:} If all vertices satisfy the constraints, then every point on a triangle or edge---being a convex combination of its vertices---also lies in the convex exclusive region.
\end{proof}

\subsection{Proof of Theorem 2: Rigid Face Plane Crossing}
\label{proof:rigid_face_plane}

\begin{proof}
\textit{If:} If a vertex crosses the plane, the face containing it clearly intersects the plane.

\textit{Only if:} Rigid motion preserves barycentric coordinates, so any point on the face remains a fixed convex combination of its vertices. The signed distance of that point to the plane is therefore a convex combination of the vertices' signed distances. If all vertices stay on the positive side, so does every point on the face.
\end{proof}

\section{Divide and Project with Dykstra's Projection}
\label{sec:dykstra}
Dykstra's Projection is a well-known dual coordinate ascent algorithm for finding the projection of a point onto the non-empty intersection of multiple convex sets. It is known to work well for simple sets and naturally extends into projections in Hilbert space. 
We provide one derivation of the algorithm here. For additional details, see \citet{boyle1986method}.

We would like to solve a minimization of the form
\begin{align}
    \min_{x\in C_i} \frac{1}{2}|x-r|^2_\MM
\end{align}
where $C_i$ are closed convex sets and $\MM$ is some symmetric positive-definite (SPD) metric. We begin by converting the constraint into their corresponding indicator functions
\begin{align}
    c_i(x) = \begin{cases}
0 & \text{if } x \in C_i,\\
\infty  & \text{otherwise}
\end{cases}
\end{align}
This converts the optimization into the unconstrained
\begin{align}
    \min_x \frac{1}{2}|x-r|^2_\MM+\sum c_i(x).
\end{align}
As written, we currently have to check $x$ against every set. To allow us to view each set independently, we follow ADMM and projective dynamics to introduce the auxiliary variables $x_i$ with the form
\begin{align}
    \min_x \frac{1}{2}|x-r|^2_\MM+\sum c_i(x_i)\\
    s.t.\;x-x_i=0\;\forall i.
\end{align}
The Lagrangian of this constrained optimization with the dual variables $\lambda_i$ then becomes 
\begin{align}
    \mathcal{L} = \frac{1}{2}|x-r|^2_\MM+\sum \left(c_i(x_i)+\lambda_i^T (x-x_i)\right).
\end{align}

\subsection{Dual formulation}

To tackle this, we seek the dual of our Lagrangian, $g(\lambda)$.
\begin{align}
    g(\lambda) = \inf_{x,x_i} \mathcal{L}
\end{align}
The solution to our original optimization can then be found at $\max_{\lambda} g(\lambda)$. 

Since we'd like to iterate over each projection one at a time, we invoke block coordinate descent methods to view each degree of freedom in isolation while keeping all others fixed. 

\subsubsection{Dual Over $x$} For the dual over our primal $x$, the dual becomes
\begin{align}
\label{eq:dual_x}
    g(\lambda) = \inf_{x} \frac{1}{2}|x-r|^2_\MM + \sum \lambda_i^T (x-x_i).
\end{align}
Since this is a simple quadratic, we find the solution
\begin{align}
    s &= \sum \lambda_i,\\
    \label{eq:x_opt}
    x &=r-\MM^{-1}s
\end{align}
The exact value of $g(x)$ is unimportant. For now, we only need to know that $x=r-\MM^{-1}s$. 

\subsubsection{Dual Over $x_i$}

For the dual over each particular $x_i$, we can substitute in $x=r-\MM^{-1}s$ to find
\begin{align}
    g_i(\lambda) &= \inf_{x_i} \frac{1}{2}|x-r|^2_\MM+\sum \left(c_i(x_i)+\lambda_i^T (x-x_i)\right)\\
    &= \inf_{x_i} \frac{1}{2}|\MM^{-1}s|^2_\MM+\sum \left(c_i(x_i)+\lambda_i^T (r-\MM^{-1}s-x_i)\right)\\
\label{eq:dual_x_i}
    &= \inf_{x_i} -\frac{1}{2}|s|^2_{\MM^{-1}}+\sum \left(c_i(x_i)+\lambda_i^T (r-x_i)\right)
\end{align}

\subsection{Dual Ascent on $\lambda_i$}
As mentioned earlier, we would like to view this as a block coordinate descent (ascent) algorithm. Starting with the initial guess $x^0=r$, we view each set $C_i$ and its corresponding dual variable $\lambda_i$ in isolation. We then seek to maximize \autoref{eq:dual_x_i} at every iteration. As all other $\lambda$s are viewed as fixed, it is convenient to express this $s$ as $s_i+\lambda_i$, where $s_i$ is the sum of all $\lambda$ excluding $\lambda_i$. From \autoref{eq:x_opt}, we can compute this as
\begin{align}
    s_i=\MM(r-x)-\lambda^k_i
\end{align}
Dropping all constant terms in $g_i$, $c_i(x_i)- \lambda_i^T x_i$ can then be replaced by the support function $\sigma_{C_i}$.
\begin{align}
    \sigma_{C_i}(d) &= \max_{x \in C_i} d^T x\\
    g_i(\lambda_i) &= \inf_{x_i} -\frac{1}{2}|s_i+\lambda_i|^2_{\MM^{-1}} +\lambda_i^T r + c_i(x_i)-\lambda_i^T x_i + const\\
     &= -\frac{1}{2}|s_i+\lambda_i|^2_{\MM^{-1}}  - \sigma_{C_i}(\lambda_i) +r^T\lambda_i+ const
\end{align}

We seek the maximizer of the dual. Here, we complete the square by adding the constant $\frac{1}{2}|\MM r-s_i|^2_{\MM^{-1}}$
\begin{align}
    \lambda^{k+1}_i& =arg \max_{\lambda_i} g_i(\lambda_i)\nonumber\\
    &= arg \min_{\lambda_i} \frac{1}{2}|s_i+\lambda_i|^2_{\MM^{-1}} + \sigma_{C_i}(\lambda_i) - r^T\lambda_i\nonumber\\
    &= arg \min_{\lambda_i} \frac{1}{2}|\lambda_i|^2_{\MM^{-1}}+(s_i^T \MM^{-1}- r^T)\lambda_i + \sigma_{C_i}(\lambda_i) \nonumber\\
    &= arg \min_{\lambda_i} \frac{1}{2}|\lambda_i|^2_{\MM^{-1}}+(s_i^T- r^T\MM)\MM^{-1}\lambda_i + \sigma_{C_i}(\lambda_i)\nonumber\\
    &= arg \min_{\lambda_i} \frac{1}{2}|\lambda_i|^2_{\MM^{-1}}+(s_i - \MM r)^T\MM^{-1}\lambda_i +\frac{1}{2}|\MM r-s_i|^2_{\MM^{-1}}+ \sigma_{C_i}(\lambda_i)\nonumber\\
    &= arg \min_{\lambda_i} \frac{1}{2}|\MM r - s_i - \lambda_i|^2_{\MM^{-1}} + \sigma_{C_i}(\lambda_i)
\end{align}
This is now suspiciously similar to a projection. Let us define $y=r - \MM^{-1}s_i$ 
\begin{align}
    \label{eq:y_vec}
    y&=r - \MM^{-1}s_i\\
    \lambda^{k+1}_i &= arg \min_{\lambda_i} \frac{1}{2}|\MM y - \lambda_i|^2_{\MM^{-1}} + \sigma_{C_i}(\lambda_i)   
\end{align}

\subsection{The Projection Operator}
To see how this is related to projections, we only need to write the down projection operator explicitly and compare their dual.
\begin{align}
    P^\MM_{C_i}(r) = arg \min_x &\frac{1}{2}|x-r|^2_{\MM} + c_i(\hat{x})\nonumber\\
    & s.t.\;x=\hat{x}\nonumber
\end{align}
The rest of the process is very similar to everything we've done before. We use the dual over $x$ to find the relation between $x$ and $\lambda$
\begin{align}
    \mathcal{L} &= \frac{1}{2}|x-r|^2_{\MM} + c_i(\hat{x}) + \lambda^T(x-\hat{x})\nonumber\\
    g_x(\lambda) &= \inf_x \frac{1}{2}|x-r|^2_{\MM} + \lambda^T x + const\nonumber\\
    x &=r-\MM^{-1}\lambda\nonumber
\end{align}
We plug this back into the dual for $\hat{x}$.
\begin{align}
    g_{\hat{x}}(\lambda) &= \inf_{\hat{x}} \frac{1}{2}|r-\MM^{-1}\lambda-r|^2_{\MM} + c_i(\hat{x}) + \lambda^T(r-\MM^{-1}\lambda-\hat{x})\nonumber\\
    &=\inf_{\hat{x}}-\frac{1}{2}|\lambda|^2_{\MM^{-1}} + \lambda^T r + c_i(\hat{x}) - \lambda^T\hat{x}\nonumber\\
    &= -\frac{1}{2}|\lambda|^2_{\MM^{-1}} - \sigma_{C_i}(\lambda) + \lambda^T r \nonumber
\end{align}
Then we maximize the dual and rearrange the terms into perfect squares by adding the constant $\frac{1}{2}|\MM r|^2_{\MM^{-1}}$. 
\begin{align}
    \lambda &= arg \min_{\lambda} \frac{1}{2}|\lambda|^2_{\MM^{-1}} + \sigma_{C_i}(\lambda) - (\MM r)^T\MM^{-1}\lambda + \frac{1}{2}|\MM r|^2_{\MM^{-1}}\nonumber\\
    &= arg \min_{\lambda} \frac{1}{2}|\MM r-\lambda|^2_{\MM^{-1}} + \sigma_{C_i}(\lambda)\nonumber
\end{align}
This is the exact form we obtained earlier for $\lambda_i^{k+1}$. Although it seems like we've just retraced our steps, we've now proven that each block coordinate descent step on our original $C_i$ is actually just the projection onto that set. The only thing left is then to find a formula for $\lambda$. 

Since we already know that $x = r-\MM^{-1}\lambda$ and $x=P^\MM_{C_i}(r)$, it follows naturally that $\lambda=\MM(r-P^\MM_{C_i}(r))$. Returning back to our original problem, we find
\begin{align}
    \lambda^{k+1}_i=\MM(y-P^\MM_{C_i}(y))
\end{align}

\subsection{The Dykstra Projection Algorithm}
We now only have to put the pieces together. Since $\lambda$ is always used in conjunction with $\MM$, we store $\MM^{-1}\lambda_i$ instead. Using \autoref{eq:x_opt}, Dykstra computes $y$ as 
\begin{align}
    y^k_i &= r - \MM^{-1}s_i\nonumber\\
    &= (r - \MM^{-1}s_i - \MM^{-1}\lambda^k_i) + \MM^{-1}\lambda^k_i\nonumber\\
    &= (r - \MM^{-1}s) + \MM^{-1}\lambda^k_i\nonumber\\
    &= x^k + \MM^{-1}\lambda^k_i.
\end{align}
$x$ is updated as 
\begin{align}
    x^{k+1} = P^\MM_{C_i}(y^k_i).
\end{align}
Then, $\lambda$ is updated as
\begin{align}
    \MM^{-1}\lambda_i^{k+1} = y^k_i - x^{k+1}.
\end{align}
This gives the general form of Dykstra using SPD metrics. You would notice that since we store everything as $\MM^{-1}\lambda_i$, this actually does not change the original algorithm at all. We only need to take care to perform the projection with $\MM$ in mind. For half-spaces, this is a trivial operation which can be written as a least norm problem. However, no special handling of the iterates are needed.

\end{document}